\documentclass[11pt, a4paper]{article}
\usepackage{url}
\usepackage{a4wide}
\usepackage{graphicx}
\usepackage{natbib}
\usepackage{enumitem}
\usepackage{hyperref}
\usepackage{setspace}
\usepackage{amsmath, amsthm, amssymb, amsfonts}
\usepackage{bm}
\usepackage{multirow}
\usepackage{color}     
\usepackage[ruled, vlined]{algorithm2e}
\usepackage{afterpage}
\usepackage{float}
\usepackage{subcaption}
\usepackage{booktabs}

\allowdisplaybreaks

\DeclareMathAlphabet\mathbfcal{OMS}{cmsy}{b}{n}

\newcommand{\mbf}{\mathbf}
\newcommand{\mc}{\mathcal}

\newcommand{\vep}{\varepsilon}

\renewcommand{\l}{\left}
\renewcommand{\r}{\right}
\newcommand{\ls}{\le}
\newcommand{\gs}{\ge}
\newcommand{\dto}{\overset{\mathcal{D}}{\to}}
\newcommand{\pto}{\overset{\mathcal{P}}{\to}}
\newcommand{\deq}{\overset{\mathcal{D}}{=}}

\def\wh{\widehat}
\def\wt{\widetilde}

\newcommand{\E}[0]{\mathsf{E}}
\newcommand{\Var}[0]{\mathsf{var}}

\newcommand{\p}{\mathsf{P}}
\newcommand{\R}{\mathbb{R}}
\newcommand{\Z}{\mathbb{Z}}

\newcommand{\nn}{\nonumber}

\newcommand{\cp}{\theta}  
\newcommand{\Cp}{\Theta}  

\theoremstyle{definition}
\newtheorem{thm}{Theorem}[section]
\theoremstyle{definition}

\theoremstyle{definition}
\newtheorem{lem}[thm]{Lemma}
\theoremstyle{definition}

\theoremstyle{definition}
\newtheorem{assum}{Assumption}[section]
\theoremstyle{remark}
\newtheorem{rem}{Remark}[section]
\theoremstyle{definition}

\theoremstyle{definition}

\title{Bootstrap confidence intervals for multiple change points based on moving sum procedures}
\author{Haeran Cho$^1$ and Claudia Kirch$^2$}

\begin{document}

\maketitle


\begin{abstract}
The problem of quantifying uncertainty about the locations of multiple change points by means of confidence intervals is addressed. The asymptotic distribution of the change point estimators obtained as the local maximisers of moving sum statistics is derived, where the limit distributions differ depending on whether the corresponding size of changes is local, i.e.\ tends to zero as the sample size increases, or fixed. A bootstrap procedure for confidence interval generation is proposed which adapts to the unknown magnitude of changes and guarantees asymptotic validity both for local and fixed changes. Simulation studies show good performance of the proposed bootstrap procedure, and some discussions about how it can be extended to serially dependent errors is provided.
\end{abstract}

\footnotetext[1]{Institute for Statistical Science, School of Mathematics, University of Bristol, UK.
Email: \url{haeran.cho@bristol.ac.uk}.
Supported by Leverhulme Trust Research Project Grant RPG-2019-390.}

\footnotetext[2]{Department of Mathematics, Otto-von-Guericke University; Center for Behavioral Brain Sciences (CBBS); Magdeburg, Germany.
Email: \url{claudia.kirch@ovgu.de}.
Supported by Deutsche Forschungsgemeinschaft - 314838170, GRK 2297 MathCoRe.} 


\section{Introduction}


Multiple change point analysis, a.k.a. data segmentation, is an actively  researched area
with a wide of range of applications in natural and social sciences, medicine, engineering and finance.
The canonical data segmentation problem,
where the aim is to detect and locate multiple change points 
in the mean of univariate time series, 
has received great attention in the past few decades
and there exist a variety of methodologies that are computationally fast
and achieve consistency in estimating the total number and the locations of multiple change points;
see \cite{cho2020data} for an overview of the literature
and discussions on how methods proposed for the canonical data segmentation problem
offer an important stepping stone for addressing more complex change point problems,
such as detecting changes in variance, time series segmentation under parametric models
and robust change point analysis.

By contrast, the literature on inference for multiple change points is relatively scarce.
Asymptotic \citep{eichinger2018} or approximate \citep{fang2020} 
distributions of suitable test statistics
have been derived under the null hypothesis of no change point,
which enable quantifying uncertainty about the number of change points.
A class of multiscale change point segmentation procedures 
aims at controlling the family-wise error rate \citep{frick2014}
or the false discovery rate \citep{li2016} of detecting too many change points. 
There also exist post-selection inference methods 
which test for a change at estimated change point locations
conditional on their estimation procedure,
see e.g.\ \citet{hyun2018post} and \citet{jewell2019testing}.
The Bayesian framework lends itself naturally to change point inference, 
see \cite{fearnhead2006} and \cite{nam2012quantifying}.

Another type of uncertainty stems from the localisation of change points.
The optimal rate of localisation in change point problems is $O_P(1)$ at best (see e.g.\ \cite{fromont2020}),
i.e.\ change point location estimators are not consistent in the usual sense, 
which makes the problem of inferring uncertainty about change point locations particularly relevant and important.
The simultaneous multiscale change point estimator (SMUCE) proposed in \cite{frick2014} 
provides a confidence set for all candidate signals
from which confidence intervals around the change points can be obtained.
Using the inverse relation between confidence intervals and hypothesis tests,
\cite{fang2020} detail how confidence regions can be generated
from an approximation of the limit distribution of the test statistic under the null hypothesis.
In all of the above, the error distributions are assumed to belong to an exponential family
such as Gaussian, or other light-tailed ones.
The narrowest significance pursuit (NSP) uses a multi-resolution sup-norm
to identify regions containing at least one change point in the mean \citep{fryzlewicz2020} 
or the median \citep{fryzlewicz2021robust} at a prescribed confidence level,
and allows for heavy-tailed errors.

\citet{meier2021mosum} outlines the bootstrap construction of confidence intervals
around the change points based on the moving sum (MOSUM) procedure
proposed in \cite{eichinger2018}.
In this paper, we show the theoretical validity of the bootstrap procedure,
i.e.\ that the proposed bootstrap $100(1 - \alpha)$\%-confidence intervals
asymptotically attain the coverage probability of $1 - \alpha$ for given $\alpha \in (0, 1)$
(see~\eqref{eq:boot:valid} below),
and demonstrate its good performance via numerical experiments.
Our theoretical contributions build upon the results derived in
\cite{antoch1995change} and \cite{antoch1999estimators} for the case of at most a single change,
and accommodate both situations where the changes are local (i.e.\ tend to zero with the sample size) 
and when they are fixed while requiring only that
the errors have more than two finite moments. 

The rest of the paper is organised as follows.
Section~\ref{sec:bootstrap} motivates the use of a bootstrap procedure for confidence interval generation
and proposes the bootstrap construction of 
{\it pointwise} and {\it uniform} confidence intervals.
Section~\ref{sec:theor} provides results on 
the asymptotic distributions of change point estimators 
obtained from the original and the bootstrap data,
based on which we establish the validity of bootstrap confidence intervals.
In Section~\ref{sec:ext}, we discuss 
the use of the proposed bootstrap procedure with asymmetric bandwidths
and its extension to the case of serially dependent errors.
Section~\ref{sec:sim} shows the good performance 
of the proposed methodology on simulated datasets in comparison with existing methods
and applies it to Hadley Centre central England temperature data,
and Section~\ref{sec:conc} concludes the paper.
The implementation of the proposed bootstrap methodology is available in 
the R package \texttt{mosum} \citep{mosum} as {\tt confint} method.

\section{Bootstrap confidence intervals for change points}
\label{sec:bootstrap}

In this paper, we consider the following model with multiple change points
\begin{align}
X_t &= f_t + \vep_t = f_0 + \sum_{j = 1}^{q_n} d_{j, n} \cdot \mathbb{I}_{\{t > \cp_{j, n}\}} + \vep_t
= \sum_{j = 0}^{q_n} \mu_{j, n} \cdot \mathbb{I}_{\{\cp_{j, n} < t \le \cp_{j + 1, n}\}} + \vep_t, \label{eq:model} 
\end{align}
where $\cp_j = \cp_{j, n}$ denote the $q_n$ change points 
(with $\cp_0 = 0$ and $\cp_{q_n + 1} = n$)
at which the mean of $X_t$ undergoes changes of (signed) size $d_j = d_{j, n}$.
We denote by $\delta_j = \delta_{j, n} = \min(\cp_j - \cp_{j - 1}, \cp_{j + 1} - \cp_j)$ 
the minimum distance from $\cp_j$ to its neighbouring change points,
and by $\Cp = \Cp_n = \{\cp_1, \ldots, \cp_{q_n}\}$ the set of change points.
Throughout the paper, we focus on the case of i.i.d.\ errors $\{\vep_t\}$ satisfying
\begin{align}
\label{eq:errors}
\E(\vep_1) = 0, \quad 0 < \sigma^2 = \Var(\vep_1) < \infty \quad
\text{and} \quad \E(\vert \vep_1 \vert^{\nu}) < \infty 
\end{align}
for some $\nu > 2$,
and provide some discussions on the case of dependent errors in Section~\ref{sec:dependent}.

Under~\eqref{eq:model}, several methods exist that consistently estimate $q_n$, the number of change points.
On the other hand, the known minimax optimal rate for the estimation of 
change point locations is $O_P(1)$ at best (see e.g.\ \cite{fromont2020}),
i.e.\ the location estimation error does not tend to zero as $n \to \infty$.
This makes the task of uncertainty quantification about change point locations
by deriving confidence intervals (CI) around $\cp_j$, highly important.

In Section~\ref{sec:motiv}, we motivate the use of a bootstrap procedure
for the construction of CIs about change point locations,
with a review of its application to the simple case of at-most-one-change (AMOC), 
i.e.\ when  $q_n \le 1$.
Then Section~\ref{sec:mosum:main} describes a procedure based on moving sums
for multiple change point detection under~\eqref{eq:model},
and Section~\ref{sec:method} presents the proposed bootstrap methodology
whose validity is established later in Section~\ref{sec:theor}.

\subsection{Motivation}
\label{sec:motiv}

In the AMOC setting,
classical test statistics such as those based on the CUSUM statistic
\begin{align}
\mc C_{k, n}(X) = \sqrt{\frac{k (n - k)}{n}} \l(\bar{X}_{0:k} - \bar{X}_{k:n}\r)
\quad \text{with} \quad \bar{X}_{s:e} = \frac{1}{e - s} \sum_{t = s + 1}^e X_t,
\nn 
\end{align}
are used to test the null hypothesis $H_0: \, q_n = 0$ (no change point) 
against $H_1: \, q_n = 1$ (a single change point).
When $H_0$ is rejected,
the CUSUM statistic can directly be used to locate $\cp \equiv \cp_1$ by its estimator 
$\wh\cp = \arg\max_{0 < k < n} \vert \mc C_{k, n}(X) \vert$.
The asymptotic distribution of $\wh\cp$ depends on unknown quantities,
most importantly, on the magnitude of the change.
For a local change with $d_1 = d_{1, n} \to 0$ as $n \to \infty$,
the limit is distribution-free \citep{antoch1995change}
whereas for a fixed change, 
the limit depends on the unknown error distribution \citep{antoch1999estimators}.
Consequently, the asymptotic distribution is of little practical use 
for constructing a CI about $\cp$
due to the difficulty involved in estimating such quantities.

The bootstrap construction of a CI utilises the difference between the bootstrap estimator, 
say $\wh\cp^* = \arg\max_{0 < k < n} \vert \mc C_{k, n}(X^*) \vert$
maximising the CUSUM statistics computed on a bootstrap sample $\{X_t^*\}_{t = 1}^n$,
and the original estimator $\wh\cp$,
as a proxy for the difference between $\wh\cp$ and the true change point~$\cp$.
Bootstrap CIs in the AMOC setting have been proposed by
\cite{antoch1995change} (accompanied by rigorous proofs for the case of local changes)
and \cite{antoch1999estimators} (their theoretical results cover fixed changes but are given without rigorous proofs).
While the asymptotic distributions (and the corresponding proofs) are different in the two regimes 
determined by the magnitude of $d_1$, 
the same bootstrap procedure can correctly mimic these asymptotic distributions
regardless of whether the change is local or fixed,
{\it without} requiring the knowledge of which regime the problem belongs to
or that of the error distribution.
As a result, the corresponding bootstrap CI is asymptotically correct in both regimes.

This motivates the use of a bootstrap CI for quantifying uncertainty  
about the change point location rather than its asymptotic counterpart.
An additional testing does not alter the theoretical validity of the bootstrap CI
since under $H_1$, the test rejects $H_0$ with asymptotic power one under weak assumptions
even when the nominal level of the test converges slowly to~$0$. 
In such a case, the chance of any false positive also tends to zero asymptotically 
and, conditional on this asymptotic one-set, the bootstrap CI is either empty (under $H_0$)  
or remains to be asymptotically honest (under $H_1$).

\subsection{Multiple change point estimation based on moving sums}
\label{sec:mosum:main}

An obvious difficulty when departing from the AMOC situation 
is that we do not know the number of change points a priori.
For the multiple change point detection problem under~\eqref{eq:model},
\cite{eichinger2018} propose a moving sum (MOSUM) procedure 
that makes use of the MOSUM statistic which, for a given bandwidth $G = G_n$, is defined as
\begin{align}
\label{eq:mosum:symm}
T_{k, n}(G; X) = \sqrt{\frac{G}{2}} \l(\bar{X}_{k - G, k} - \bar{X}_{k, k + G}\r)
\quad \text{for} \quad 
G \le k \le n - G.
\end{align}
The statistic $T_{k, n}(G; X)$ takes a large value in modulus around true change points
while taking a small value outside their $G$-environments.
Therefore, the MOSUM procedure
achieves simultaneous detection and localisation of multiple change points by
(i) performing a model selection step closely related to change point testing in the AMOC setting,
using the asymptotic distribution of 
$\max_{G \le k \le n - G} \vert T_{k, n}(G; X) \vert$ under $H_0$
to determine `significant' local maxima of the MOSUM statistics, and
(ii) identifying the corresponding local maximisers of $\vert T_{k, n}(G; X) \vert$
as change point location estimators.
Combining the output from the MOSUM procedure applied with a range of bandwidths,
it is feasible to perform change point analysis at multiple scales;
see Appendix~\ref{sec:mosum} for further details of the MOSUM procedure
and its multiscale extension as proposed by \cite{cho2019two}.

The model selection step in~(i) is performed in such a way
that the local maximisers in~(ii) are asymptotically equivalent to the following \textit{oracle} estimators:
\begin{align}
\label{eq:cp:tilde}
\wt\cp_j = \wt\cp_{j, n} = {\arg\max}_{\cp_j - G_j < k \le \cp_j + G_j} \vert T_{k, n}(G_j) \vert
\quad \text{for} \quad j = 1, \ldots, q_n.
\end{align}
That is, each $\wt\cp_j$ is the local maximiser of the MOSUM statistic in the neighbourhood of $\cp_j$
that is determined by a suitable bandwidth $G_j$.
Here, `oracle' refers to the fact that such estimators are clearly not accessible in practice
due to knowing neither the total number nor the locations of the change points. 
We assume the following on $G_j, \, j = 1, \ldots, q_n$:
\begin{align}
\label{eq:bandwidth}
G_j = G_{j, n} \to \infty \text{ as } n \to \infty \quad \text{and} \quad
2G_j < \delta_j.
\end{align}

\cite{eichinger2018} and \cite{cho2019two} show that 
MOSUM-based procedures are consistent
both in estimating the number of change points as well as their locations
and derive the localisation rate (i.e.\ how close the estimators are to the true change points asymptotically) 
under mild assumptions on $\{\vep_t\}$, see Appendix~\ref{sec:mosum}.
An important step in the proof of such a consistency result 
is to show that the change point estimators generated by such procedures,
say $\wh\cp_j, \, 1 \le j \le \wh q_n$, coincide 
with the oracle estimators $\wt\cp_j, \, 1 \le j \le q_n$, on an asymptotic one-set.
We formalise this key observation as the following meta-assumption:
\begin{assum}
\label{assum_meta_est}
\begin{enumerate}[label = (\alph*)]
\item \label{assum_meta_est_one} 
For a given $j \in \{1, \ldots, q_n\}$, the estimator $\wh\cp_j = \wh\cp_{j, n}$ of $\cp_j$ satisfies
\begin{align*}
\p\l( \mc A_j \r) \to 1 \text{ as } n \to \infty, \quad \text{where} \quad
\mc A_j = \mc A_{j, n} = \l\{ \wh\cp_{j, n} = \wt\cp_{j, n} \r\}.
\end{align*}
\item \label{assum_meta_est_two} The set of change point estimators 
$\wh\Cp = \wh\Cp_n = \{\wh\cp_j, \, 1 \le j \le \wh q_n: \, \wh\cp_1 < \ldots < \wh\cp_{\wh q_n}\}$ satisfies
\begin{align*}
\p\l( \mc A \r) \to 1 \text{ as } n \to \infty, \quad \text{where} \quad
\mc A = \mc A_n = \l\{\wh q_n = q_n \text{ and } \wh\cp_{j, n} = \wt\cp_{j, n}, \,  j = 1, \ldots, q_n\r\}.
\end{align*}
\end{enumerate}
\end{assum}
The equivalence of the oracle estimators and the accessible estimators 
obtained with a model selection step
(as detailed in~\eqref{eq:hat:tilde:one}--\eqref{eq:hat:tilde:two}),
is crucial for the bootstrap CIs introduced in Section~\ref{sec:method} below,
since it allows us to construct bootstrap estimators mimicking the oracle estimators 
without having to perform any model selection step in the bootstrap world.

\subsection{Bootstrap methodology}
\label{sec:method}

In this section, we describe the construction of bootstrap CIs for multiple change points,
which closely resembles the bootstrap methodology introduced 
by \cite{antoch1995change} in the AMOC setting.

MOSUM-based change point detection procedures already incorporate
some uncertainty quantification for the number of change points and, even their locations,
since $\vert T_{k, n}(G; X) \vert$ exceeding a critical value indicates
that $\{k - G + 1, \ldots, k + G\}$ contains a true change point with high probability.
However, our aim here is to construct (asymptotically) honest CIs that
quantify the uncertainty about the locations of the change points at a prescribed level,
with their widths narrower than those 
given by the bandwidths involved in detecting the corresponding change points.

In what follows, we assume that a set of change point estimators,
$\wh\Cp = \{\wh\cp_j, \, 1 \le j \le \wh q_n\}$, is given
with $\wh q_n$ denoting the estimator of the number of change points,
and we adopt the notational convention that $\wh\cp_0 = 0$ and $\wh\cp_{\wh q_n + 1} = n$.
We suppose that each estimator $\wh\cp_j$ is detected
with a bandwidth $G_j$ fulfilling~\eqref{eq:bandwidth},
which in turn is used in the construction of bootstrap CIs as described below.

\begin{enumerate}[label = \textbf{\textit{Step~\arabic*}}:, itemindent = 25pt]
\item Generate a bootstrap sample
$\{X_t^*, \, 1 \le t \le n\}$ 
by randomly drawing $\{X_t^*, \, \wh\cp_j < t \le \wh\cp_{j + 1}\}$ with replacement 
from $\{X_t, \, \wh\cp_j < t \le \wh\cp_{j + 1}\}$ for $j = 0, \ldots, \wh q_n$.

\item Compute the MOSUM statistics $T_{k, n}(G_j; X^*)$ as in~\eqref{eq:mosum:symm} 
with $\{X^*_t\}$ in place of $\{X_t\}$, and locate
\begin{align}
\label{eq:boot:max}
\wt\cp_j^* = {\arg\max}_{\wh\cp_j - H_j < k \le \wh{\cp}_j + H_j} \vert T_{k, n}(G_j; X^*) \vert
\end{align}
for each $j = 1, \ldots, q_n$, 
 where $H_j = \min(G_j, 2\wh\delta_j/3)$
with $\wh\delta_j = \min(\wh\cp_j - \wh\cp_{j - 1}, \wh\cp_{j + 1} - \wh\cp_j)$.

\item For a given bootstrap sample size $B$,
repeat Steps~1--2 $B$ times and record $\wt\cp_j^{*(b)}, \, j = 1, \ldots, \wh q_n$,
the local maximisers obtained as in~\eqref{eq:boot:max}, for $b = 1, \ldots, B$.
\end{enumerate}

\begin{rem}
\label{rem_argmax_boot}
\begin{enumerate}[label = (\alph*)]
\item \label{rem_argmax_boot:one} In our theoretical analysis, we assume that each $G_j$ satisfies~\eqref{eq:bandwidth}
in addition to $\wh{\cp}_j-\cp_j=o_P(\delta_j)$ (see Assumption~\ref{assum_precision_est} below) 
such that $2(\wh{\cp}_j-\wh\cp_{j-1})/3 \ge (2/3 + o_P(1))\delta_j \ge (4/3 + o_P(1))G_j$ i.e.\
$H_j = G_j$ in~\eqref{eq:boot:max} with probability converging to one.
Consequently, the bootstrap estimator $\wt\cp_j^*$ mimics the definition of the oracle estimator $\wt\cp_j$ in~\eqref{eq:cp:tilde},
with $\wh\cp_j$ serving as a change point in the bootstrap sample.

\item In practice, the choice of $G_j$ fulfilling~\eqref{eq:bandwidth} is not available 
and each change point estimator is associated with 
either a pre-determined bandwidth (as in the case in \cite{eichinger2018}),
or a bandwidth chosen from a range of bandwidths
by a multiscale MOSUM procedure (as is the case in \cite{cho2019two}).
Therefore, we cannot guarantee that 
adjacent estimators, say $\wh\cp_{j - 1}$ and $\wh\cp_{j + 1}$,
are strictly outside the interval $(\wh\cp_j - G_j, \wh\cp_j + G_j]$.
For example, if $\wh\cp_{j - 1}$ falls into this interval, 
two estimators $\wh\cp_{j - 1}$ and $\wh\cp_j$ compete against each other 
to be the local maximiser of $\vert T_{k, n}(G_j; X^*) \vert$ over this interval.
When more than $100\,\alpha$\% of the bootstrap realisations happen to
yield local maxima near $\wh\cp_{j - 1}$, the radius of the bootstrap CI is as wide as $G_j$
even if the change at $t = \cp_j$ (as well as $t = \wh\cp_j$ for the bootstrap realisations) 
is highly pronounced to be detectable.
To prevent such events, we propose the slight modification 
involving $H_j$ as in~\eqref{eq:boot:max}
which performs well in practice as shown in Section~\ref{sec:sim}.
\end{enumerate}
\end{rem}

At a given level $\alpha \in (0, 1)$, 
a {\it pointwise} $100(1 - \alpha)\%$ bootstrap CI for each $\cp_j$ is constructed as
\begin{align}
\mc C^{\text{pw}}_j(\alpha) &= \l[ \wh\cp_j - Q_j(\alpha), \wh\cp_j + Q_j(\alpha) \r] \quad \text{with}
\nn \\
Q_j(\alpha) &= \inf\l\{c: \, 
\frac{1}{B} \sum_{b = 1}^B \mathbb I \l( \l\vert \wt\cp_j^{*(b)} - \wh\cp_j \r\vert \le c \r) \ge 1 - \alpha \r\}.
\label{eq:pw:ci}
\end{align}
A {\it uniform} bootstrap CI, 
which provides a guarantee for the simultaneous coverage of $\cp_j, \, j = 1, \ldots, q_n$
(as shown later in Section~\ref{sec:theor}), is constructed as follows:
Estimating the (signed) size of change 
as $\wh d_j = \bar{X}_{\wh\cp_j, \wh\cp_{j + 1}} - \bar{X}_{\wh\cp_{j - 1}, \wh\cp_j}$, 
and the (local) variance as
\begin{align*}
\wh\sigma_j^2 &= \frac{1}{\wh\cp_{j + 1} - \wh\cp_{j - 1} - 2} 
\l( \sum_{t = \wh{\cp}_{j - 1} + 1}^{\wh{\cp}_j} (X_t - \bar{X}_{\wh{\cp}_{j - 1}, \wh{\cp}_j})^2 +
\sum_{t = \wh{\cp}_j + 1}^{\wh{\cp}_{j + 1}} (X_t - \bar{X}_{\wh{\cp}_j, \wh{\cp}_{j + 1}})^2 \r)
\end{align*}
for $j = 1, \ldots, \wh q_n$, 
%
a uniform $100(1 - \alpha)\%$-CI is given by
\begin{align}
\mc C^{\text{unif}}_j(\alpha) &= \l[ \wh\cp_j - \wh d_j^{-2} \wh\sigma_j^2 Q(\alpha), 
\wh\cp_j + \wh d_j^{-2} \wh\sigma_j^2 Q(\alpha) \r] \quad \text{with}
\nn \\
Q(\alpha) &= \inf\l\{c: \, \frac{1}{B} \sum_{b = 1}^B 
\mathbb I \l( 
\max_{1 \le j \le \wh q_n} \frac{\wh d_j^2}{\wh\sigma_j^2}
\l\vert \wt\cp_j^{*(b)} - \wh\cp_j \r\vert \le c \r) \ge 1 - \alpha \r\}.
\label{eq:unif:ci}
\end{align}
The quantities $Q_j(\alpha)$ (resp. $Q(\alpha)$)
are empirical versions of the quantiles of the conditional distribution of $\wt\cp_j^* - \wh\cp_j$ 
(as shown in Theorem~\ref{thm:bootstrap} below) obtained by Monte Carlo simulations 
and converge to the true quantiles as $B \to \infty$,
such that the respective bootstrap CIs are asymptotically honest in the sense 
made precise in~\eqref{eq:boot:valid} below.
Unlike the pointwise bootstrap CIs,
uniform CIs involve the estimation of the signal-to-noise ratio $d_j/\sigma_j$ 
such that $\wt\cp_j^{*} - \wh\cp_j$ are treated on an equal footing across $j = 1, \ldots, \wh q_n$.
Lemma~\ref{lemma_variance_boot} in Appendix~\ref{sec:proofs} shows that
both $\wh{d}_j$ and $\wh\sigma_j^2$ are consistent and in particular,
$\wh{d}_j$ is consistent not only when $d_j$ is fixed but also when $d_j \to 0$
in the sense that ${\wh d}_j/d_j \pto 1$.

\section{Theoretical validity of bootstrap confidence intervals}
\label{sec:theor}

As discussed in Section~\ref{sec:motiv}, in the AMOC setting,
bootstrap CIs have been shown to adapt to whether the (unknown) size of change is local or fixed
without requiring the knowledge of the error distribution,
which makes their use more practical than the asymptotic CIs. 
In this section, we show that this is also the case 
in the presence of multiple change points
with the bootstrap procedure introduced in Section~\ref{sec:method}.

In the AMOC setting, 
the proof of the validity of bootstrap CIs proceeds in two steps: 
First, the asymptotic distribution of (scaled) difference $\wh\cp - \cp$ is established,
and then it is shown that (analogously scaled) $\wh\cp^* - \wh\cp$
has the same limit distribution conditional on the observations. 
When the limit distribution is continuous,
quantiles of both differences converge to a true asymptotic quantile
such that asymptotic honesty of the bootstrap CIs follows
irrespective of the regime determined by the size of the change.

In the multiple change point problem, the estimators $\wh\cp_j$ typically involve a model selection step
while the construction of bootstrap estimators $\wt\cp^*_j$ only mimics the uncertainty 
stemming from random fluctuations in local maximisation of 
$\vert T_{k, n}(G_j; X^*) \vert$ as in~\eqref{eq:boot:max}.
Indeed, our bootstrap procedure is designed to 
mimic the asymptotic distribution of the oracle estimator in \eqref{eq:cp:tilde}.
Nonetheless, the accessible estimators $\wh\cp_j$ asymptotically 
coincide with the oracle estimators under Assumption~\ref{assum_meta_est},
which allows us to establish the theoretical validity of the proposed bootstrap procedure 
along the same lines as in the AMOC situation.

For notational ease in the statement of theoretical results,
we slightly modify~\eqref{eq:boot:max} to
\begin{align*}
\wt\cp_j^* = \begin{cases}
{\arg\max}_{\wh\cp_j - G_j < k \le \wh{\cp}_j + G_j} \vert T_{k, n}(G_j; X^*) \vert & \text{for } 1 \le j \le \min(q_n,\wh q_n), \\
0 & \text{for } \wh q_n < j \le q_n
\end{cases}
\end{align*}
(see Remark~\ref{rem_argmax_boot}~\ref{rem_argmax_boot:one}
for the discussion on asymptotic equivalence between $G_j$ and $H_j$ appearing in~\eqref{eq:boot:max}).
Also, we define $\wh\cp_j = n$ for $\wh q_n < j \le q_n$.
In doing so, when $\wh q_n < q_n$,
the difference between $\wt\cp_j^*$ and $\wh\cp_j$ for $\wh q_n < j \le q_n$, 
is made as large as possible.
However, this does not influence the asymptotic result due to Assumption~\ref{assum_meta_est}. 
With these modifications, all the following statements involving the differences 
$\wt\cp_j^* - \wh\cp_j$ and $\wh\cp_j - \cp_j$ are well-defined for any $j = 1, \ldots, q_n$.

Then under Assumption~\ref{assum_meta_est},
each accessible estimator $\wh\cp_j$ coincides with 
the oracle estimator $\wt\cp_j$ on the asymptotic one-sets $\mc A_j$ and $\mc A$ (defined in the assumption)
such that
\begin{align}
\label{eq:hat:tilde:one}
\sup_{x \in \R} \l\vert 
\p\l(\sigma^{-2} d_j^2 \vert \wh\cp_j - \cp_j\vert\le x \r)
- \p\l(\sigma^{-2} d_j^2 \vert \wt\cp_j - \cp_j \vert \le x \r) \r\vert &\to 0
\end{align}
and, when $q_n = q$ is fixed,
\begin{align}
\label{eq:hat:tilde:two}
\sup_{\mbf x \in \R^q} \l\vert 
\p\l(\cap_{j = 1}^q \l\{\sigma^{-2} d_j^2 \vert \wh\cp_j - \cp_j\vert\le x_j\r\} \r)
- \p\l(\cap_{j = 1}^q \l\{\sigma^{-2} d_j^2 \vert \wt\cp_j - \cp_j \vert \le x_j\r\} \r) \r\vert &\to 0.
\end{align}
In Section~\ref{sec:asymp} below, we derive the asymptotic distribution of $\wt\cp_j - \cp_j$
and in Section~\ref{sec:bootstrap:theor}, we show that the difference $\wt\cp_j^* - \wh\cp_j$ 
(conditionally on $X_1, \ldots, X_n$) shares the same limit distribution.
Combined with~\eqref{eq:hat:tilde:one}--\eqref{eq:hat:tilde:two},
these results indicate that
we can approximate the quantiles of the true difference $\wh\cp_j - \cp_j$
by those of the bootstrap difference $\wt\cp_j^* - \wh\cp_j$, 
which are accessible via Monte Carlo methods.
From this, the (asymptotic) validity of the proposed pointwise and uniform CIs follow, i.e.\
\begin{align}
& \p\l( \cp_j \in \mc C^{\text{pw}}_j(\alpha) \r) \to 1 - \alpha \text{ for each $j = 1, \ldots, q_n$, \ and} \nn \\
& \p\l( \cap_{j = 1}^{q} \l\{\cp_j \in \mc C^{\text{unif}}_j(\alpha)\r\} \r) \to 1-\alpha.
\label{eq:boot:valid}
\end{align}

\subsection{Asymptotic distribution of oracle change point estimators}
\label{sec:asymp}

Theorem~\ref{thm:asymp} derives the asymptotic distribution of $\wt\cp_j$
both when the changes are local ($d_j = d_{j, n} \to 0$) and when they are fixed. 
Thanks to Assumption~\ref{assum_meta_est},
the same asymptotic behaviour holds for the accessible change point estimators
produced by MOSUM-based procedures. 

\begin{thm}
\label{thm:asymp}
Let $\{X_t\}_{t = 1}^n$ satisfy~\eqref{eq:model}--\eqref{eq:errors} and $G_j$ fulfil~\eqref{eq:bandwidth}.
\begin{enumerate}[label = (\alph*)]
\item \label{thm:asymp:one} If $d_j = d_{j, n} \to 0$ and 
$d_j^2 G_j \to \infty$,
then it holds as $n \to \infty$,
\begin{align*}
\sigma^{-2} d_j^2( \wt\cp_j - \cp_j) \dto \arg\max_s\l\{W_s - \vert s \vert/\sqrt{6}: \, s \in \R\r\}
\end{align*}
for $j = 1, \ldots, q_n$,
where $\{W_s: \, s \in \R\}$ is a standard Wiener process.

\item \label{thm:asymp:two} If $d_j$ is fixed
and the errors $\{\vep_t\}$ are continuous, 
then it holds as $n \to \infty$,
\begin{align*}
&  \wt\cp_j - \cp_j \dto 
 \arg\max_\ell \l\{ -d_j \Gamma_{\vep}(\ell) - \ell d_j^2: \, 
\ell \in \Z\r\}, \quad \text{with}
\\
& \Gamma_{\vep}(\ell) = \l\{ \begin{array}{ll}
\sum_{t = \ell}^{-1} \vep_t^{(1)} -  2 \sum_{t = \ell}^{-1} \vep_t^{(2)} + \sum_{t = \ell}^{-1} \vep_t^{(3)} & 
\text{when } \ell < 0,
\\
0 & \text{when } \ell = 0,
\\
\sum_{t = 1}^\ell \vep_t^{(1)} - 2 \sum_{t = 1}^\ell \vep_t^{(2)} + \sum_{t = 1}^\ell \vep_t^{(3)} & 
\text{when } \ell > 0
\end{array}\r.
\end{align*}
for $j = 1, \ldots, q_n$,
where  $\{\vep_t^{(i)}, \, t \in \Z \} \deq \{\vep_t, \, t \in \Z\}$, $i = 1, 2, 3$, 
are mutually independent copies of the original error sequence.

\item \label{thm:asymp:three} 
Suppose that the number of changes is fixed at $q_n = q$.
For each change point, let the assumptions in~\ref{thm:asymp:one} or \ref{thm:asymp:two} 
be fulfilled in addition to $4G_j < \delta_j$. Then it holds as $n \to \infty$,
\begin{align*}
& \sigma^{-2} \l(d_1^2 (\wt\cp_j - \cp_j), \ldots, d_q^2(\wt\cp_q - \cp_q) \r)
\dto \l(S_1, \ldots, S_q\r), \quad \text{where}
\\
& S_j = \l\{\begin{array}{ll}
\arg\max_s \{ W_s^{(j)} - |s|/\sqrt{6}:\, s \in \R\} & \text{when } d_j = d_{j, n} \to 0, \\
\sigma^{-2} d_j^2 \arg\max_\ell \{-d_j \Gamma^{(j)}_\vep(\ell) - \ell d_j^2: \, \ell \in \Z\} 
& \text{when $d_j$ is fixed},
\end{array}\r.
\end{align*}
with $W_j^{(j)}$ (resp. $\Gamma^{(j)}_\vep$), $j = 1, \ldots, q$, are
mutually independent and distributed according to~\ref{thm:asymp:one} (resp.~\ref{thm:asymp:two}).
\end{enumerate}
\end{thm}

In the case of local changes, the results reported in Theorem~\ref{thm:asymp}
are closely related to Theorem~3.3 of \cite{eichinger2018}
which also permits time series errors.
For the corresponding result in the AMOC situation, 
see \cite{antoch1995change} (local change) and \cite{antoch1999estimators} (fixed change).
Limiting distributions for multiple change point estimators 
have also been obtained by \cite{bai1998estimating} in the context of linear models,
\cite{yau2016} for a time series segmentation problem
and \cite{kaul2021inference} for the high-dimensional mean change point detection problem.

The additional assumption of continuity of the errors 
in the case of fixed changes (Theorem~\ref{thm:asymp}~\ref{thm:asymp:two}),
is required to avoid ties (a.s.) of the maximum of the limit distribution. 
If the error distribution is not continuous (e.g.\ discrete or mixed),
those ties may be resolved differently on the RHS of $\dto$ than on the LHS,
an issue stemming from that the $\arg\max$ is not continuous 
if the limit does not have a unique, isolated maximum.
Therefore, while the underlying process 
defining the $\arg\max$ on the LHS 
(denoted by $V_{k, n}$ in the proof given in Appendix~\ref{pf:thm:asymp})
will weakly converge to the process underlying 
the $\arg\max$ on the RHS (denoted by $\wt V_{k, n}$ in the proof) even for discrete errors,
the $\arg\max$ itself may not because the continuous mapping theorem is not applicable.
For the local change considered in~\ref{thm:asymp:one}, 
the Wiener process with drift on the RHS does not suffer from this issue
and thus the continuity of the errors is not required.
\cite{ferger2004continuous} provide additional insights into the theoretical behaviour of the $\arg\max$ if ties occur.
In practice, we may either ignore the ties or report their occurrence explicitly. 

As in the AMOC situation, the asymptotic behaviour of the oracle estimator $\wt\cp_j$
(and by Assumption~\ref{assum_meta_est}, the accessible estimator $\wh\cp_j$)
depends on the regime determined by the magnitude of the change and,
in the fixed change case, on the unknown error distribution.
Consequently, the limit distribution itself is not suitable for CI generation.
Section~\ref{sec:bootstrap:theor} shows that for the bootstrap estimators $\wt\cp_j^*$, we have
$\wt\cp_j^* - \wh\cp_j$ (conditional on the data) mimic the distribution of $\wt\cp_j - \cp_j$, 
and thus the bootstrap procedure produces asymptotically honest 
bootstrap CIs under Assumption~\ref{assum_meta_est}.

\subsection{Asymptotic distribution of bootstrap change point estimators}
\label{sec:bootstrap:theor}

Since the bootstrap procedure is based on the change point estimators 
$\wh\Cp = \{\wh\cp_j, \, 1 \le j \le \wh q_n\}$,
we require the estimators to be sufficiently precise in the following sense:
\begin{assum}
\label{assum_precision_est}
	For given $j \in \{1, \ldots, q_n\}$, we have
	\begin{align*}
&	\wh\cp_{i} - \cp_{i} = o_P(\delta_i) \quad \text{for} \quad i \in \{j - 1, j, j + 1\},  \qquad \text{and} \\
&	d_i^2(\wh\cp_{i} - \cp_{i}) = o_P(d_j^2|\cp_j - \cp_{i}|) \quad \text{for} \quad 
i \in \{j - 1, j + 1\}.
	\end{align*}
\end{assum}

As in the case of Assumption~\ref{assum_meta_est},
MOSUM-based change point detection procedures 
achieve consistency in multiple change point estimation
and thus produce estimators that fulfil Assumption~\ref{assum_precision_est};
we refer to Appendix~\ref{sec:mosum} for detailed discussions.

\begin{thm}
\label{thm:bootstrap}
Denote $\p^*(\cdot) = \p(\cdot|X_1,\ldots,X_n)$.
Let~\eqref{eq:model}--\eqref{eq:errors} 
and Assumption~\ref{assum_precision_est} hold 
(for a given $j$ in \ref{thm:bootstrap:one} and \ref{thm:bootstrap:two}, and for all $j$ in \ref{thm:bootstrap:three}),
and $G_j$ fulfil~\eqref{eq:bandwidth}.
\begin{enumerate}[label = (\alph*)]
\item \label{thm:bootstrap:one}
If $d_j = d_{j, n} \to 0$ and 
$d_j^2G_j \to \infty$,
then the following limit distribution holds for all $x \in \R$ as $n \to \infty$,
\begin{align*}
	\p^*\left(	\sigma^{-2} d_j^2 (\wt{\cp}^*_j - \wh\cp_j) \le x\right) \pto
	\p\left( \arg\max_{s \in \R}\l\{W_s - |s|/\sqrt{6}\r\} \le x\right) 
\end{align*}
for each $j = 1, \ldots, q_n$, where $\{W_s\}$ is as in Theorem~\ref{thm:asymp}~\ref{thm:asymp:one}.

\item \label{thm:bootstrap:two} If $d_j$ is fixed 
and the errors $\{\vep_j\}$ are continuous, 
then the following limit distribution holds for all $x \in \R$ as $n \to \infty$,
\begin{align*}
	\p^*\left( \wt{\cp}^*_j - \wh\cp_j \le x \right) \pto
	\p\left( \arg\max_{\ell\in\Z} \l\{- d_j \Gamma_{\vep}(\ell) - \ell d_j^2\r\}\le x\right)
\end{align*}
for $j = 1, \ldots, q_n$, where $\{\Gamma_{\vep}(\ell)\}$ is as in Theorem~\ref{thm:asymp}~\ref{thm:asymp:two}.

\item \label{thm:bootstrap:three}
Suppose that the number of changes is fixed at $q_n = q$.
For each change point, let the assumptions in \ref{thm:bootstrap:one} or \ref{thm:bootstrap:two}  
be fulfilled in addition to $4 G_j < \delta_j$. 
Then, the following limit distribution holds for all $\mathbf{x} = (x_1, \ldots, x_q)^\top \in \R^q$ as $n \to \infty$,
\begin{align*}
&\p^*\left( \sigma^{-2} d_1^2 (\wt\cp_1^* - \wh\cp_1) \ls x_1, \ldots,
\sigma^{-2} d_q^2(\wt\cp_q^* - \wh\cp_q) \ls x_q \r) \pto
\p\l(S_1\ls x_1, \ldots, S_q\ls x_q\r), \text{ where}
\\
& S_j = \l\{\begin{array}{ll}
\arg\max_s \{ W_s^{(j)} - |s|/\sqrt{6}:\, s \in \R\} & \text{when } d_j = d_{j, n} \to 0, \\
\sigma^{-2} d_j^2 \arg\max_\ell \{-d_j \Gamma^{(j)}_\vep(\ell) - \ell d_j^2: \, \ell \in \Z\} 
& \text{when $d_j$ is fixed},
\end{array}\r.
\end{align*}
with $\{W_j^{(j)}\}$ (resp. $\{\Gamma^{(j)}_\vep\}$), $j = 1, \ldots, q$, are
mutually independent and distributed according to~\ref{thm:asymp:one} (resp.~\ref{thm:asymp:two}).
\end{enumerate}
\end{thm}


To the best of our knowledge, the literature on bootstrap CIs for change points
considers only the case of local changes with a distribution-free limit;
an exception is \cite{antoch1999estimators}
where their Theorem~7.1 (given without an explicit proof)
is on the fixed change case in the AMOC setting.

\subsection{Consistency of the bootstrap procedure}
\label{sec:consistency}

Recall that by Assumption~\ref{assum_meta_est},
the accessible estimators $\wh\cp_j$ coincide with the oracle ones $\wt\cp_j$
on asymptotic one-sets such that~\eqref{eq:hat:tilde:one}--\eqref{eq:hat:tilde:two} follow.
Then, Theorems~\ref{thm:asymp} and~\ref{thm:bootstrap}
establish that for each $j = 1, \ldots, q_n$,
\begin{align*}
	\sup_{x \in \R} \l\vert \p^*\l(\sigma^{-2} d_j^2 \vert \wt\cp_j^*
	- \wh\cp_j\vert\le x \r)
	- \p\l(\sigma^{-2} d_j^2 \vert \wh\cp_j - \cp_j \vert \le x \r) \r\vert \pto 0
\end{align*}
and, when $q_n = q$ is fixed,
\begin{align*}
\sup_{\mbf x \in \R^q} \l\vert 
\p^*\l(\cap_{j = 1}^q \l\{\sigma^{-2} d_j^2 \vert \wt\cp_j^*
- \wh\cp_j \vert \le x_j \r\} \r)
- \p^*\l(\cap_{j = 1}^q \l\{\sigma^{-2} d_j^2 \vert \wh\cp_j - \cp_j \vert \le x_j \r\} \r) 
\r\vert \pto 0.
\end{align*}

Together with that $\wh d_j / d_j \pto 1$ (Lemma~\ref{lemma_variance_boot})
and Assumption~\ref{assum_meta_est},
the validity of the bootstrap CIs proposed in~\eqref{eq:pw:ci}--\eqref{eq:unif:ci} follows
in the sense of~\eqref{eq:boot:valid}.
In particular, the bootstrap CIs are asymptotically honest 
whether the changes are local or fixed,
although their construction does not require the knowledge of the regime determined by the magnitude of the changes
or the error distribution.
The additional model selection step involved in the estimators $\wh\cp_j$
does not alter the theoretical validity of the bootstrap CIs by Assumption~\ref{assum_meta_est}.

The simulation studies reported in Section~\ref{sec:sim} show that 
the coverage of the bootstrap CI constructed with 
the oracle estimators $\wt\cp_j$ is generally right on target as expected from the above asymptotic theory.
 While asymptotically equivalent, bootstrap CIs constructed with the estimators $\wh\cp_j$
involving an additional model selection step,
have somewhat more conservative coverage in finite samples.
Heuristically, this is not surprising as in the latter case,
the empirical coverage is computed conditioning on the success of the model selection step, 
and changes underlying those realisations belonging to the conditioning set
tend to be more pronounced.

\section{Extensions} 
\label{sec:ext}

\subsection{Asymmetric bandwidths}
\label{sec:asymm}


The MOSUM statistic defined in~\eqref{eq:mosum:symm}
is readily extended to accommodate the use of asymmetric bandwidths $\mbf G = (G_\ell, G_r)$, as
\begin{align*}
T_{k, n}(\mbf G; X) = \sqrt{\frac{G_\ell G_r}{G_\ell + G_r}} \l(\bar{X}_{k - G_\ell, k}
- \bar{X}_{k, k + G_r}\r), \quad k = G_\ell, \ldots, n - G_r.
\end{align*}
In practice, provided that the asymmetric bandwidth is not too unbalanced,
its use can improve small sample performance of the MOSUM procedure,
see Figure~6 of \cite{meier2021mosum} for an illustration.
Their Theorem~1 extends the asymptotic null distribution
of the MOSUM test statistic to the asymmetric case and similarly,
we can extend Theorem~\ref{thm:asymp} and derive the asymptotic distribution
of the corresponding (oracle) change point estimators obtained as in~\eqref{eq:cp:tilde},
i.e.\ with the bandwidth $\mbf G_j = (G_{j, \ell}, G_{j, r})$,
\begin{align*}
\wt\cp_j = {\arg\max}_{\cp_j - G_{j, \ell} < k \le \cp_j + G_{j, r}}
\vert T_{k, n}(\mbf G_j; X) \vert.
\end{align*}
Analogously, we obtain the bootstrap estimators as
\begin{align*}
\wt\cp_j^* = {\arg\max}_{\wh\cp_j - H_{j, \ell}< k \le \wh\cp_j + H_{j, r}}
\vert T_{k, n}(\mbf G_j; X^*) \vert,
\end{align*}
where $H_{j, \ell} = \min\{G_{j, \ell}, 2(\wh\cp_j - \wh\cp_{j - 1})/3\}$ and 
$H_{j, r} = \min\{G_{j, r}, 2(\wh\cp_{j + 1} - \wh\cp_j)/3\}$, similarly as in Section~\ref{sec:method}.
Then analogously, we approximate the distribution of $\wt\cp_j - \cp_j$
with that of $\wt\cp_j^* - \wh\cp_j$,
using the symmetric construction of the CIs 
by means of $Q_j(\alpha)$ and $Q(\alpha)$ defined as in~\eqref{eq:pw:ci} and~\eqref{eq:unif:ci}.

\subsection{Dependent errors}
\label{sec:dependent}


In practice, it is more natural to allow for serial dependence in $\{\vep_t\}$. 
In the context of testing for a mean in the AMOC setting, different
time series bootstrap methods have successfully been applied 
such as block permutation \cite{kirch2007block},  
block bootstrap \cite{sharipov2016sequential} or  
frequency domain-based \cite{kirch2011} methods,
and subsampling has been studied by \cite{betken2018subsampling}
in the context of mean change point analysis in long-range dependent time series;
there also exist bootstrap-based testing procedures for more complex change point problems,
see e.g.\ \cite{bucher2016dependent} and \cite{emura2021change}.

For the multiple change point detection problem in~\eqref{eq:model},
compared to the i.i.d.\ setting, there are fewer methods that
guarantee consistent change point estimation when serial correlations are permitted in $\{\vep_t\}$,
such as those proposed in 
\cite{tecuapetla2017}, \cite{dette2018}, \cite{romano2020} and \cite{cho2020multiple}.
The single-scale MOSUM procedure studied in \cite{eichinger2018}
and the multiscale MOSUM procedure combined 
with the localised pruning proposed in \cite{cho2019two},  
have been shown to yield consistent estimators for heavy-tailed and/or serially correlated $\{\vep_t\}$,
see Appendix~\ref{sec:mosum}.

A natural question is whether we can automatically adapt to the regime determined by the magnitude of changes
using time series bootstrap methods. While for local changes, most time series bootstrap methods are expected to yield consistent results, this is no longer the case for the fixed change situation 
where a standard block bootstrap procedure will not work off the shelf.
In this section, we aim at explaining where the main difficulties lie in the construction of bootstrap CIs
for change point locations in time series settings.

For local changes, the limit distribution follows from a central limit theorem
such that time series bootstrap methods are expected to work well.
Indeed, in the AMOC setting with a local change,
\cite{huvskova2008, huvskova2010} propose to use a block bootstrap and 
show its asymptotic validity.
More precisely, in place of Step~1 of the bootstrap procedure proposed in Section~\ref{sec:method},
one draws blocks of length $K$ (with $K = K_n \to \infty$ at an appropriate rate) 
from the estimated residuals $\wh\vep_t = X_t - \wh f_t$ to form a bootstrap sample $\{X_t^*\}$,
where $\wh f_t$ denotes the piecewise constant signal 
that takes into account the possible presence of the single change point.
With some additional technicality, 
the results in \cite{huvskova2010} can be extended 
to show the consistency of thus-constructed bootstrap CIs 
for multiple local changes under the problem~\eqref{eq:model} considered here.
Other time series bootstrap procedures 
such as a stationary bootstrap, dependent wild bootstrap or even frequency domain methods
are similarly conjectured to achieve consistency.

However, the case of fixed changes needs to be handled with more care 
in the presence of serial dependence,
since the limit distribution is no longer based on a central limit theorem
such that one cannot generally expect a bootstrap procedure to work well.
The success of the bootstrap method in the i.i.d.\ case is due to that $\wt\cp_j - \cp_j = O_P(1)$ 
(see~\eqref{eq:tilde:consist} in the proof of Theorem~\ref{thm:asymp} in Appendix~\ref{pf:thm:asymp}), 
i.e.\ the asymptotic distribution of $\wt\cp_j$ 
effectively depends on a sequence of finitely many errors. 
The joint distribution of this sequence must be mimicked correctly for the construction of the CIs,
and the i.i.d.\ bootstrap described in Section~\ref{sec:method}
correctly approximates the joint distribution of finitely many independent errors asymptotically
thanks to the consistency of the empirical distribution function,
see~\eqref{eq_boot_FCLT} in the proof of Theorem~\ref{thm:bootstrap}.

In the time series case, for the validity of bootstrap CIs, 
a bootstrap procedure is required to correctly mimic the joint dependence structure of 
the three relevant finite stretches appearing in~\eqref{eq_boot_FCLT}
(see also~\eqref{eq_Unl} to see where the three stretches come from).
While the three stretches are (asymptotically) independent under appropriate assumptions, 
for correct approximation of the joint distribution,
each of these stretches needs to be covered by a single block;
if a bootstrap procedure does not fulfil this requirement
and two blocks are involved in covering one of those stretches
(involvement of more than two blocks is not possible asymptotically 
as the block length diverges while the length of each stretch is finite), 
then those two blocks are (conditionally) independent unlike the original series 
and \eqref{eq_boot_FCLT} does not hold.

One possible approach to fulfil this requirement
is to center bootstrapped blocks from the residuals $\{\wh\vep_t\}$ at the estimator $\wh\cp_j$
as well as $\wh\cp_j \pm G_j$ for individual $j = 1, \ldots, \wh q_n$.
Since in effect, the asymptotic distribution 
such as that reported in Theorem~\ref{thm:asymp}~\ref{thm:asymp:two}
depends on finitely many observations around $\wh\cp_j$ and $\wh\cp_j \pm G_j$ only,
this bootstrap procedure essentially amounts to subsampling
where only one block for each of the three stretches (which can be considered as a subsample)
is involved in determining the distribution of $\wt\cp_j^*$ for individual change points.
Alternatively, for small enough bandwidths $G_j$ 
(asymptotically, since smaller bandwidths are permitted as more moments exist for $\{\vep_t\}$),
one can use a block length $K$ diverging faster than $G_j$ 
such that a single block centered at $\wh\cp_j$ covers all three stretches simultaneously. 
A similar idea is explored in \cite{ng2021bootstrap} 
who also suggest to apply subsampling locally at the estimated change point locations.

To summarise, while many time series resampling 
procedures are expected to return valid bootstrap CIs for local changes,
the same will typically not be the case for fixed changes
without a carefully designed subsampling method that 
takes into account the specific structure of the asymptotic distribution involved,
and their good practical performance will require longer stretches of stationarity between adjacent change points.

\section{Numerical studies}
\label{sec:sim}

In this section, we investigate the practical performance of the bootstrap CIs on simulated datasets.
We consider both the bootstrap CIs constructed with
the oracle estimators $\wt\cp_j$ as in~\eqref{eq:cp:tilde}
(which are inaccessible in practice),
and those based on the change point estimators $\wh\cp_j$
which are obtained after a model selection step.
In the former case, we expect the bootstrap CIs to closely attain the given confidence level 
since the bootstrap actually mimics the distribution of $\wt\cp_j$.
While the model selection step employed for $\wh\cp_j$ is asymptotically negligible,
simulation results suggest that it leads to somewhat more conservative CIs for the latter case 
in small samples.

\subsection{Set-up}
\label{sec:sim:model}

We consider the test signals {\tt blocks}, {\tt fms}, {\tt mix}, {\tt teeth10} and {\tt stairs10}
first introduced in \cite{fryzlewicz2014},
see Figure~\ref{fig:testsignals} in Appendix 
which plots realisations from the five test signals with Gaussian errors.
We introduce an additional scaling factor $\vartheta$ and modify the test signals as follows:
Denoting the mean of the original test signal 
by $f^\circ_t = f^\circ_0 + \sum_{j = 1}^{q_n} d_j^{\circ} \mathbb{I}_{\{t > \cp_j^{\circ} \}}$,
with the locations of the change points therein by $\cp^\circ_j$
and the (signed) size of change by $d^\circ_j$ for $j = 1, \ldots, q_n$, 
we consider the scaled signals
$f_t = f^\circ_0 + \sum_{j = 1}^{q_n} d_j \mathbb{I}_{\{t > \cp_j \}}$
with $d_j = d^\circ_j/\vartheta$ and 
the change points $\cp_j, \, j = 1, \ldots, q_n$ satisfying
$\cp_{j + 1} - \cp_j = \vartheta^2(\cp^\circ_{j + 1} - \cp^\circ_j)$.
In doing so, we keep the detectability of each change point
determined by $d_j^2 \delta_j$ constant across the scaling factor $\vartheta$,
while exploring the two different regimes -- 
the local change where $d_j = d_{j, n} \to 0$
and the fixed change with constant $d_j$ --
by varying $\vartheta \in \{1, 2, 4\}$.
In what follows, we only report the results from $\vartheta \in \{1, 4\}$ 
(with $\vartheta = 1$ corresponding to the fixed change regime
and $\vartheta = 4$ to the local one) for brevity.
Also, we only provide the results for {\tt mix} and {\tt teeth10} test signals
when $\{\vep_t\}$ follow Gaussian distributions in the main text,
and the rest of the simulation results, including when the errors follow $t_5$ distributions, 
are given in the supplement (see Appendix~\ref{sec:sim:add});
we observe little difference is observed in the results
obtained with either Gaussian or $t_5$-distributed errors. 

All results given below and in Appendix~\ref{sec:sim:add}
are based on $2000$ realisations for each simulation setting 
(with the exception of Section~\ref{sec:sim:comp} where $1000$ realisations were generated),
and we set $B = 1000$ for bootstrap sample generation.
We consider $1 - \alpha \in \{0.8, 0.9, 0.95\}$ for the confidence levels.

\subsection{Results}

\subsubsection{Bootstrap CIs constructed with the oracle estimators in~\eqref{eq:cp:tilde}}
\label{sec:sim:notest}

\begin{figure}[htbp]
\begin{subfigure}{\textwidth}
	\centering
\includegraphics[height=.39\textheight]{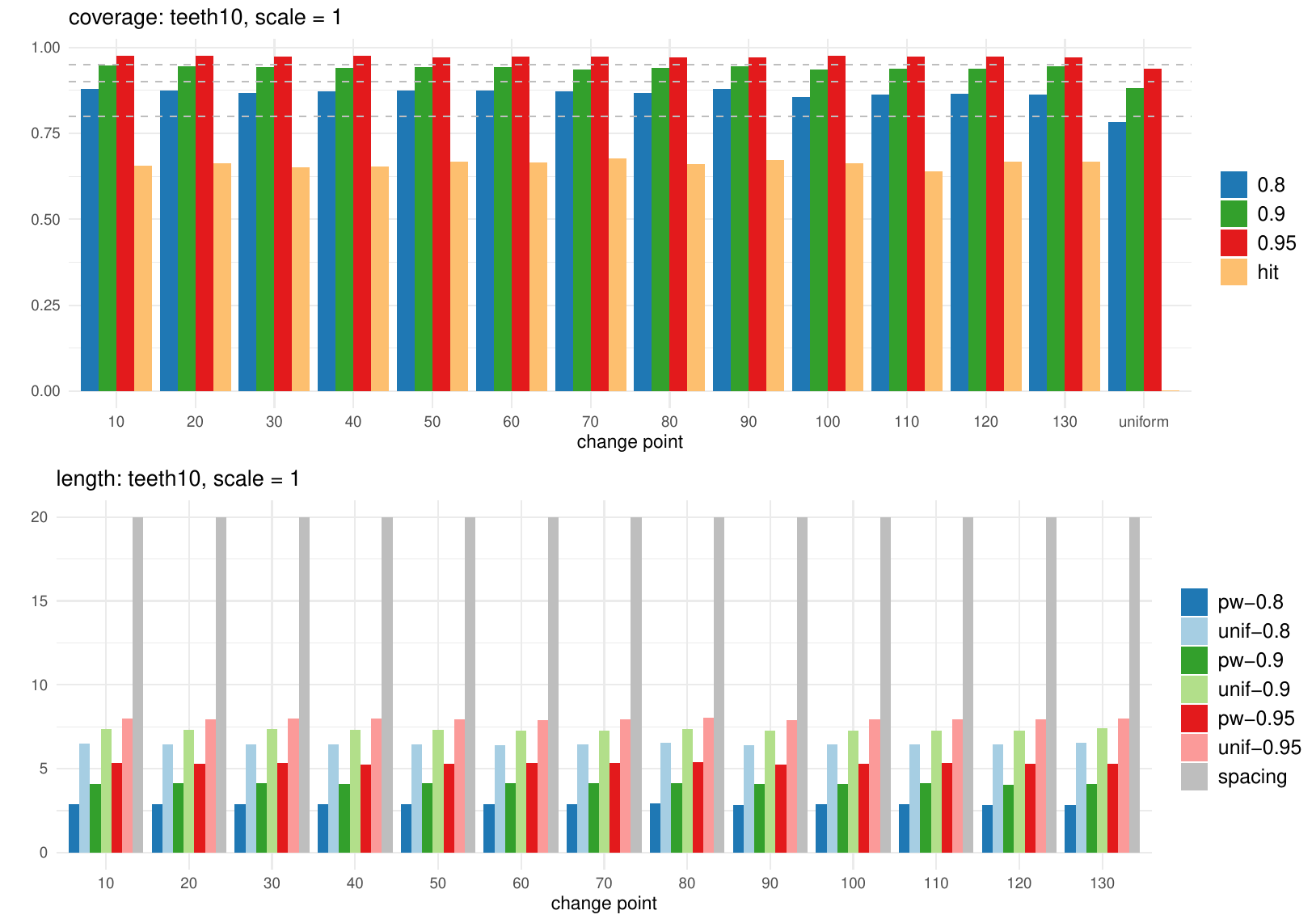}
\subcaption{ $\vartheta = 1$.}
\end{subfigure}
\begin{subfigure}{\textwidth}
	\centering
\includegraphics[height=.39\textheight]{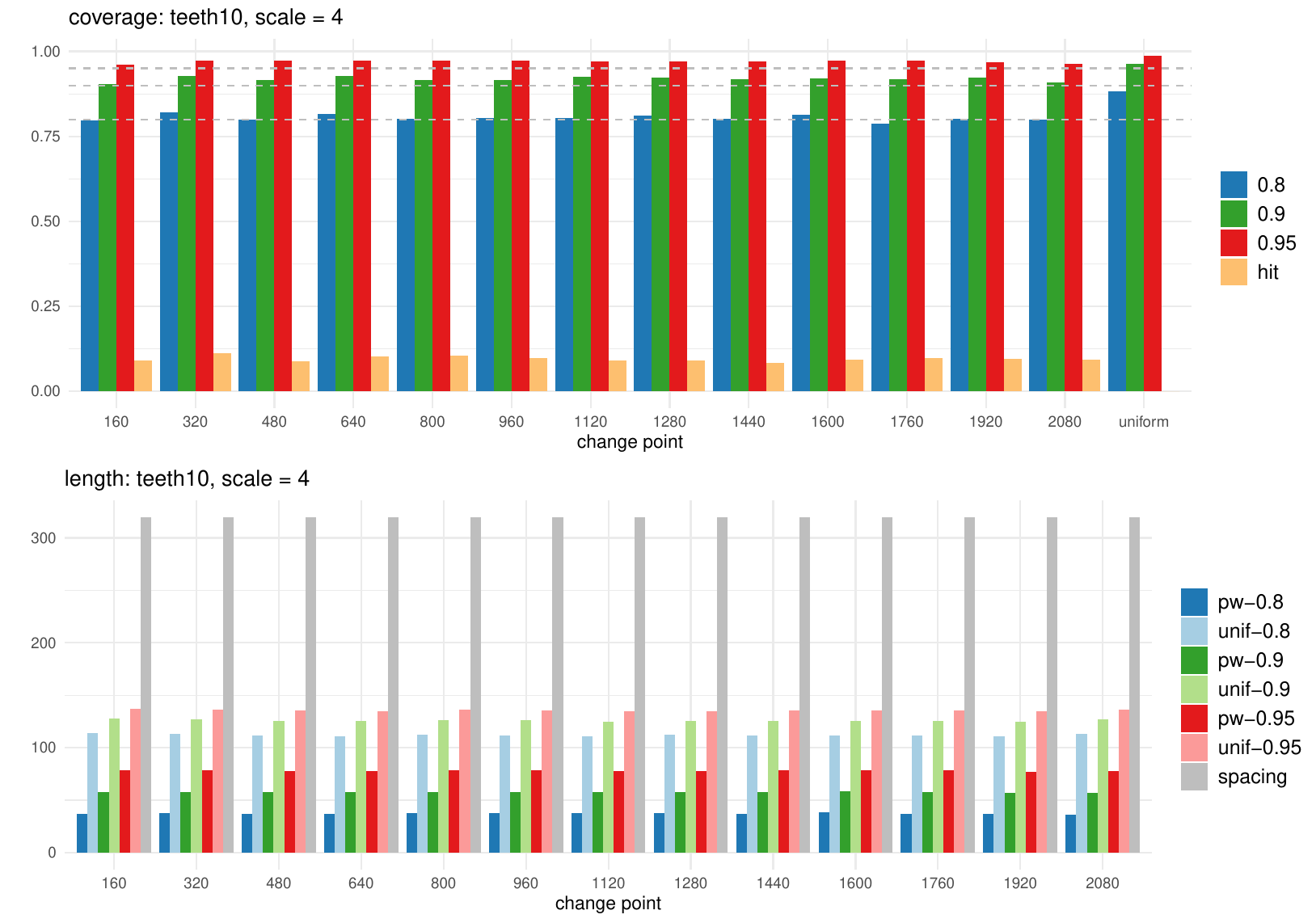}
\subcaption{$\vartheta = 4$.}
\end{subfigure}
\caption{{\tt teeth10}:
Bootstrap CIs constructed with the oracle estimators in~\eqref{eq:cp:tilde}. 
Top panel (of each sub-figure): coverage of pointwise bootstrap CIs for each $\cp_j$
(their locations given as the $x$-axis labels) and that of the uniform ones. 
Horizontal lines indicate $1 - \alpha \in \{0.8, 0.9, 0.95\}$.
We also report the proportion of the event where $\wh\cp_j = \cp_j$ exactly (`hit').
Bottom: lengths of pointwise and uniform bootstrap CIs
at $1 - \alpha \in \{0.8, 0.9, 0.95\}$. 
The grey columns `spacing' reports  
twice the minimum distance to adjacent change points $2\delta_j$,
for each $\cp_j, \, j = 1, \ldots, q_n$.}
\label{fig:teeth10}
\end{figure}

\begin{figure}[ht!]
\begin{subfigure}{\textwidth}
	\centering
\includegraphics[height=.39\textheight]{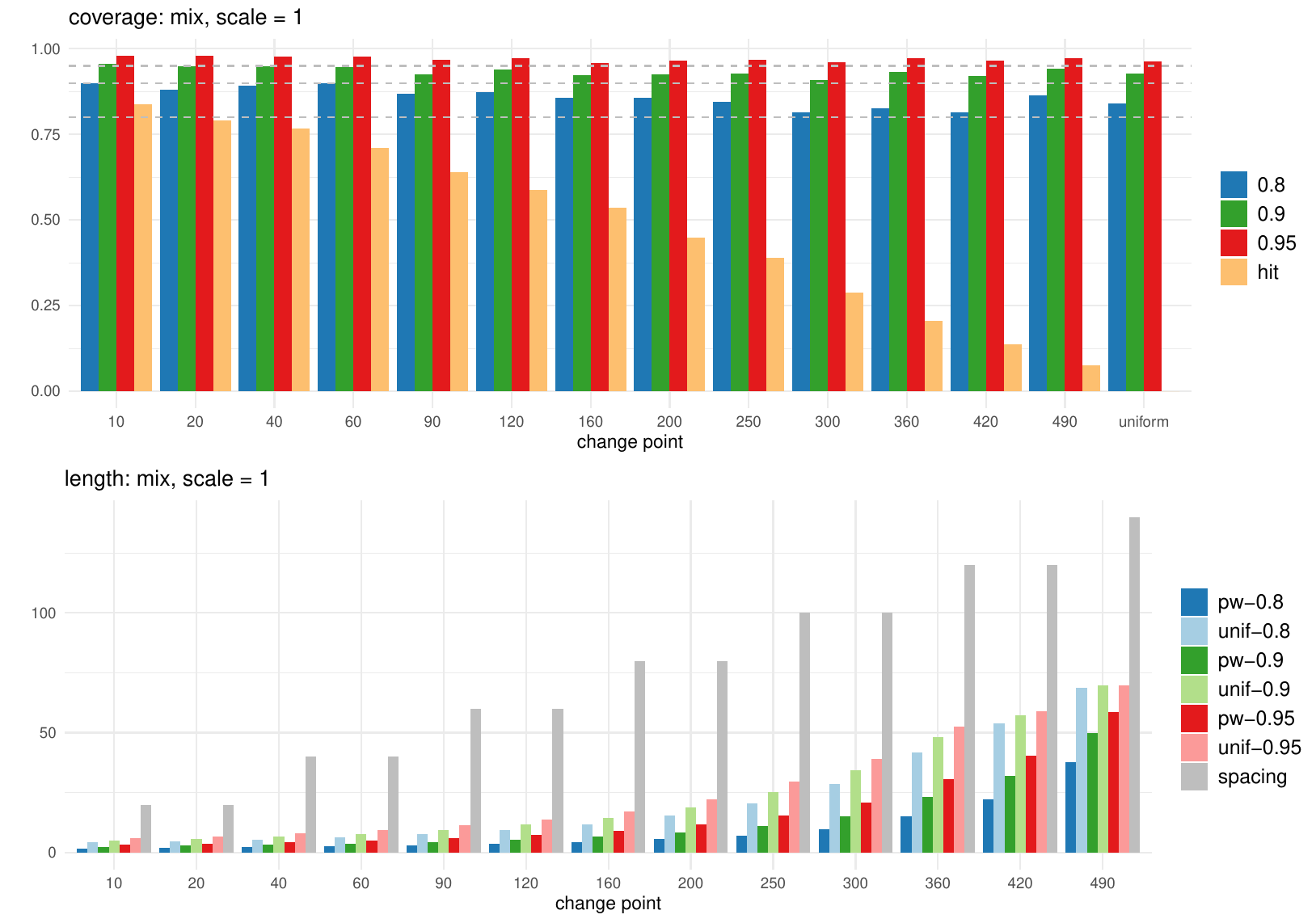}
\subcaption{ $\vartheta = 1$.}
\end{subfigure}
\begin{subfigure}{\textwidth}
	\centering
\includegraphics[height=.39\textheight]{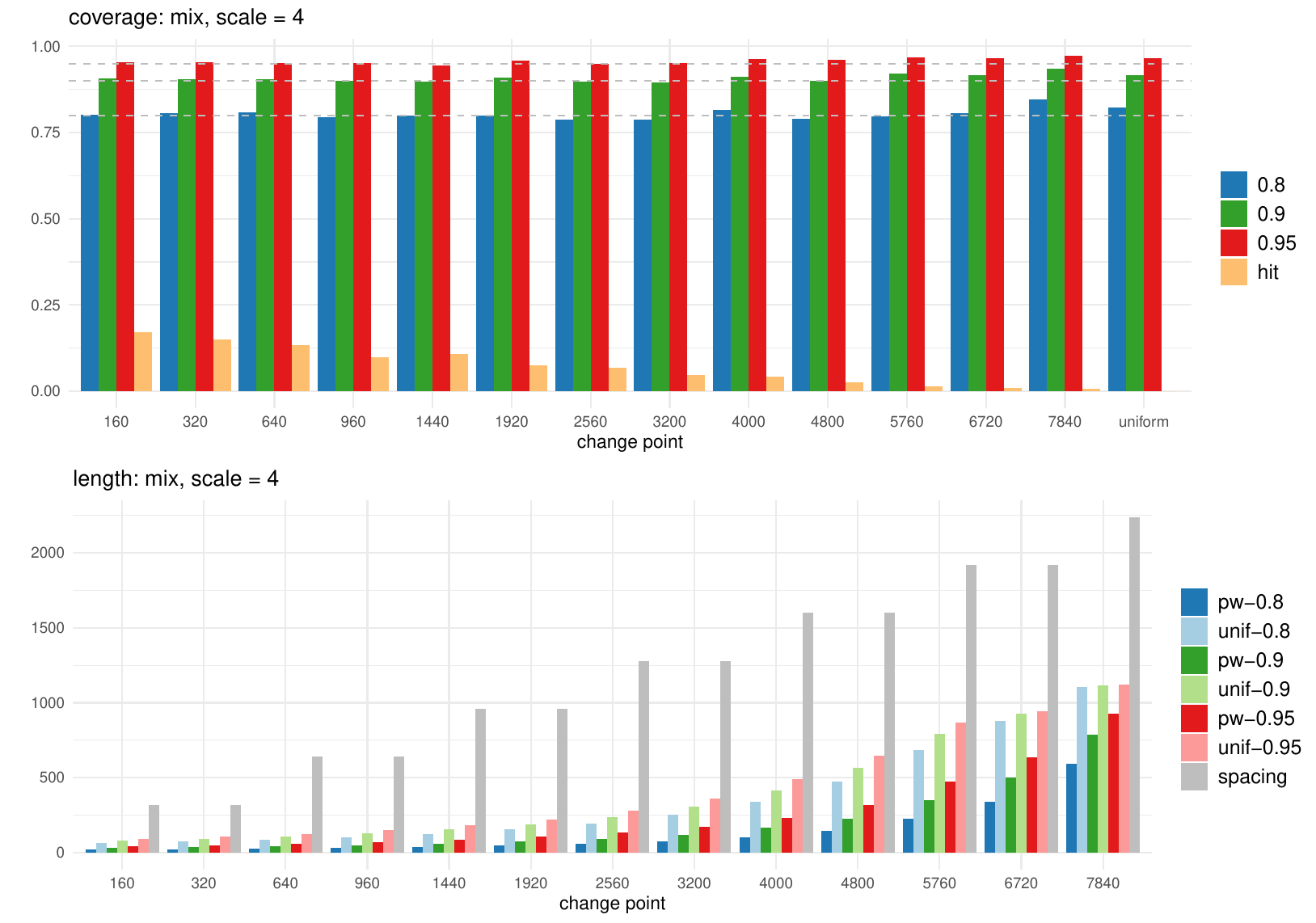}
\subcaption{$\vartheta = 4$.}
\end{subfigure}
\caption{{\tt mix}: Bootstrap CIs constructed with the oracle estimators in~\eqref{eq:cp:tilde}. See Figure~\ref{fig:teeth10} for detailed descriptions.}
\label{fig:mix}
\end{figure}

We first investigate the coverage and the length of bootstrap CIs
generated with the estimators $\wt\cp_j$ defined in~\eqref{eq:cp:tilde}
with $G_j = \delta_j/2$ for all $j = 1, \ldots, q_n$,
which do not involve any model selection step
that amounts to testing whether there indeed exist change points in their vicinity or not.
When possibly multiple change points are present,
the oracle estimators $\wt\cp_j$ are accessible only in simulations.
In contrast, an analogue of $\wt\cp_j$ is accessible in the AMOC setting,
and \cite{huvskova2008} take a similar approach in their simulation studies. 
For given $j$, the coverage of pointwise CIs
is calculated as the proportion of simulation realisations
where $\mc C_j^{\text{pw}}$ contains $\cp_j$.
For the uniform CIs, it is calculated as the proportion of the realisations
where the uniform bootstrap CIs $\mc C_j^{\text{unif}}$
contain the corresponding $\cp_j$ simultaneously for all $j = 1, \ldots, q_n$.
The lengths of CIs reported are obtained by averaging the respective CIs over the $2000$ realisations.

Figures~\ref{fig:teeth10}--\ref{fig:mix} report the coverage and the lengths of
bootstrap CIs for {\tt teeth10} and {\tt mix} 
test signals with $\vartheta \in \{1, 4\}$,
when $\{\vep_t\}$ are generated from Gaussian distributions.
The reported coverage is close to the nominal level throughout the test signals and $\vartheta$,
while the lengths of CIs are not trivial,
i.e.\ the CIs are considerably shorter than the distance to neighbouring change points.
 We observe that the coverage gets closer to the nominal level
with increasing $\vartheta$, i.e.\ as the size of changes corresponds to the local change regime,
which agrees with that the hit rate is considerably lower 
when $\vartheta = 4$ compared to when $\vartheta = 1$.
Since the bootstrap CIs are for discrete quantities, 
they are expected not to achieve the confidence level exactly
but to be on the conservative side,
particularly with smaller $\vartheta$ which corresponds to the fixed change regime
(see Theorem~\ref{thm:asymp}~\ref{assum_meta_est_two} where
the limit distribution of $\wt\cp_j$ is discrete, 
in contrast to that of the local change regime as in~\ref{assum_meta_est_one}).
Between $\vartheta \in \{1, 4\}$,
the absolute lengths of the CIs are naturally greater when $\vartheta = 4$,
but their ratio to the corresponding minimum spacing $\delta_j$ remains approximately constant 
across $\vartheta$.

\subsubsection{Bootstrap CIs constructed with model selection}
\label{sec:sim:lp}

\begin{figure}[htp!]
\begin{subfigure}{\textwidth}
	\centering
\includegraphics[height=.39\textheight]{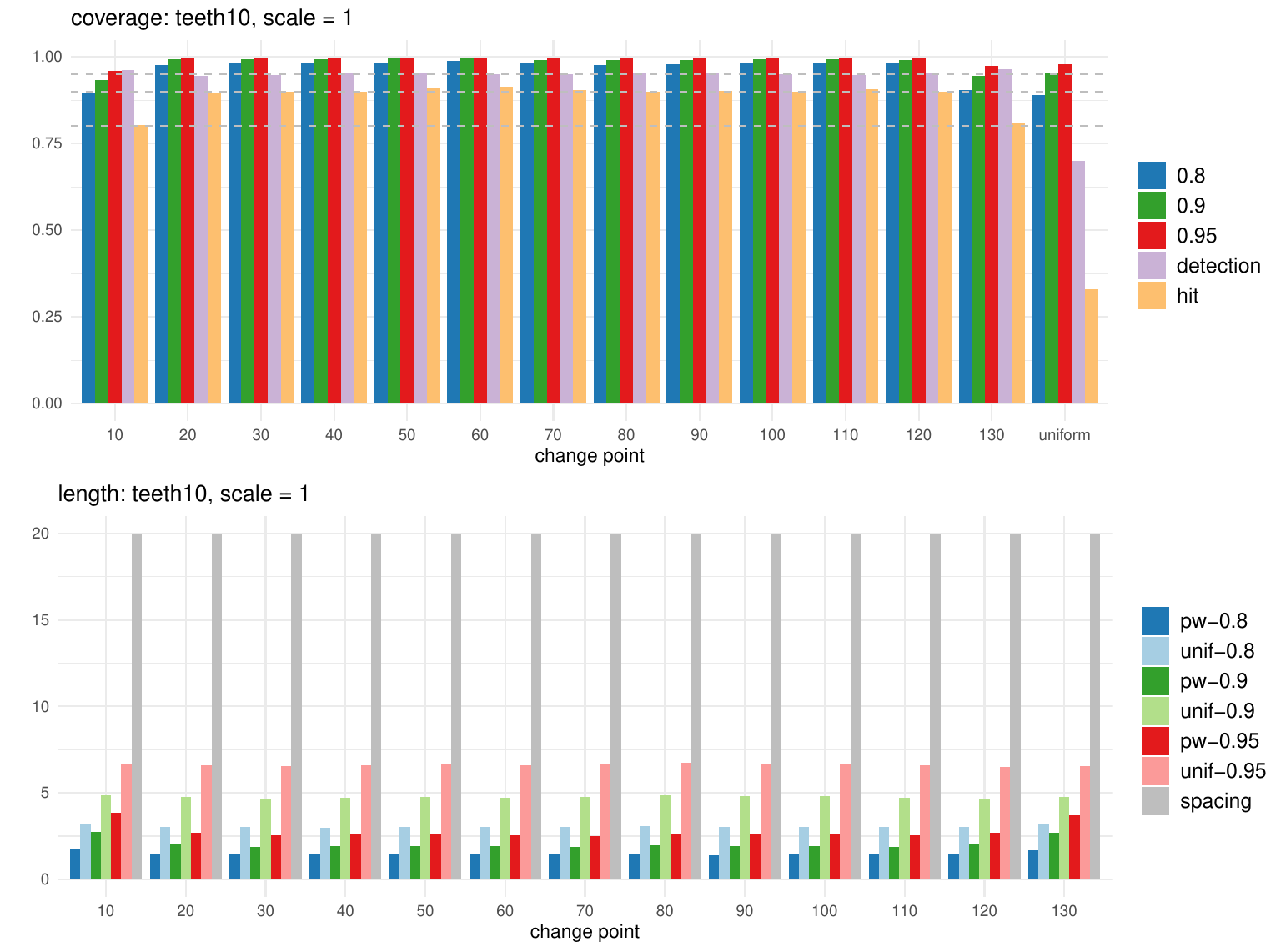}
\subcaption{ $\vartheta = 1$.}
\end{subfigure}
\begin{subfigure}{\textwidth}
	\centering
\includegraphics[height=.39\textheight]{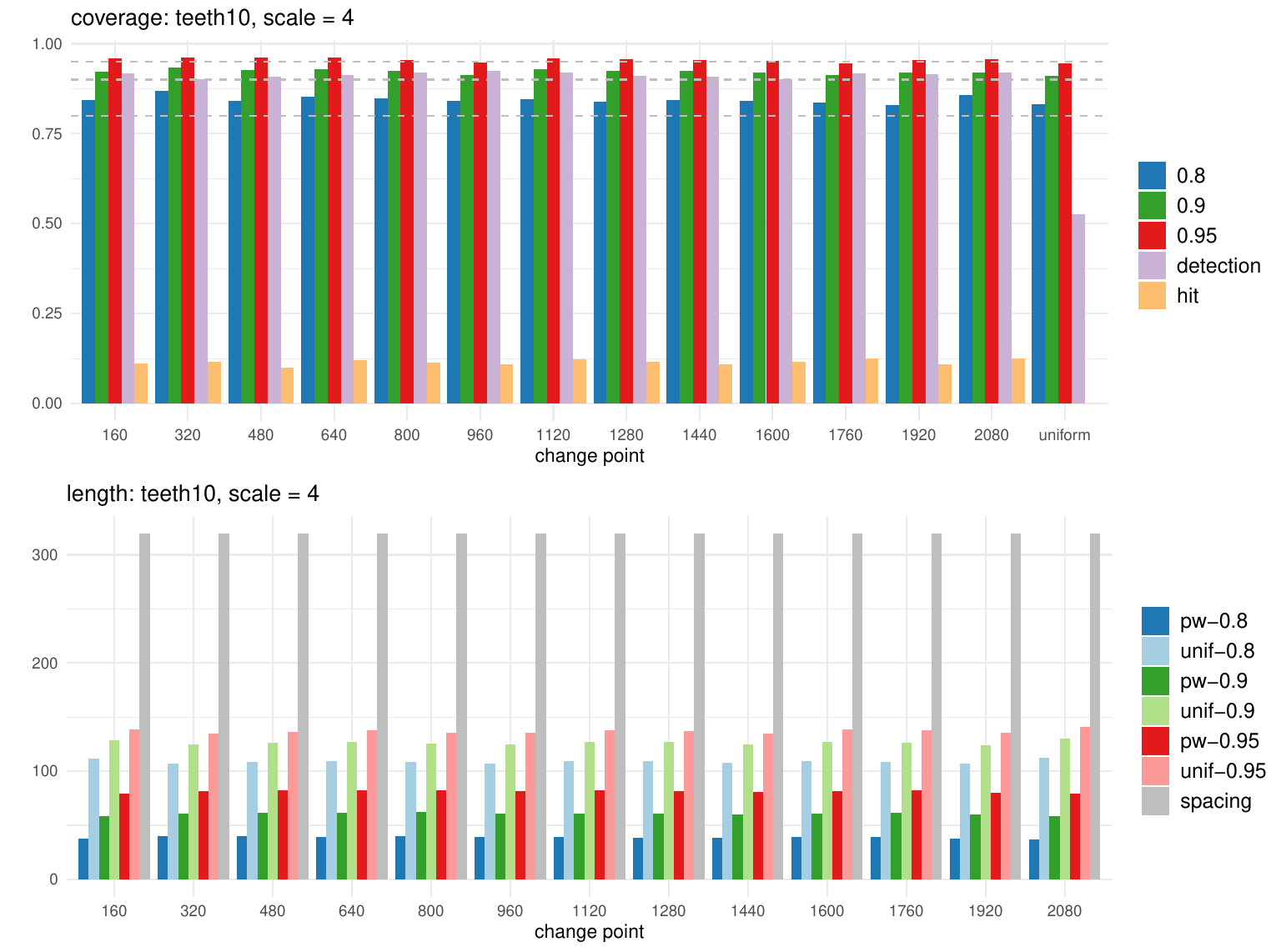}
\subcaption{$\vartheta = 4$.}
\end{subfigure}
\caption{{\tt teeth10}: Bootstrap CIs constructed with model selection.
Top panel (of each sub-figure): coverage of pointwise bootstrap CIs for each $\cp_j$
(their locations given as the $x$-axis labels) and that of the uniform ones. 
Horizontal lines indicate $1 - \alpha \in \{0.8, 0.9, 0.95\}$.
We also report the proportion of the event where $\wh\cp_j = \cp_j$ exactly (`hit'),
and that where we do detect the corresponding change points (`detection').
Bottom: lengths of pointwise and uniform bootstrap CIs
at $1 - \alpha \in \{0.8, 0.9, 0.95\}$. 
The grey columns `spacing' reports  
twice the minimum distance to adjacent change points $2\delta_j$,
for each $\cp_j, \, j = 1, \ldots, q_n$.}
\label{fig:full:teeth10}
\end{figure}

\begin{figure}[htp!]
\begin{subfigure}{\textwidth}
	\centering
\includegraphics[height=.39\textheight]{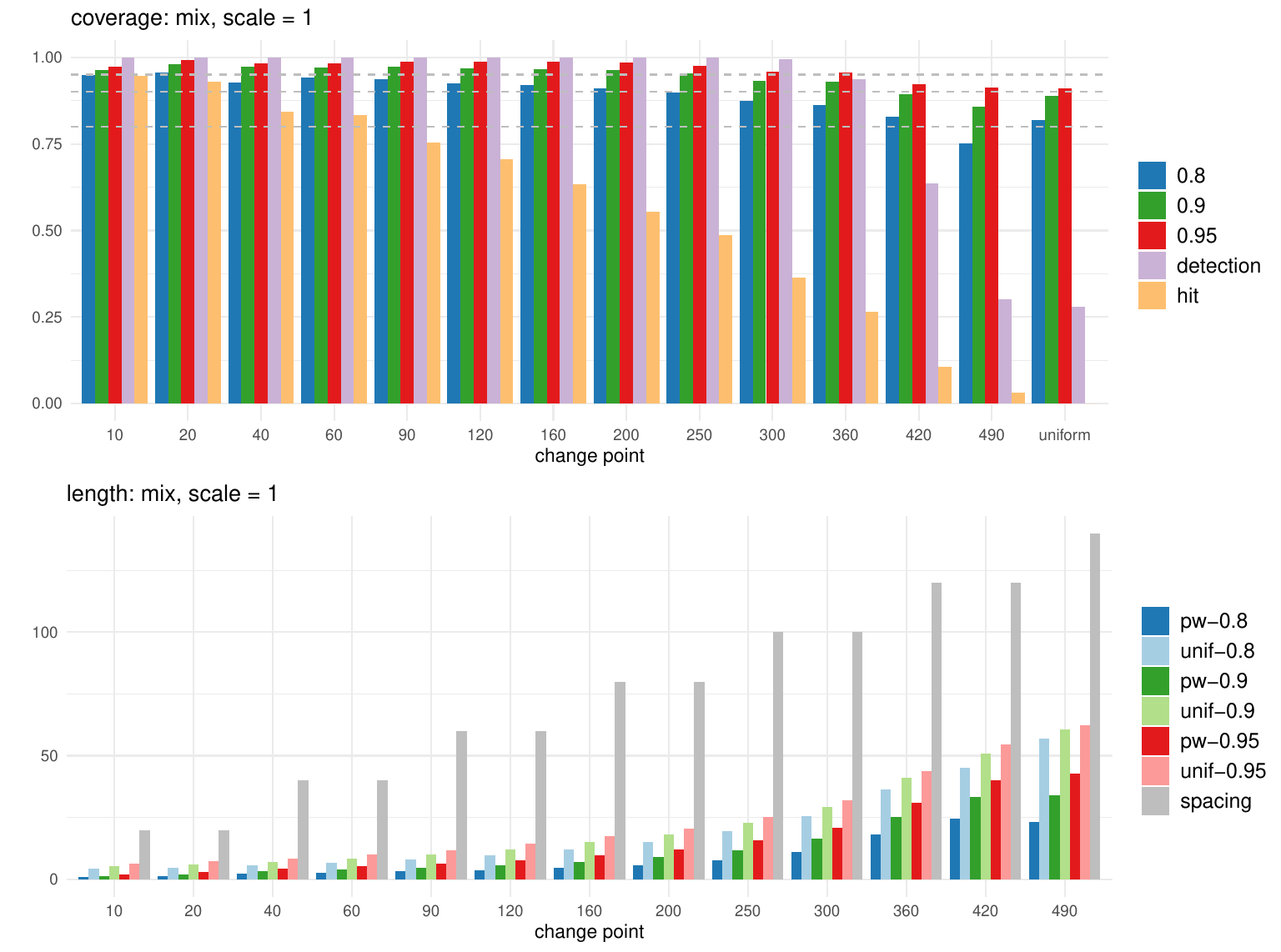}
\subcaption{ $\vartheta = 1$.}
\end{subfigure}
\begin{subfigure}{\textwidth}
	\centering
\includegraphics[height=.39\textheight]{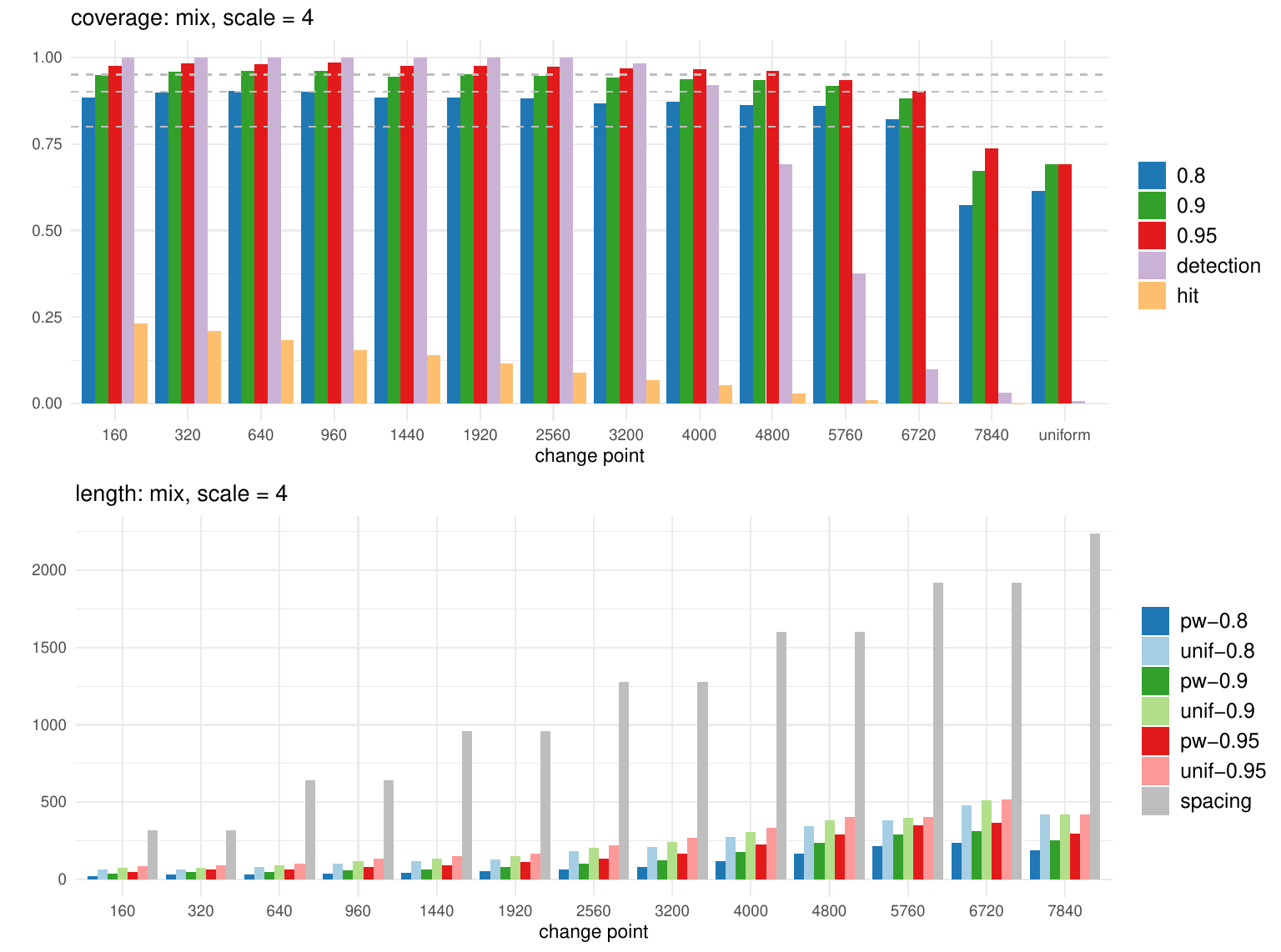}
\subcaption{$\vartheta = 4$.}
\end{subfigure}
\caption{{\tt mix}: Bootstrap CIs constructed with model selection.
See Figure~\ref{fig:full:teeth10} for detailed descriptions.}
\label{fig:full:mix}
\end{figure}

We examine the performance of bootstrap CIs
when applied with the change point estimators from the two-stage change point detection procedure proposed in \cite{cho2019two}: 
Termed MoLP, it first generates candidate change point estimators
using the MOSUM procedure with a range of bandwidths
(which includes asymmetric bandwidths discussed in Section~\ref{sec:asymm}),
and then prunes down the set of candidates to obtain the final estimators 
via the localised pruning methodology proposed therein,
see Appendix~\ref{sec:mosum} for further details.
The implementation of this two-stage procedure
is readily available as the function {\tt multiscale.localPrune} 
in the R package {\tt mosum} \citep{mosum}, 
and we apply it with all the tuning parameters chosen as recommended by default;
the set of bandwidths is also obtained according to the automatic bandwidth generation
implemented therein with the minimum bandwidth set at $10\vartheta$.
Here, $G_j$ denotes the bandwidth at which the change point estimator $\wh\cp_j$ is detected
which is chosen in a data-driven way and therefore often differs from $\delta_j/2$. 
For the test signals {\tt teeth10} and {\tt stairs10} with $\vartheta = 1$,
the set of bandwidths excludes the bandwidth $G_j = 5$ used in Section~\ref{sec:sim:notest},
as the minimum bandwidth coincides with the minimum spacing $\delta_j = 10$ of those signals.
Even in this adverse situation where the condition $2G_j < \delta_j$
required for Theorem~\ref{thm:bootstrap} is violated,
the proposed methodology works well.
Some preliminary numerical results 
indicate that the issue discussed in Remark~\ref{rem_argmax_boot} occurs in this situation, 
if the bootstrap estimator is calculated over the $G_j$-environment around 
each of the original change point estimators instead of the modified one used in~\eqref{eq:boot:max}.

Unlike in Section~\ref{sec:sim:notest}, 
the set of estimators $\wh\Cp$ returned by MoLP
may not contain the estimators for all $\cp_j, \, j = 1, \ldots, q_n$, in practice.
Therefore, we match each change point $\cp_j$ with an estimator $\wh\cp \in \wh\Cp$ as follows:
If $\wh\Cp \cap \mc I_j \ne \emptyset$
with $\mc I_j = \{\lfloor (\cp_{j - 1} + \cp_j)/2 \rfloor + 1, \ldots, \lfloor (\cp_j + \cp_{j + 1})/2 \rfloor \}$,
we regard that $\cp_j$ has been detected, and set an indicator $Z_j = 1$; 
otherwise, we set $Z_j = 0$.
If there are multiple estimators falling into the set $\mc I_j$,
we set the one closest to $\cp_j$ as its estimator $\wh\cp_j$.
Then, the coverage of the pointwise CIs is calculated as
the proportion of realisations where $\mc C^{\text{pw}}_j$ contains $\cp_j$
conditional on $Z_j = 1$, for individual $j = 1, \ldots, q_n$.
The coverage of the uniform CIs is calculated as 
that of $\mc C^{\text{unif}}_j$ containing $\cp_j$, 
conditional on $\vert \wh\Cp \vert = q_n$ and $Z_j = 1$ simultaneously for all $j = 1, \ldots, q_n$;
the lengths of the CIs are also calculated by taking average conditional on the
detection of the corresponding change points.
Figures~\ref{fig:full:teeth10}--\ref{fig:full:mix} plot the results
from the {\tt teeth10} and {\tt mix} test signals when $\vartheta \in \{1, 4\}$.

\begin{table}[htb]
\caption{Average coverage of the bootstrap $90\%$-CIs constructed with the oracle estimators
and the estimators obtained from MoLP.
Under each `$\cp_j$', we report the coverage of the corresponding pointwise CI
and under `uniform', that of the uniform CI as described in the main text.}
We also report the proportion of realisations where individual change points are detected (by MoLP)
and where all change points are correctly detected (see the rows headed `detection').
\label{table:cov}
\centering
\resizebox{\columnwidth}{!}{
\begin{tabular}{ccc  ccc ccc ccc ccc c  c}
\toprule
test signal &	$\vartheta$ &	estimator &	$\cp_1$ &	$\cp_2$ &	$\cp_3$ &	$\cp_4$ &	$\cp_5$ &	$\cp_6$ &	$\cp_7$ &	$\cp_8$ &	$\cp_9$ &	$\cp_{10}$ &	$\cp_{11}$ &	$\cp_{12}$ &	$\cp_{13}$ &	uniform	\\
\cmidrule(lr){1-3} \cmidrule(lr){4-16} \cmidrule(lr){17-17}
{\tt mix} &	1 &	oracle &	0.956 &	0.948 &	0.95 &	0.946 &	0.926 &	0.938 &	0.922 &	0.926 &	0.928 &	0.908 &	0.934 &	0.922 &	0.942 &	0.927	\\	
&	&	MoLP &	0.964 &	0.981 &	0.972 &	0.971 &	0.974 &	0.969 &	0.965 &	0.964 &	0.953 &	0.931 &	0.93 &	0.894 &	0.857 &	0.888	\\	
&	&	detection &	0.999 &	0.999 &	1 &	1 &	1 &	1 &	1 &	1 &	1 &	0.994 &	0.938 &	0.635 &	0.3 &	0.279	\\	
\cmidrule(lr){2-3} \cmidrule(lr){4-16} \cmidrule(lr){17-17}
&	4 &	oracle &	0.906 &	0.904 &	0.905 &	0.9 &	0.898 &	0.908 &	0.898 &	0.895 &	0.911 &	0.9 &	0.922 &	0.918 &	0.935 &	0.917	\\	
&	&	MoLP &	0.948 &	0.959 &	0.96 &	0.961 &	0.945 &	0.952 &	0.946 &	0.942 &	0.938 &	0.933 &	0.919 &	0.883 &	0.672 &	0.692	\\	
&	&	detection &	1 &	1 &	1 &	1 &	0.999 &	0.999 &	0.998 &	0.982 &	0.921 &	0.692 &	0.376 &	0.098 &	0.031 &	0.0065	\\	
\cmidrule(lr){1-3} \cmidrule(lr){4-16} \cmidrule(lr){17-17}
{\tt teeth10} &	1 &	oracle &	0.948 &	0.946 &	0.944 &	0.941 &	0.942 &	0.942 &	0.936 &	0.94 &	0.946 &	0.935 &	0.939 &	0.938 &	0.946 &	0.882	\\	
&	&	MoLP &	0.933 &	0.993 &	0.994 &	0.992 &	0.996 &	0.995 &	0.992 &	0.991 &	0.991 &	0.993 &	0.994 &	0.991 &	0.946 &	0.954	\\	
&	&	detection &	0.962 &	0.946 &	0.948 &	0.952 &	0.952 &	0.951 &	0.951 &	0.954 &	0.953 &	0.951 &	0.948 &	0.952 &	0.964 &	0.7	\\	
\cmidrule(lr){2-3} \cmidrule(lr){4-16} \cmidrule(lr){17-17}
&	4 &	oracle &	0.904 &	0.928 &	0.916 &	0.927 &	0.916 &	0.916 &	0.926 &	0.923 &	0.919 &	0.92 &	0.918 &	0.922 &	0.908 &	0.964	\\	
&	&	MoLP &	0.923 &	0.933 &	0.928 &	0.929 &	0.926 &	0.912 &	0.929 &	0.924 &	0.925 &	0.919 &	0.914 &	0.92 &	0.921 &	0.91	\\	
&	&	detection &	0.918 &	0.902 &	0.908 &	0.913 &	0.921 &	0.926 &	0.92 &	0.91 &	0.908 &	0.901 &	0.916 &	0.915 &	0.921 &	0.526	\\	
\bottomrule
\end{tabular}}
\end{table}

For those change points that are well-detected,
the coverage observed here tends to be slightly more conservative
compared to that reported in Section~\ref{sec:sim:notest}, 
which is attributed to the additional testing and conditioning,
see Table~\ref{table:cov} for an overview of the comparison.
On the other hand, for change points which are difficult to detect
(i.e.\ the test statistics in their vicinity do not exceed the theoretically motivated threshold
due to the corresponding $d_j^2\delta_j$ being small),
the coverage is poor.
Compare e.g.\ the coverage of the last change point $\cp_{13}$ in the {\tt mix} test signal
reported in Figure~\ref{fig:full:mix}~(a) (resp. Figure~\ref{fig:full:mix}~(b)) 
with that in Figure~\ref{fig:mix}~(a) (resp. Figure~\ref{fig:mix}~(b));
the coverage is below the nominal level in the former
and in the latter, we observe that the corresponding CIs are wide. 
The low coverage of uniform CIs observed in Figure~\ref{fig:full:mix}~(b)
is inherited from that of the change point $\cp_{13}$. 
In fact, when $\vartheta = 4$, the events of $Z_{13} = 1$ and $\cap_{j = 1}^{13} \{Z_j = 1\}$
occur only on $3.1\%$ and $0.65\%$ of the realisations, respectively.
It indicates that the change at $\cp_{13}$ is too small to be asymptotically detectable. 
Consequently, the local maxima of the MOSUM statistic are typically not significant and, 
even if they are, it can be considered spurious (false positives). 
Additionally, the non-detectability of $\cp_{13}$ results in 
the location of the local maximum in this stretch being arbitrary 
in the sense of not being sufficiently close to $\cp_{13}$ in general 
(recall the generous criterion used to determine whether change points are detected). 
This example emphasises that the CIs are valid only conditional on 
actually having the change points detected, and 
are no substitute for the uncertainty quantification 
related to whether the change point estimators are spurious or not.

\subsubsection{Comparison with SMUCE and NSP}
\label{sec:sim:comp}

\begin{table}[htb]
\caption{Coverage measures, mean and median length of the intervals and execution time (in seconds)
returned by MoLP, NSP and SMUCE, averaged over $1000$ realisations.}
\label{table:nsp}
\centering
\setlength{\tabcolsep}{2pt}
{\small \begin{tabular}{ll ccc ccc ccc}
\toprule
&	$1 - \alpha$ &	\multicolumn{3}{c}{$0.8$} &			\multicolumn{3}{c}{$0.9$} &			\multicolumn{3}{c}{$0.95$} 			\\	\cmidrule(lr){3-5} \cmidrule(lr){6-8} \cmidrule(lr){9-11}
test signal &	&	MoLP &	NSP &	SMUCE &	MoLP &	NSP &	SMUCE &	MoLP &	NSP &	SMUCE	\\	\cmidrule(lr){1-2} \cmidrule(lr){3-5} \cmidrule(lr){6-8} \cmidrule(lr){9-11}
{\tt blocks} &	CM1 &	0.842 &	0.927 &	0.542 &	0.874 &	0.964 &	0.494 &	0.889 &	0.983 &	0.456	\\	
&	CM2 & 	0.535 &	0.023 &	0.006 &	0.563 &	0.012 &	0.002 &	0.581 &	0.009 &	0.001	\\	
&	mean &	31.241 &	66.151 &	63.999 &	34.053 &	70.712 &	70.58 &	36.138 &	73.771 &	73.080	\\	
&	median &	23.15 &	52.656 &	44.968 &	25.296 &	57.457 &	50.493 &	27.18 &	61.14 &	54.368	\\	
& 	time &	0.115 &	14.448 &	0.047 &	0.114 &	14.444 &	0.048 &	0.115 &	14.336 &	0.048	\\	\cmidrule(lr){1-2} \cmidrule(lr){3-5} \cmidrule(lr){6-8} \cmidrule(lr){9-11}
{\tt fms} &	CM1 &	0.785 &	0.923 &	0.587 &	0.868 &	0.959 &	0.535 &	0.892 &	0.982 &	0.491	\\	
&	CM2 & 	0.735 &	0.475 &	0.441 &	0.804 &	0.374 &	0.267 &	0.825 &	0.302 &	0.162	\\	
&	mean &	13.381 &	28.498 &	25.82 &	15.952 &	31.279 &	30.161 &	17.653 &	33.145 &	33.015	\\	
&	median &	8.227 &	15.07 &	14.09 &	10.498 &	16.142 &	15.99 &	12.224 &	17.118 &	16.988	\\	
& 	time &	0.035 &	3.483 &	0.015 &	0.035 &	3.544 &	0.015 &	0.035 &	3.571 &	0.016	\\	\cmidrule(lr){1-2} \cmidrule(lr){3-5} \cmidrule(lr){6-8} \cmidrule(lr){9-11}
{\tt mix} &	CM1 &	0.727 &	0.967 &	0.329 &	0.823 &	0.986 &	0.261 &	0.866 &	0.997 &	0.184	\\	
&	CM2 & 	0.189 &	0.023 &	0.004 &	0.22 &	0.017 &	0.001 &	0.232 &	0.006 &	0.001	\\	
&	mean &	16.243 &	25.781 &	22.742 &	18.733 &	27.288 &	23.825 &	20.728 &	28.405 &	24.209	\\	
&	median &	10.528 &	18.95 &	17.311 &	13.146 &	20.401 &	18.911 &	15.379 &	21.876 &	19.322	\\	
& 	time &	0.058 &	3.786 &	0.017 &	0.059 &	3.818 &	0.017 &	0.06 &	3.824 &	0.018	\\	\cmidrule(lr){1-2} \cmidrule(lr){3-5} \cmidrule(lr){6-8} \cmidrule(lr){9-11}
{\tt teeth10} &	CM1 &	0.713 &	0.98 &	0.311 &	0.826 &	0.989 &	0.44 &	0.887 &	0.995 &	0.560	\\	
&	CM2 & 	0.484 &	0.046 &	0.001 &	0.571 &	0.015 &	0 &	0.622 &	0.008 &	0.000	\\	
&	mean &	3.834 &	10.883 &	13.129 &	5.522 &	12.248 &	15.096 &	7.235 &	13.698 &	18.200	\\	
&	median &	3.303 &	10.437 &	12.263 &	4.797 &	11.928 &	14.256 &	6.445 &	13.264 &	17.601	\\	
& 	time &	0.021 &	0.778 &	0.008 &	0.021 &	0.839 &	0.008 &	0.021 &	0.889 &	0.008	\\	\cmidrule(lr){1-2} \cmidrule(lr){3-5} \cmidrule(lr){6-8} \cmidrule(lr){9-11}
{\tt stairs10} &	CM1 &	0.998 &	0.994 &	0.258 &	0.998 &	0.998 &	0.336 &	0.998 &	0.998 &	0.421	\\	
&	CM2 & 	0.986 &	0.112 &	0.054 &	0.986 &	0.066 &	0.041 &	0.986 &	0.065 &	0.048	\\	
&	mean &	17.193 &	8.918 &	7.826 &	17.819 &	9.651 &	8.274 &	18.211 &	10.329 &	8.578	\\	
&	median &	18.266 &	8.743 &	7.617 &	18.938 &	9.501 &	8.038 &	19.298 &	10.218 &	8.352	\\	
& 	time &	0.022 &	0.648 &	0.008 &	0.022 &	0.687 &	0.008 &	0.023 &	0.717 &	0.008	\\	\bottomrule
\end{tabular}}
\end{table}

We compare the proposed bootstrap procedure combined with MoLP as in Section~\ref{sec:sim:lp},
with SMUCE \citep{frick2014} and NSP \citep{fryzlewicz2020}
on the five test signals generated with $\vartheta = 1$.
SMUCE returns a confidence set for $f_t$ at a given confidence level $\alpha$
from which confidence bands around change points can be derived;
their coverage is  comparable with that of the proposed uniform bootstrap CIs,
and we provide a comparative study between these two methods in Appendix~\ref{sec:smuce}.
NSP aims at returning intervals that contain at least one change point at a prescribed level $\alpha$, 
and thus its objective is different from those of the CIs obtained by MoLP or SMUCE.

To account for this, we use two coverage measures, referred to as CM1 and CM2, for their evaluation.
Proposed in \cite{fryzlewicz2020}, CM1 reports the proportion (over $1000$ realisations) 
where all of the intervals returned by a given method contain at least one true change point.
CM2 additionally checks whether each of the true change points is covered by one of the intervals returned.
In other words, CM2 is the proportion of realisations where each change point is covered by an interval,
in addition to there being no spurious interval that does not contain a change point. 
Effectively, this measures the unconditional coverage uniformly over all change points 
and is distinguished from the conditional coverage reported in Section~\ref{sec:sim:lp}.
By construction, CM1 is always greater than or equal to CM2 
and it takes the value one even when a method does not return any intervals,
whereas CM2 is recorded as zero in such a case. 
Additionally, we report the mean and the median length of the intervals
and execution time (on a 4 GHz Intel Core i7 with 16 GB of RAM running macOS Catalina), see Table~\ref{table:nsp}.
In the case of MoLP, the execution time includes time taken by both the detection and the bootstrap procedures.

Overall, MoLP attains good coverage while taking only a fraction of a second 
to perform change point detection and generate bootstrap CIs based on $B = 2000$ bootstrap samples.
Also, the lengths of bootstrap CIs are generally shorter than
the intervals returned by NSP or SMUCE with the exception of {\tt stairs10}.
The bootstrap CIs returned by MoLP are not designed to account for spurious change point estimators, 
and they do not provide finite-sample guarantee of controlling CM1. 
However, their coverage is typically not too far below the desired level and even above it in the case of \texttt{stairs10} 
(due to the change point estimators being very precise). 
In contrast, NSP is the only method giving finite-sample guarantees of controlling CM1 and 
also the only one achieving this goal.
As expected, the coverage measures returned by MoLP increase as $\alpha$ decreases.
NSP attains CM1 close to one regardless of $\alpha$
and in the case of SMUCE, its coverage can increase with $\alpha$,
an observation also made by several other authors \citep{chen2014, fryzlewicz2020}.
According to CM2, MoLP performs best on all test signals.
NSP does not explicitly set out to detect changes and as such,
the intervals it outputs often under-detects the number of change points
in the sense that some changes points are not covered by any of those intervals,
as evidenced by the overall small values of CM2.

\subsection{Application to central England temperature}
\label{sec:real}

\begin{figure}[htb]
\centering
\includegraphics[width = .7\textwidth]{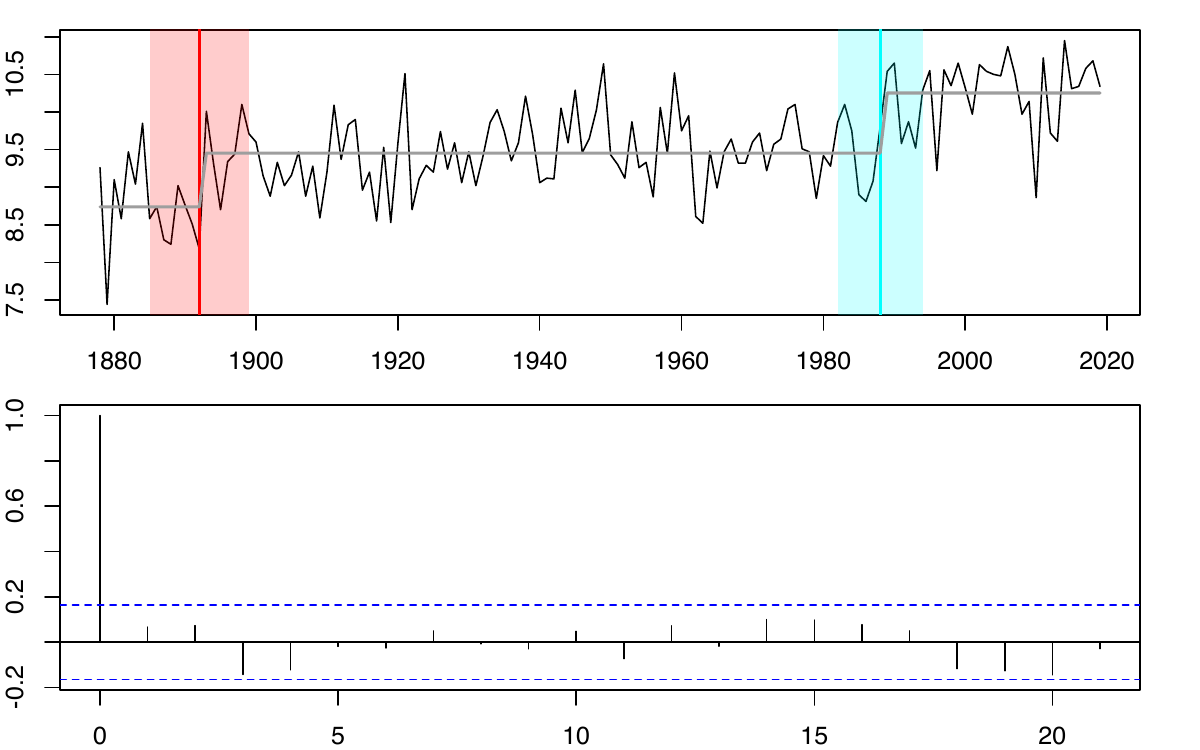}
\caption{Top: Yearly average temperatures between 1878 and 2019
plotted together with the change points estimated by MoLP (vertical lines),
shaded areas representing the $90\%$ uniform confidence intervals around the change points
and piecewise constant mean (bold lines)
Bottom: Autocorrelations of the data after the two mean shifts are removed.}
\label{fig:hadcet}
\end{figure}

The Hadley Centre central England temperature (HadCET) dataset \citep{parker1992}
contains the mean daily and monthly temperatures
representative of a roughly triangular area 
enclosed by Lancashire, London and Bristol, UK.
We analyse the yearly average of the monthly mean temperatures 
of the period of 1878--2019 ($n = 142$) for change points using MoLP
(applied with the significance level $0.2$ and other tuning parameters set at their default values)
and derive their 90\% CIs, see Table~\ref{table:hadcet}
and also the top panel of Figure~\ref{fig:hadcet} for the visualisation of the uniform CIs.
The bottom panel of Figure~\ref{fig:hadcet} shows that
after the mean shifts are accounted for, the series exhibits little autocorrelations
and thus justifies the use of the proposed bootstrap methodology.
The second change point detected at 1987/88 
coincides with the global regime shift in Earth's biophysical systems
identified around 1987 \citep{reid2016}, 
which is attributed to anthropogenic warming and a volcanic eruption.

\begin{table}[ht]
\caption{Pointwise and uniform $90\%$ CIs obtained with 
$\wh\cp_1 = 1892$ and $\wh\cp_2 = 1988$.}
\label{table:hadcet}
\centering
{\small
\begin{tabular}{c c c}
\toprule
& pointwise & uniform \\
\midrule
$\cp_1$ & $[1887, 1897]$ & $[1885, 1899]$ \\
$\cp_2$ & $[1984, 1992]$ & $[1983, 1993]$ \\
\bottomrule
\end{tabular}}
\end{table}

\section{Conclusions}\label{sec:conc}
In this paper we rigorously analyse a bootstrap method for the construction of CIs 
for the location of change points obtained from MOSUM-based procedures,
both theoretically and in simulation studies. 
We show that for local changes (i.e.\ $d_j = d_{j, n} \to 0$), 
the limit distributions of change point estimators
are continuous as a functional of a Wiener process with drift,
while those for the fixed changes remain discrete.
Such results hold under mild assumptions made in~\eqref{eq:errors}
that permit non-Gaussianity and heavy-tails,
and could even be extended to serially correlated errors as discussed in Section~\ref{sec:dependent}.
Both our theoretical investigation into the asymptotic distribution of the change point estimators 
and bootstrap CI construction, are based on the oracle estimators in~\eqref{eq:cp:tilde} and~\eqref{eq:boot:max},
with the former tied to MOSUM-based change point estimators
as noted in Assumption~\ref{assum_meta_est}.
To the best of our knowledge, in the fixed change case,
there was no proof of the validity of the bootstrap procedure in the literature
even for the AMOC setting with CUSUM-based estimators.
The proof requires non-standard steps 
due to the non-Gaussianity of the limit distribution of the change point estimators. 
Despite the distinct behaviour of the limit distributions for fixed and local changes, 
the proposed bootstrap procedure adapts to the magnitude of the change without knowing 
the regime it belongs to. 

Numerical studies show that the bootstrap works well with the oracle estimators,
and only slightly more conservative when constructed with the estimators
which are obtained with an additional model selection step
(and thus accessible in practice).
The results suggest that the bootstrap CIs behave well 
in the presence of heavy-tailed errors,
and also in comparison with competing methodologies.
Both the implementation of the proposed bootstrap methodology
as well as the multiple change point detection procedure 
adopted for simulation studies (MoLP) are available in the R package {\tt mosum} \citep{mosum}.

\bibliographystyle{apalike}
\bibliography{fbib}

\clearpage

\appendix

\section{Proofs}
\label{sec:proofs}

\subsection{Proof of Theorem~\ref{thm:asymp}}
\label{pf:thm:asymp}


Assumptions~A.1~b) and~c) of \cite{eichinger2018} are fulfilled 
in the situation considered here; see \cite{komlos1975, komlos1976} and \cite{major1976approximation} 
for the invariance principle (although for the proof of this theorem, the functional central limit theorem is sufficient), and Theorem 3.7.8 of \cite{stout1974almost}, 
for the moment assumption on the sums of the errors.	
First, note that 
\begin{align*}
\wt\cp_j = \arg\max_{k: \, \vert k - \cp_j \vert \le G_j} V_{k, n}(G_j),
\quad \text{where} \quad V_{k,n}(G_j)=(T_{k,n}(G_j))^2-(T_{\cp_j,n}(G_j))^2.
\end{align*}
The proof of Theorem~3.2 in \cite{eichinger2018} shows that for any $c>1$
\begin{align}
& \p\left( |\wt\cp_j-\cp_j| > c \sigma^2 d_j^{-2} \right)
\le \p\left(\max_{|k-\cp_j| > c \sigma^2 d_j^{-2}}  V_{k,n}(G_j) \ge 
\max_{|k-\cp_j|\le c \sigma^2 d_j^{-2}}  V_{k,n}(G_j) \right)
\nn \\
& \le O(c^{-1} )+o(1).
\label{eq:tilde:consist}
\end{align}

Therefore,
\begin{align*}
& \p\left(\frac{d_{j}^2 (\wt\cp_j - \cp_j)}{\sigma^2} \le x \right)=
\p\left(-c \le \frac{d_{j}^2 (\wt\cp_j -\cp_j)}{\sigma^2} \le x\right) + 
O(c^{-1}) + o(1).
\end{align*}
Furthermore by Lemma~5.2  in \cite{eichinger2018} and the decomposition of $V_{k,n}(G_j)$ as given in their~(5.8), it holds for any $k$ satisfying $-c \le \sigma^{-2} d_j^2( k-\cp_j) < 0$,
\begin{align}
\label{eq:vk}
V_{k,n}(G_j) &= - d_{j}^2|\cp_j - k| - d_{j}\left(\sum_{t = k-G_j+1}^{\cp_j-G_j} \vep_t - 2\sum_{t = k + 1}^{\cp_j} \vep_t + \sum_{t = k+G_j+1}^{\cp_j+G_j} \vep_t
 \right) + o_P(1)
\\
&=: \wt V_{k, n}(G_j) + o_P(1), \nn
\end{align}
where the $o_P(1)$ term is uniform over $\cp_j - G_j \le k < \cp_j$, and
\begin{align}\label{eq_Unl}
& \left\{ d_{j}\left(\sum_{t = k-G_j+1}^{\cp_j-G_j}\vep_t - 
2\sum_{t = k+1}^{\cp_j} \vep_t
+ \sum_{t = k+G_j+1}^{\cp_j+G_j} \vep_t \right) : \, k = \cp_j - 1, \ldots, \cp_j - c\sigma^2d_{j}^{-2} \right\}
\\
&\overset{\mathcal{D}}{=} \left\{U_n(\ell) = d_{j} 
\left(\sum_{t = \ell}^{-1}\vep_t^{(1)} - 2\sum_{t=\ell}^{-1}\vep_t^{(2)} +
\sum_{t=\ell}^{-1} \vep_t^{(3)} \right) : \, 
\ell = -1, \ldots, -c\sigma^2 d_{j}^{-2} \right\}.\notag
\end{align}
Analogous assertions hold for $\cp_j < k \le \cp_j+G_j$.

So far, the preceding arguments hold for both the local and the fixed change cases.
Now, the proof for the local change ($d_j = d_{j, n} \to 0$) in~\ref{thm:asymp:one} 
is concluded as in the proof of Theorem~3.3 in \cite{eichinger2018} 
by making use of the following functional central limit theorem
\begin{align*}
	\left\{\frac{U_n(\lfloor s\sigma^2d_{j}^{-2}\rfloor)}{\sigma^2}: \, 
	-c\le s\le c\right\} \overset{D[-c,c]}{\longrightarrow} 
	\{\sqrt{6}\, W(s), \, -c \le s \le c\}.
\end{align*}
We elaborate on the proof here in order to highlight the difference 
between the local and the fixed change cases. 
Note that
\begin{align*}
	&\p\left( -c \ls \frac{d_{j}^2 (\wt\cp_j - \cp_j)}{\sigma^2} \ls x\right)\\
	&= \p\left( \max_{-c \ls d_j^2(k-\cp_j)/\sigma^2\ls x }\wt V_{k,n}(G_j)
	\ge  \max_{x < d_j^2(k-\cp_j)/\sigma^2\ls c } \wt V_{k,n}(G_j)
+o_P(1) \right).
\end{align*}
Thus, for any $\eta>0$, we obtain

\begin{align}
	&\p\left(-c \ls \frac{d_{j}^2(\wt\cp_j - \cp_j)}{\sigma^2}\ls x \right) 
	\nn \\
	&\le\p\left(\max_{-c \ls d_j^2(k-\cp_j)/\sigma^2\ls x }\wt V_{k,n}(G_j)\ge  \max_{x < d_j^2(k-\cp_j)/\sigma^2\ls c }	\wt V_{k,n}(G_j)-\eta\right) + o(1)
	 \label{eq:upper:bound}\\
	 &\to  \p\left(
	 \max_{-c\ls s\ls x}\left(-|s|-\sqrt{6} W(s)\right)\gs \max_{x<s\ls c}\left(-|s|-\sqrt{6} W(s)\right)-\eta\right)
	 \nn
\end{align}
as well as
\begin{align}
	&\p\left(-c\ls  d_{j}^2\frac{\wt\cp_j - \cp_j}{\sigma^2}\ls x\right) 
	\nn \\
	&\ge\p\left(  \max_{-c \ls d_j^2(k-\cp_j)/\sigma^2\ls x }
	\wt V_{k,n}(G_j)
	\ge \max_{x < d_j^2(k-\cp_j)/\sigma^2\ls c }\wt V_{k,n}(G_j)
	+\eta\right) +o(1)
	 \label{eq:lower:bound} \\
	 &\to  \p\left(\max_{-c\ls s\ls x}\left(-|s|-\sqrt{6} W(s)\right)\gs \max_{x<s\ls c}\left(-|s|-\sqrt{6} W(s)\right)+\eta\right).
	 \nn
\end{align}
Since the maximum of a Wiener process with drift has a continuous distribution, 
the limits of both the upper and lower bounds on~\eqref{eq:upper:bound}--\eqref{eq:lower:bound} 
coincide on letting $\eta \to 0$ with
\begin{align*}
\p\left(-c\ls {\arg\max}_{-c\ls s\ls c} \left(W(s)-|s|/\sqrt{6}\right)\ls x\right),
\end{align*}
such that the results follows by letting $c\to\infty$.

The proof of the fixed change case proceeds similarly,
apart from that $U_n(\ell)$ already coincides with the limit distribution 
such that no additional functional central limit theorem is necessary. 
Hence, the upper and lower bounds in~\eqref{eq:upper:bound}--\eqref{eq:lower:bound} 
coincide as $\eta \to 0$ as long as the maximum over $U_n(\ell)$ (plus drift) 
has a continuous distribution, which in turn holds provided that $\{\vep_t\}$ have a continuous distribution. 
On the other hand, for discrete distributions where ties in the maximum can occur with positive probability, 
the two bounds are not guaranteed to converge. 
In fact, when ties in the maximum exist, the $\arg\max$ 
over the corresponding functional of $\{\wt V_{k, n}(G_j)\}$ 
(always picking the first of the two maximisers)
can differ from the $\arg\max$ over the functional of $\{V_{k, n}(G_j)\}$,
since the $o_P(1)$ term can make the second maximum strictly larger than the first for all $n$ 
(with the two coinciding in the limit).


The assertion for~\ref{thm:asymp:three} follows immediately 
from \ref{thm:asymp:one} and \ref{thm:asymp:two} 
by noting that the maximisers involve different errors 
provided that $4G_j < \delta_j$ for all $j = 1, \ldots, q$,
and thus are independent.

\subsection{Proof of Theorem~\ref{thm:bootstrap}}
\label{pf:thm:bootstrap}

In this section, the notations $o_P$ and $O_P$ are reserved
for the functionals of the original observations $X_1, \ldots, X_n$,
which are deterministic given those observations.
Also, in order to facilitate the proofs below, 
we always condition on the following set
\begin{align}
	\mathcal{M}_j = \mc M_{j, n} = &\left\{  \wh\cp_{j - 1} > \cp_{j - 2}, \quad \wh\cp_{j + 1} < \cp_{j + 2}, 
	\quad \vert \wh\cp_{j} - \cp_{j} \vert < \frac{1}{3} \delta_j
	\right. \nn\\
	& \left. \quad \text{and} \quad 
	\vert \wh\cp_{i} - \cp_{i} \vert < \frac{1}{3} \vert \cp_j - \cp_{i}\vert \text{ for } i = j - 1, j + 1, 
	\right\}.
	\label{eq_M_set}
\end{align}
Under Assumption~\ref{assum_precision_est}, it holds that $\p(\mc M_{j, n}) \to 1$ for each $j$.
In addition, we use the notations $\E^*(\cdot) = \E(\cdot \vert X_1, \ldots, X_n)$
and $\Var^*(\cdot) = \Var(\cdot \vert X_1, \ldots, X_n)$.

\subsubsection{Auxiliary lemmas}

In what follows, we denote for $\wh\cp_{j-1} < t \le \wh\cp_j$,
\begin{align}\label{eq_eps_star}
	\vep_t^* = X_t^*-\E^*(X_t^*).
\end{align}

\begin{lem}
\label{lemma_variance_boot}
Let~\eqref{eq:model}--\eqref{eq:errors} 
and Assumption~\ref{assum_precision_est} hold for a given $j$.
\begin{enumerate}[label = (\alph*)]
\item \label{lemma_variance_boot:one} 
It holds 
\begin{align*}
	&\E^*(X^*_t) = \begin{cases}
\bar{X}_{\wh\cp_{j-1}, \wh\cp_j}= \E(X_{\cp_j})  + o_P(|d_j|) & \text{for } \wh\cp_{j-1} < t \le \wh\cp_j, \\
\bar{X}_{\wh\cp_{j}, \wh\cp_{j+1}}= \E(X_{\cp_{j+1}})  + o_P(|d_j|) & \text{for } \wh\cp_j< t \le \wh\cp_{j+1}.
	\end{cases}
\end{align*}
In particular, $d_j^* = \wh d_j = \bar{X}_{\wh\cp_{j}, \wh\cp_{j+1}} - \bar{X}_{\wh\cp_{j-1}, \wh\cp_j}
= d_j + o_P(|d_j|)$.
%
%
\item \label{lemma_variance_boot:two} 
It holds that
$(\sigma^*_t)^2 = \Var^*(X^*_t)=\Var^*(\vep^*_t) = \sigma^2 + o_P(1)$ 
with $(\sigma^*_t)^2$ being constant over each segment $\wh\cp_{j-1} < t \le \wh\cp_j$.
Then, the $o_P(1)$ term is uniform as long as only finitely many segments are involved.
\end{enumerate}
\end{lem}

\begin{proof}
For the proof of~\ref{lemma_variance_boot:one},
note that from the H\'{a}jek-R\'{e}nyi inequality for i.i.d.\ random variables, it holds
\begin{align}
\label{eq:hr:bound}
	&\max_{1\le k \le \frac{4}{3}(\cp_j-\cp_{j-1})} \left\vert \sum_{t = \cp_j-k}^{\cp_j}\vep_t \right\vert
	 =O_P(\sqrt{\cp_j-\cp_{j-1}}), \quad  
	 \max_{1\le k < \frac{1}{3}(\cp_j-\cp_{j-1})} \left\vert \sum_{t = \cp_j + 1}^{\cp_j+k} \vep_t\right\vert
	 =O_P(\sqrt{\cp_j-\cp_{j-1}}).
\end{align}
Then on $\mathcal{M}_j$ defined in~\eqref{eq_M_set},
the following decomposition holds
\begin{align}
\E^*(X_t^*) = \bar{X}_{\wh\cp_{j-1}, \wh\cp_j} &=  
\E(X_{\cp_j}) + \frac{1}{\wh\cp_j - \wh\cp_{j - 1}}
\sum_{\wh\cp_{j-1} + 1}^{\wh\cp_j} \vep_t \nn \\
& \qquad 
- \frac{d_{j-1} (\cp_{j-1} - \wh\cp_{j-1})\mathbb{I}_{\{\cp_{j - 1} > \wh\cp_{j - 1}\}}}
{\wh\cp_j - \wh\cp_{j-1}}
 + \frac{d_j (\wh\cp_j - \cp_j)\mathbb{I}_{\{\wh\cp_j > \cp_j\}}}
{\wh\cp_j - \wh\cp_{j-1}} \nn \\
&
= \mu_{j} + O_P( (\cp_j-\cp_{j-1})^{-1/2})+o_P( |d_j|)=\mu_j+o_P(|d_j|),
\label{eq:boot:mean}
\end{align}
from the bounds in~\eqref{eq:hr:bound}, 
Assumption~\ref{assum_precision_est} and that $d_j^2(\cp_j-\cp_{j-1}) \to \infty$. 
Analogous assertions hold when $t > \wh\cp_j$, 
completing the proof of~\ref{lemma_variance_boot:one}.
Consequently,  with $o_P(1)$ uniformly in $\wh\cp_{j-1} < t \le \wh\cp_j$
\begin{align}
\label{eq:boot:vep:decomp}
	\wh\vep_t = X_t - \bar{X}_{\wh \cp_{j-1}, \wh\cp_j} = \vep_t + o_P(1) - d_{j-1} \mathbb{I}_{\{\wh\cp_{j-1} < t \le \cp_{j-1}\}}
+ d_j \mathbb{I}_{\{\cp_j < t \le \wh\cp_j\}}.
\end{align}

For the proof of~\ref{lemma_variance_boot:two}, 
note that for $\wh\cp_{j-1} < t \le \wh\cp_j$, 
\begin{align*}
\Var^*(\vep^*_t) & = \frac{1}{\wh\cp_j - \wh\cp_{j - 1}}
\sum_{t = \wh\cp_{j-1} + 1}^{\wh\cp_j} \l(X_t - \bar{X}_{\wh\cp_{j-1}, \wh\cp_j}\r)^2
= \frac{1}{\wh\cp_j - \wh\cp_{j-1}} \sum_{t = \wh\cp_{j-1} + 1}^{\wh\cp_j} \vep_t^2 + 
o_P(1) = \sigma^2+o_P(1), 
\end{align*}
by Assumption~\ref{assum_precision_est} and
the law of large numbers, 
where similar arguments as those adopted in~\eqref{eq:boot:mean} 
have been used for the terms including the indicators.
\end{proof}

The following lemma is the bootstrap analogue to Lemma~5.2 of \cite{eichinger2018}.
\begin{lem}
\label{lemma_52_boot}
Let the assumptions in Theorem~\ref{thm:bootstrap} hold for a given $j$. 
Then for any sequences $\beta_n > 0$ and $\xi_n \ge 1$, it holds:
\begin{enumerate}[label = (\alph*)]
\item \label{lemma_52_boot:one} $\displaystyle 
\p^*\l( \max_{\wh\cp_j - G_j \le k\le \wh\cp_j - \xi_n} 
\frac{\l\vert T_{\wh\cp_j, n}(G_j; \vep^*) - T_{k, n}(G_j; \vep^*) \r\vert}{\wh\cp_j - k}
> \beta_n (1+o_P(1)) \r) = O_P\l(\l( \beta_n^2 G_j \xi_n \r)^{-1}\r)$.
\item \label{lemma_52_boot:two} $\displaystyle 
\p^*\l( \max_{\wh\cp_j - \xi_n \le k \le \wh\cp_j} 
\l\vert T_{\wh\cp_j, n}(G_j; \vep^*) - T_{k, n}(G_j; \vep^*)\r\vert > \beta_n (1+o_P(1))\r) =
O_P\l(\beta_n^{-2}\, \frac{\xi_n}{G_j} \r)$.
\item  \label{lemma_52_boot:three}
$\displaystyle \p^*\l(   
\max_{\wh\cp_j - G_j \le k\le \wh\cp_j - \xi_n} \l\vert 
T_{\wh\cp_j, n}(G_j; \vep^*) + T_{k, n}(G_j; \vep^*) \r\vert > \beta_n (1+o_P(1))\r) = O_P\l(\beta_n^{-2} \r)$.
\end{enumerate}
\end{lem}

\begin{proof}
Some straightforward calculations show that for $k \le \wh\cp_j$, we have
\begin{align}
\label{eq:boot:vep:mosum}
&T_{\wh\cp_j, n}(G_j; \vep^*) - T_{k, n}(G_j; \vep^*) = 
\frac{1}{\sqrt{2G_j}}\left(\sum_{t = k+G_j+1}^{\wh\cp_j+G_j} \vep^*_t
+ \sum_{t = k-G_j+1}^{\wh\cp_j-G_j} \vep^*_t
- 2 \sum_{t = k+1}^{\wh\cp_j}\vep^*_t \right).
\end{align}
By the H\'{a}jek-R\'{e}nyi inequality for i.i.d.\ random variables
and Lemmas~\ref{lemma_variance_boot}~\ref{lemma_variance_boot:two},
the first summand in the RHS of~\eqref{eq:boot:vep:mosum} satisfies
\begin{align*}
&	\p^*\l(\max_{\wh\cp_j - G_j \le k\le \wh\cp_j - \xi_n}
\l\vert\frac{\sum_{t = k + G_j + 1}^{\wh\cp_j + G_j} \vep^*_t}
{\wh\cp_j - k} \r\vert>\beta_n (1+o_P(1)) \sqrt{G_j}\r)
= O_P(1) \frac{1}{\beta_n^2G_j\xi_n} \Var^*(\vep^*_{\wh\cp_j + 1})
\\
& = O_P\l( (\beta_n^2G_j\xi_n)^{-1} \r),
\end{align*}
as well as
\begin{align*}
&	\p^*\l(\max_{\wh\cp_j - \xi_n \le k \le \wh\cp_j} 
\l\vert\sum_{t = k + G_j + 1}^{\wh\cp_j + G_j} \vep^*_t \r\vert>\beta_n (1+o_P(1)) \sqrt{G_j}\r)
= O_P(1) \frac{\xi_n}{\beta_n^2G_j} \Var^*(\vep^*_{\wh\cp_j + 1})
= O_P\l(\beta_n^{-2}\, \frac{\xi_n}{G_j} \r).
\end{align*}
Analogous assertions hold for the other two summands in~\eqref{eq:boot:vep:mosum},
which lead to \ref{lemma_52_boot:one} and~\ref{lemma_52_boot:two}. 
As for~\ref{lemma_52_boot:three}, 
noting that $T_{\wh\cp_j, n}(G_j; \vep^*) + T_{k, n}(G_j; \vep^*) =
2 T_{\wh\cp_j, n}(G_j; \vep^*) + T_{k, n}(G_j; \vep^*) - T_{\wh\cp_j, n}(G_j; \vep^*)$,
the arguments adopted in the proof of~\ref{lemma_52_boot:two}
and Chebyshev's inequality lead to the conclusion.
\end{proof}

\begin{lem}\label{lem_fixed_changes}
\label{lem:next} Let~\eqref{eq:model}--\eqref{eq:errors} hold
and suppose that the change point estimators satisfy 
Assumption~\ref{assum_precision_est} for a given $j$.
Then, it holds for $\wh\cp_{j-1} < t \le \wh\cp_{j+1}$,
\begin{align*}
	&\sup_{x \in \R} \left\vert \p^*(\vep_t^*\ls x) - \p(\vep_1\ls x) \right\vert \pto 0.
\end{align*}
\end{lem}

\begin{proof}
Denote $\wh{F}_{a, b}(Z_t; x) = (b-a)^{-1} \sum_{t = a+1}^b \mathbb{I}_{\{Z_t \le x\}}$ 
and $F(x) = \p(\vep_1 \le x)$. 
Then for $\wh\cp_{j-1} < t \le \wh\cp_{j}$,
the following decomposition holds on $\mathcal{M}_j$ defined in~\eqref{eq_M_set}:
\begin{align*}
& \l\vert \wh{F}_{\wh\cp_{j-1}, \wh\cp_j}\l( \vep_t - d_{j-1} \mathbb{I}_{\{\wh\cp_{j - 1} < t \le \cp_{j-1}\}}
+ d_j \mathbb{I}_{\{\cp_j < t \le \wh\cp_j\}}; x \r) - F(x) \r\vert
\\
&	\le \l\vert \wh{F}_{\wh\cp_{j-1},\wh\cp_j}\l( \vep_t - d_{j - 1} \mathbb{I}_{\{\wh\cp_{j - 1} < t \le \cp_{j-1} \}}
+ d_j \mathbb{I}_{\{\cp_j < t \le \wh\cp_j\}}; x \r) - \wh{F}_{{\wh\cp_{j-1}, \wh\cp_j}}(\vep_t; x)\r\vert
\\
& + \l\vert \wh{F}_{{\wh\cp_{j-1},\wh\cp_j}}(\vep_t; x)
- \wh{F}_{{\cp_{j-1},\cp_j}}(\vep_t; x) \r\vert
+ \l\vert \wh{F}_{{\cp_{j-1}, \cp_j}}(\vep_t; x) - F(x)\r\vert
=: D_1(x) + D_2(x) + o_P(1),
\end{align*}
where the $o_P(1)$ holds uniformly in $x$ due to the Glivenko-Cantelli theorem. 
Furthermore, 
\begin{align*}
& \sup_x |D_1(x)| \le 
\frac{1}{\wh\cp_j - \wh\cp_{j-1}}
\l( \vert \wh\cp_{j - 1} - \cp_{j - 1} \vert + \vert \wh\cp_j - \cp_j \vert \r) = o_P(1)
\end{align*}
by Assumption~\ref{assum_precision_est}.
Similarly, uniformly in $x$, we have
\begin{align*}
&|D_2(x)| \\
&\le  \l\vert \frac{1}{\cp_j - \cp_{j-1}} - \frac{1}{\wh\cp_j - \wh\cp_{j-1}} \r\vert\;
\l\vert \sum_{t =\cp_{j-1} + 1}^{\cp_j} \mathbb{I}_{\{\vep_t \le x\}} \r\vert +
 \frac{1}{\wh\cp_j - \wh\cp_{j-1}} \l\vert 
\sum_{t = \cp_{j-1} + 1}^{\cp_j} \mathbb{I}_{\{\vep_t \le x\}}
- \sum_{t = \wh\cp_{j-1}+1}^{\wh\cp_j} \mathbb{I}_{\{\vep_t \le x\}} \r\vert
\\
&\le \frac{2}{\wh\cp_j - \wh\cp_{j-1}} 
\l( \vert \wh\cp_{j-1} - \cp_{j-1} \vert + \vert \wh\cp_j - \cp_j\vert \r) = o_P(1).
\end{align*}
Analogous assertions hold for $t > \wh\cp_j$ which concludes the proof.
\end{proof}

\begin{proof}[Proof of Theorem~\ref{thm:bootstrap}]
The proof proceeds analogously as the proof of Theorem~\ref{thm:asymp},
replacing the sample quantities therein with their bootstrap counterparts 
and Lemma~5.2 of \cite{eichinger2018} with Lemma~\ref{lemma_52_boot}.
Lemma~\ref{lemma_variance_boot}~\ref{lemma_variance_boot:two} 
and Lemma~5.1 in \cite{huvskova2008} 
replace the standard functional central limit theorem 
adopted in the proof of Theorem~\ref{thm:asymp}~\ref{thm:asymp:one} where $d_j = d_{j, n} \to 0$, 
while Lemma~\ref{lem_fixed_changes} is required to deal with the fixed change situation.

To elaborate, consider $V_{k,n}^*(G_j) = (T_{k,n}(G_j; X^*))^2 - (T_{\cp_j,n}(G_j; X^*))^2$. 
Then, standard arguments analogous to those adopted in the proof of Theorem~3.2 in \cite{eichinger2018} yield
\begin{align*}
	& \p^*\left( |\wt\cp_j^*-\wh\cp_j| > c \sigma^2 d_j^{-2} \right)
\le \p^*\left(\max_{|k-\wh\cp_j| > c \sigma^2 d_j^{-2}}  V^*_{k,n}(G_j) \ge 
\max_{|k-\wh\cp_j|\le c \sigma^2 d_j^{-2}}  V^*_{k,n}(G_j) \right)\\
&\le O_P(c^{-1} )+o_P(1).
\end{align*}
This follows from Lemma~\ref{lemma_52_boot} 
where the additional $1 + o_P(1)$ factor therein is needed to account for 
 $\wh d_j = d_j (1+o_P(1))$, which follows from 
Lemma~\ref{lemma_variance_boot}~\ref{lemma_variance_boot:one}. 
Additionally, 
for $-c \le \sigma^{-2} d_j^2( k-\wh\cp_j) < 0$,
\begin{align*}
	V^*_{k,n}(G_j) &=-d_{j}^2(1+o_P(1))|\wh\cp_j-k| \\
	& \quad - d_{j} (1+o_P(1))\,\left(
	\sum_{t=k-G_j+1}^{\wh\cp_j-G_j}\vep^*_t -
	2\sum_{t=k+1}^{\wh\cp_j}\vep^*_t +
	\sum_{t=k+G_j+1}^{\wh\cp_j+G_j}\vep^*_t \right) +{R^{*(1)}_n}(k)\\
	&=- d_{j}^2|\wh\cp_j-k| - d_{j} \,\left(
	\sum_{t=k-G_j+1}^{\wh\cp_j-G_j}\vep^*_t -
	2\sum_{t=k+1}^{\wh\cp_j}\vep^*_t + 
	\sum_{t=k+G_j+1}^{\wh\cp_j+G_j}\vep^*_t \right) +{R^{*(1)}_n}(k)+{R^{*(2)}_n}(k)\\
	&=:\wt V_{k,n}^*(G_j)+R_n^*(k).
\end{align*}
It holds by arguments analogous to those in the proof of Theorem~\ref{thm:asymp} 
(which make use of the decomposition in Equation~(5.8) of \cite{eichinger2018}) for ${R^{*(1)}_n}(k)$, 
and from the (conditional) stochastic boundedness of $U_n^*(\ell)$ defined below for ${R^{*(2)}_n}(k)$,
that for any $\tau>0$, 
\begin{align*}
	\p^*\left( \sup_{ k: \, \vert k - \wh\cp_j \vert \le c\sigma^2 d_j^{-2}}|R_n^*(k)| \ge \tau \right) = o_P(1).
\end{align*}

On $\mathcal{M}_j$ and for $n$ large enough such that $c \sigma^2d_j^{-2} < G_j$, we have
\begin{align*}
	&\left\{ d_{j}\left(
	\sum_{t=k-G_j+1}^{\wh\cp_j-G_j}\vep^*_t -
	2\sum_{t=k+1}^{\wh\cp_j}\vep^*_t +
	\sum_{t=k+G_j+1}^{\wh\cp_j+G_j}\vep^*_t \right): k = \wh\cp_j-1,\ldots,\wh\cp_j-c\sigma^2d_{j}^{-2}\,\Big|\, X_1,\ldots,X_n\right\}
\\
&\overset{\mathcal{D}}{=} \left\{U_n^*(\ell)=d_{j}\left(
\sum_{t=\ell}^{-1}(\vep^*_t)^{(1)} - 2\sum_{t=\ell}^{-1}(\vep_t^*)^{(2)}
+ \sum_{t=\ell}^{-1}(\vep_t^*)^{(3)}\right) : \ell=-1,\ldots,-c \sigma^2 d_{j}^{-2}\,\Big|\,X_1,\ldots,X_n\right\},
\end{align*}
where $\{(\vep^*_t)^{(3)}\}$ are distributed according to 
$\{\vep^*_t, \, \wh\cp_j < t \le \wh\cp_{j+1}\}$ 
and independent of $\{(\vep^*_t)^{(i)}\}$, $i = 1, 2$, 
which in turn are independent copies of $\{\vep^*_t, \, \wh\cp_{j-1} < t \le \wh\cp_j\}$.
Similar assertions hold for $k$ satisfying $0 \le \sigma^{-2} d_j^2( k-\wh\cp_j) \le c$:
Here, $\{(\vep^*_t)^{(1)}\}$ are distributed according to 
$\{\vep^*_t, \, \wh\cp_{j-1} < t \le \wh\cp_j\}$ 
and independent of $\{(\vep^*_t)^{(i)}\}$, $i = 2, 3$, 
which are independent copies of $\{\vep^*_t, \, \wh\cp_j < t \le \wh\cp_{j+1}\}$, 
and all of $\{(\vep^*_t)^{(i)}, \, t = 1, \ldots, c\sigma^2 d_j^{-2} ,\, i = 1, 2, 3\}$, 
are independent of 
$\{(\vep^*_t)^{(i)}, \, t = -1, \ldots, - c\sigma^2 d_j^{-2},\, i = 1, 2, 3\}$.

Arguments so far hold in the local and the fixed change cases.
Now, in the case of the local change ($d_j = d_{j, n} \to 0$) as in~\ref{thm:bootstrap:one} 
the proof is concluded as in the proof of Theorem~\ref{thm:asymp}~\ref{thm:asymp:one}
by replacing the functional central limit theorem there with a version suitable 
for the triangular arrays present in the bootstrap distribution 
as given in Lemma~5.1 of \cite{huvskova2008},
where the assumptions therein are fulfilled due to Lemma~\ref{lemma_variance_boot}~\ref{lemma_variance_boot:two}.
Hence for any $x$, it holds 
\begin{align*}
	\p^*\left(	\frac{d_j^2(\wt{\cp}^*_j - \wh\cp_j)}{\sigma^2} \le x \right) \pto
	\p\left( \arg\max_{s \in \R}\l\{W_s - |s|/\sqrt{6}\r\} \le x\right).
\end{align*}

For the fixed change case,
unlike in the proof of Theorem~\ref{thm:asymp} 
where the distribution in the limit is the same as that of $U_n(\ell)$ 
such that no additional limit theorem is required, 
here we do need that $U_n^*(\ell)$ converges to $U_n(\ell)$ in an appropriate sense. 
Due to the cutting technique employed in this proof, 
the required convergence follows from Lemma~\ref{lem_fixed_changes} as below:
First, by Lemma~\ref{lem_fixed_changes} in combination with (conditional) independence 
under the bootstrap distribution and Lemma~\ref{lemma_variance_boot},
it holds with $L = c \sigma^2 d_{j}^{-2}$ (which is constant for fixed $d_j$) for any $x_t^{(i)}$:
\begin{align}
\label{eq_boot_FCLT}
\p^*\left( (\vep_{t}^*)^{(i)} \ls x_t^{(i)}, \, |t| = 1, \ldots , L, \, i = 1, 2, 3 \right)
\pto \p\left(  \vep_t^{(i)} \ls x_t^{(i)}, \, |t| = 1, \ldots, L, \, i = 1, 2, 3 \right)
\end{align}
with $\vep_t^{(i)}$ as in Theorem~\ref{thm:asymp}~\ref{thm:asymp:two},
such that for any $x_\ell$,
\begin{align*}
	&\p^*\left( U_n^*(\ell) \ls x_{\ell}: \, \ell = -1, \ldots, -L \right)
	\pto \p\left(U_n(\ell)\ls x_{\ell}: \, \ell = -1, \ldots, -L \right).
\end{align*}
The proof can then be concluded as in the proof of Theorem~\ref{thm:asymp} 
on noting that the errors in the limit distribution have a continuous distribution. 
\end{proof}
%
%
%

\section{MOSUM-based procedures for multiple change point estimation}
\label{sec:mosum}


In \cite{eichinger2018}, simultaneous estimation of multiple change points
via a single-scale MOSUM procedure has been considered
which, for a bandwidth $G = G_n$, estimates the locations of the change points
as where significant local maxima of the MOSUM statistics in~\eqref{eq:mosum:symm} are attained.
For the identification of these significant local maxima, different criteria have been considered.
One such a criterion regards $\wh\cp$ as a change point estimator
when it is the local maximiser of the MOSUM statistic within its
$\lfloor \eta G\rfloor$-radius for some $\eta \in (0, 1)$,
and $\vert T_{\wh\cp, n}(G; X) \vert > \wh\sigma_n \; D_n(G; \alpha)$.
Here, $D_n(G; \alpha)$ is a threshold that controls asymptotically the family-wise error rate 
(of $\sigma^{-1} \vert T_{k, n}(G; \vep) \vert$ exceeding $D_n(G; \alpha)$ over $G \le k \le n - G$) 
at the significance level $\alpha \in (0, 1)$,
and $\wh\sigma_n^2$ a suitable estimator of $\sigma^2$.
In the context of testing for a single change point, 
permutation methods have been proposed as an alternative to the asymptotic threshold \citep{antoch2001permutation, huvskova2004}. Such an approach has also been suggested by \cite{arias2018distribution}
when adopting scan statistics for anomaly detection.

Let the set of estimators obtained as described above be denoted by 
$\wh\Cp(G)  = \{\wh\cp_j(G), \, 1 \le j \le \wh q_n(G)\}$.
Corollary~D.2 of \cite{cho2019two}, improving upon Theorem~3.2 of \cite{eichinger2018}, 
shows that $\wh\Cp(G)$ satisfies:
\begin{align}
\label{eq:consistency}
\p\l( \wh q_n(G) = q_n, \, \wh\cp_j(G) = \wt\cp_j, \, j = 1, \ldots, q_n, \text{ and }
\max_{1 \le j \le q_n} d_j^2 \vert \wh\cp_j(G) - \cp_j \vert \le \rho_n\r) \to 1
\end{align}
under mild assumptions on $\{\vep_t\}$ permitting heavy-tails and serial dependence, for a suitable bandwidth $G$. 
The conditions on $G$ depend on the moments of the error sequence on the one hand
(such that $G$ can be smaller as $\nu$ in~\eqref{eq:errors} increases)  
as well as on the distance between neighbouring change points 
(i.e.\ $2G < \min_{1 \le j \le q_n} \delta_j$)
on the other hand.
We require the size of the changes to be sufficiently large such that
$\min_{1 \le j \le q_n} d_j^2 G \ge D_n$ with $D_n \to \infty$ at an appropriate rate
in relation to the behaviour of $\{\vep_t\}$ and, the resulting localisation rate satisfies
$D_n^{-1} \rho_n \to 0$.
This, together with the condition on the size of the changes and~\eqref{eq:consistency},
shows that Assumptions~\ref{assum_meta_est} and~\ref{assum_precision_est} are fulfilled 
by $\wh\Cp(G)$ with an appropriately chosen $G$.

In the important special case where $q_n = q$ is finite 
(as in Theorems~\ref{thm:asymp}~\ref{thm:asymp:three} 
and~\ref{thm:bootstrap}~\ref{thm:bootstrap:three}), 
the consistency result in~\eqref{eq:consistency} holds for $\rho_n$
that diverges at an arbitrarily slow rate,
i.e.\ $\max_{1 \le j \le q} d_j^2 \vert \wh\cp_j(G) - \cp_j \vert =O_P(1)$.
Theorem~\ref{thm:asymp} 
is closely related to the later result, 
not only showing that this rate is exact but also deriving 
the corresponding (non-degenerate) limit distribution. 
In fact, this localisation rate is minimax optimal 
in the multiple change point detection problem in~\eqref{eq:model}
(see \cite{fromont2020}).

Generally speaking,
this single-scale MOSUM procedure performs  best
with the bandwidth chosen as large as possible 
while avoiding to have more than one change point within the moving window at any time.
Therefore, it lacks adaptivity when the change points are {\it heterogeneous},
i.e.\ when the data sequence contains both large changes over short intervals
and small changes over long intervals of stationarity.
In such a situation, applying the MOSUM procedure with a range of bandwidths,
and then combining information across the multiple bandwidths,
is one way of addressing this lack of adaptivity of the single-scale MOSUM procedure,
at the cost of requiring a more complicated model selection procedure.
The results in this paper take into account the possibility 
of using different bandwidths $G_j$ for the detection of individual $\cp_j, \, j = 1, \ldots, q_n$,
see~\eqref{eq:cp:tilde}.

One such model selection procedure is the localised pruning proposed by \cite{cho2019two}.
When applied to the set of candidate change point estimators
generated by the multiscale MOSUM procedure, 
it returns $\wh\Cp = \{\wh\cp_j, \, 1 \le j \le \wh q_n\}$
that achieves consistency by correctly estimating the number of change points $q_n$
as well as `almost' inheriting the localisation property of $\wt\cp_j$,
in the sense that $\max_{1 \le j \le q_n} d_j^2 \vert \wh\cp_j - \cp_j \vert = O_P( \nu_n\rho_n)$ 
with $\rho_n$ as in \eqref{eq:consistency}
for $\nu_n \to \infty$ at an arbitrary rate (see Corollary~4.2 in \cite{cho2019two}),
when the set of bandwidths is suitably chosen.
Furthermore, \cite{cho2019two} formulate 
a rigorous framework permitting the aforementioned heterogeneity in change points
and show that, under such a general setting,
the multiscale MOSUM procedure combined with the localised pruning
(termed MoLP in Section~\ref{sec:sim:lp})
is (almost) minimax optimal in terms of both separation 
(related to correctly estimating the number of change points)
and localisation rates
when $\{\vep_t\}$ are i.i.d.\ sub-Gaussian or when $q_n$ is finite. 

For algorithmic descriptions of the above procedures and further information about the R~package
{\tt mosum} implementing them, see \cite{meier2021mosum}.

\clearpage

\section{Additional simulation results}
\label{sec:sim:add}

In this section, we provide additional simulation results
obtained with $\{\vep_t\}$ following Gaussian distributions (Appendix~\ref{sec:sim:add})
as in Section~\ref{sec:sim},
and $t_5$ distributions (Appendix~\ref{sec:sim:t})
for the five test signals, see Figure~\ref{fig:testsignals} for illustration.
We keep the signal-to-noise ratio constant across the two scenarios.
When generating bootstrap CIs with the additional model selection step using MoLP procedure,
we apply a slightly larger penalty of $\log^{1.1}(n)$ 
when the errors are generated from $t_5$ distributions,
in place of $\log^{1.01}(n)$ recommended by default and adopted for the Gaussian errors,
for the localised pruning procedure of \cite{cho2019two};
all other tuning parameters are set identically.

As noted in Section~\ref{sec:sim},
the bootstrap CIs constructed with the oracle estimators $\wt\cp_j$
closely attain the desired confidence level
while those constructed with the estimators obtained from the additional model selection step
tend to be more conservative,
and we observe little difference in their behaviour
whether $\{\vep_t\}$ follow Gaussian distributions
or $t_5$ distributions.

\begin{figure}[htbp]
\centering
\includegraphics[width=1\textwidth]{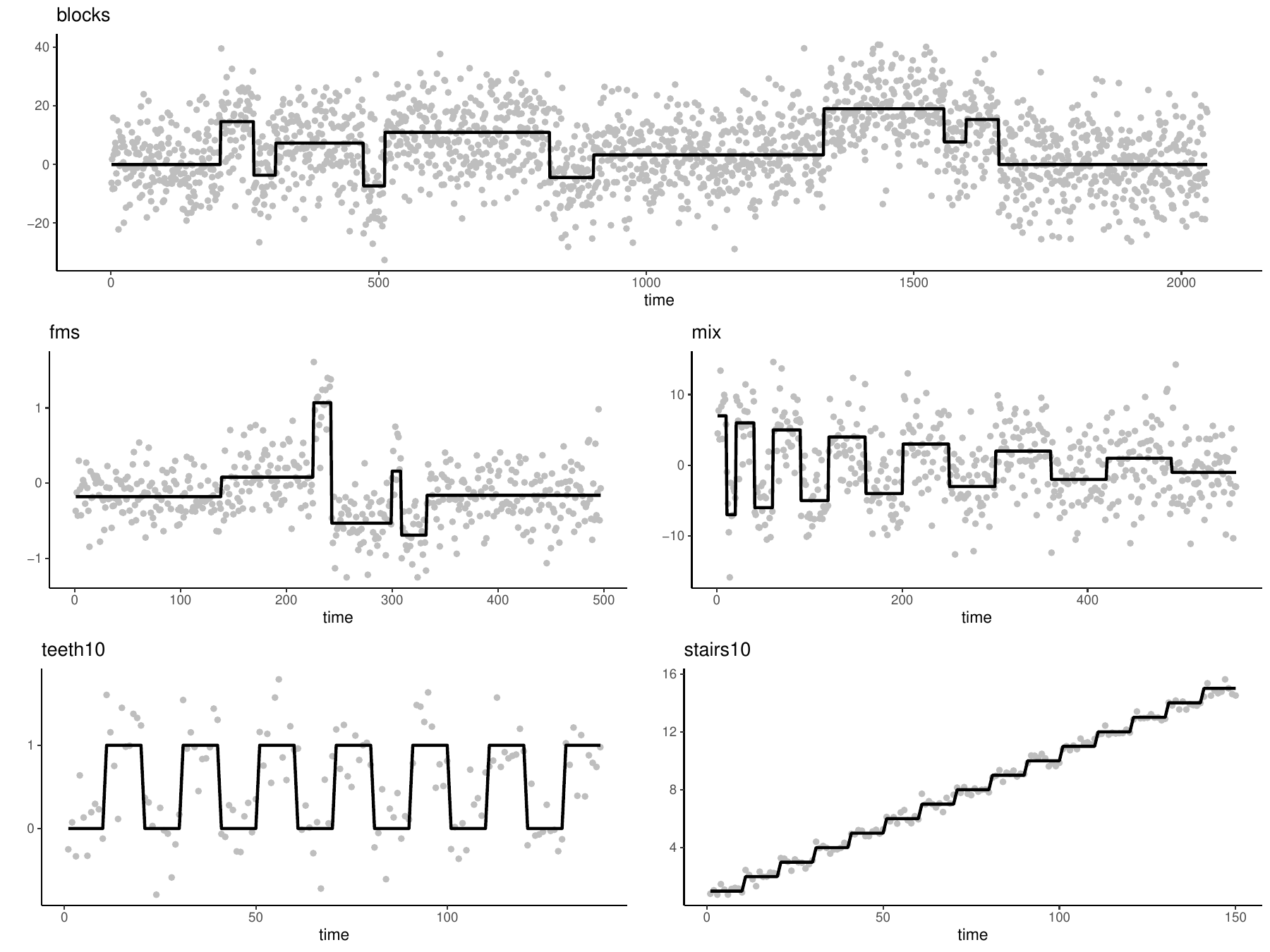}
\caption{Realisations from the {\tt blocks}, {\tt fms}, {\tt mix}, {\tt teeth10} and {\tt stairs10} test signals from \cite{fryzlewicz2014} with Gaussian errors.}
\label{fig:testsignals}
\end{figure}

\clearpage

\subsection{Gaussian errors}
\label{sec:sim:gauss}

\subsubsection{Bootstrap CIs constructed with the oracle estimators in~\eqref{eq:cp:tilde}}

\begin{figure}[htbp]
\centering
\includegraphics[width=.8\textwidth]{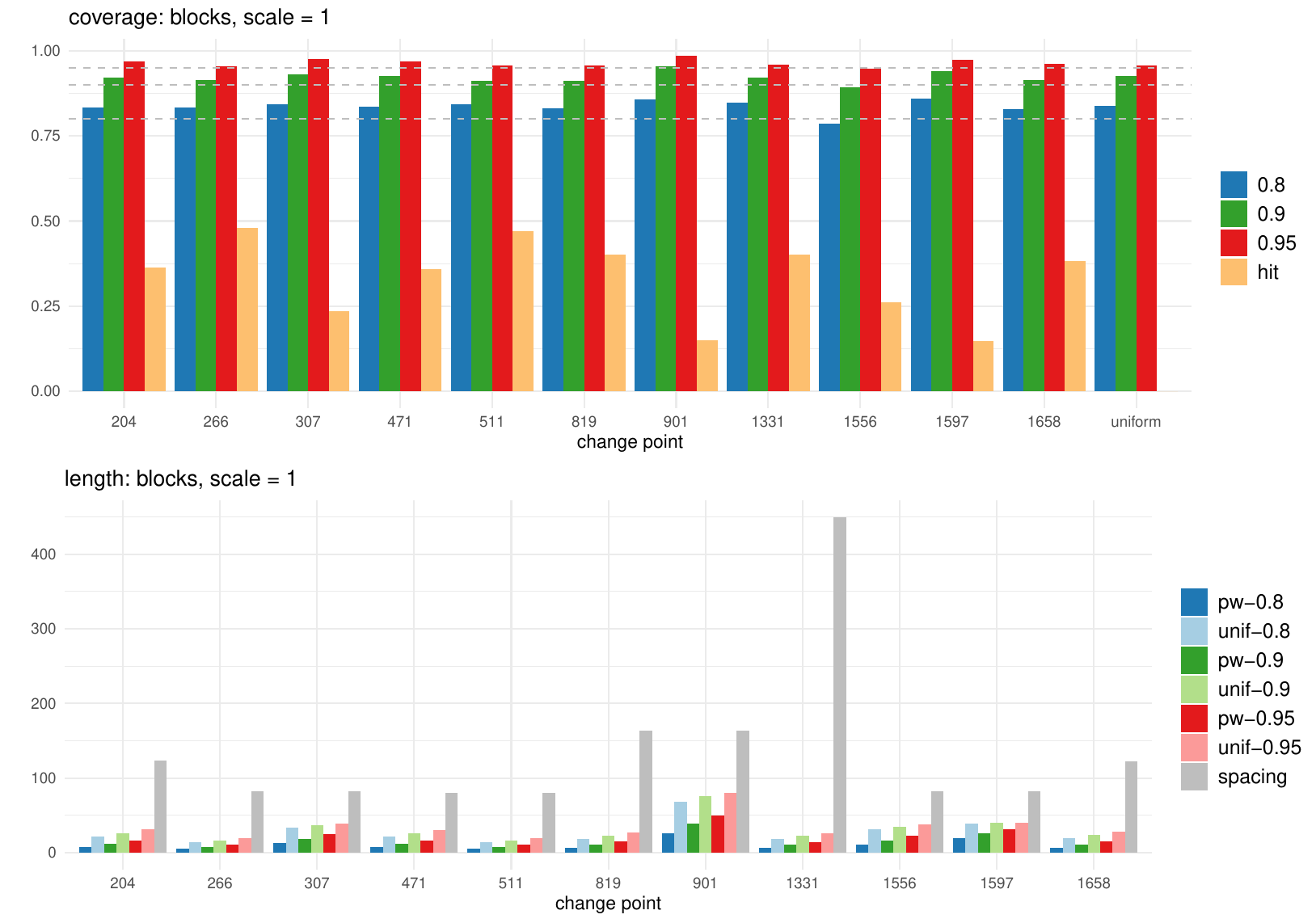}
\caption{{\tt blocks} with $\vartheta = 1$.}
\label{fig:blocks:one}
\end{figure}


\begin{figure}[htbp]
\centering
\includegraphics[width=.8\textwidth]{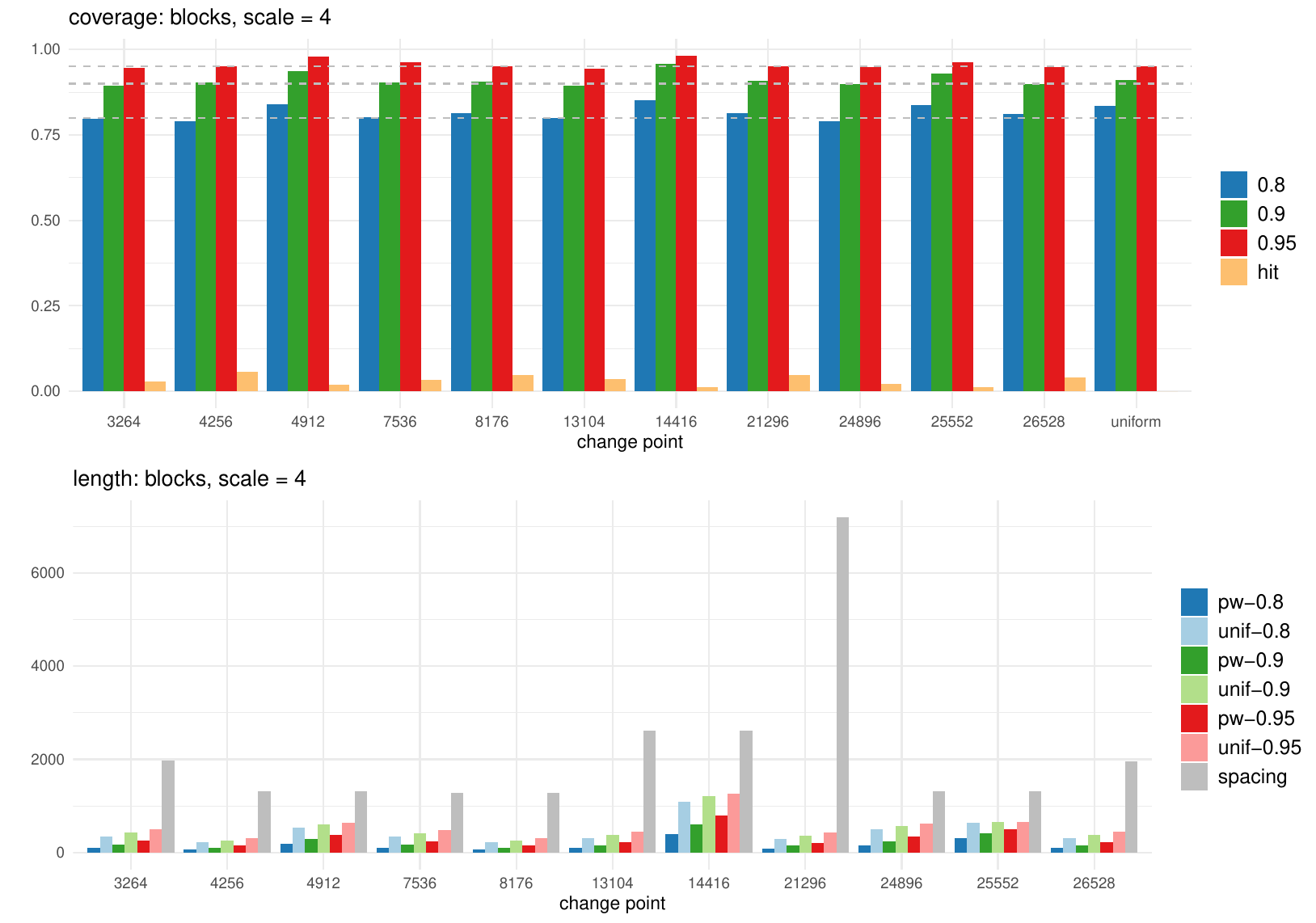}
\caption{{\tt blocks} with $\vartheta = 4$.}
\label{fig:blocks:four}
\end{figure}

\begin{figure}[htbp]
\centering
\includegraphics[width=.8\textwidth]{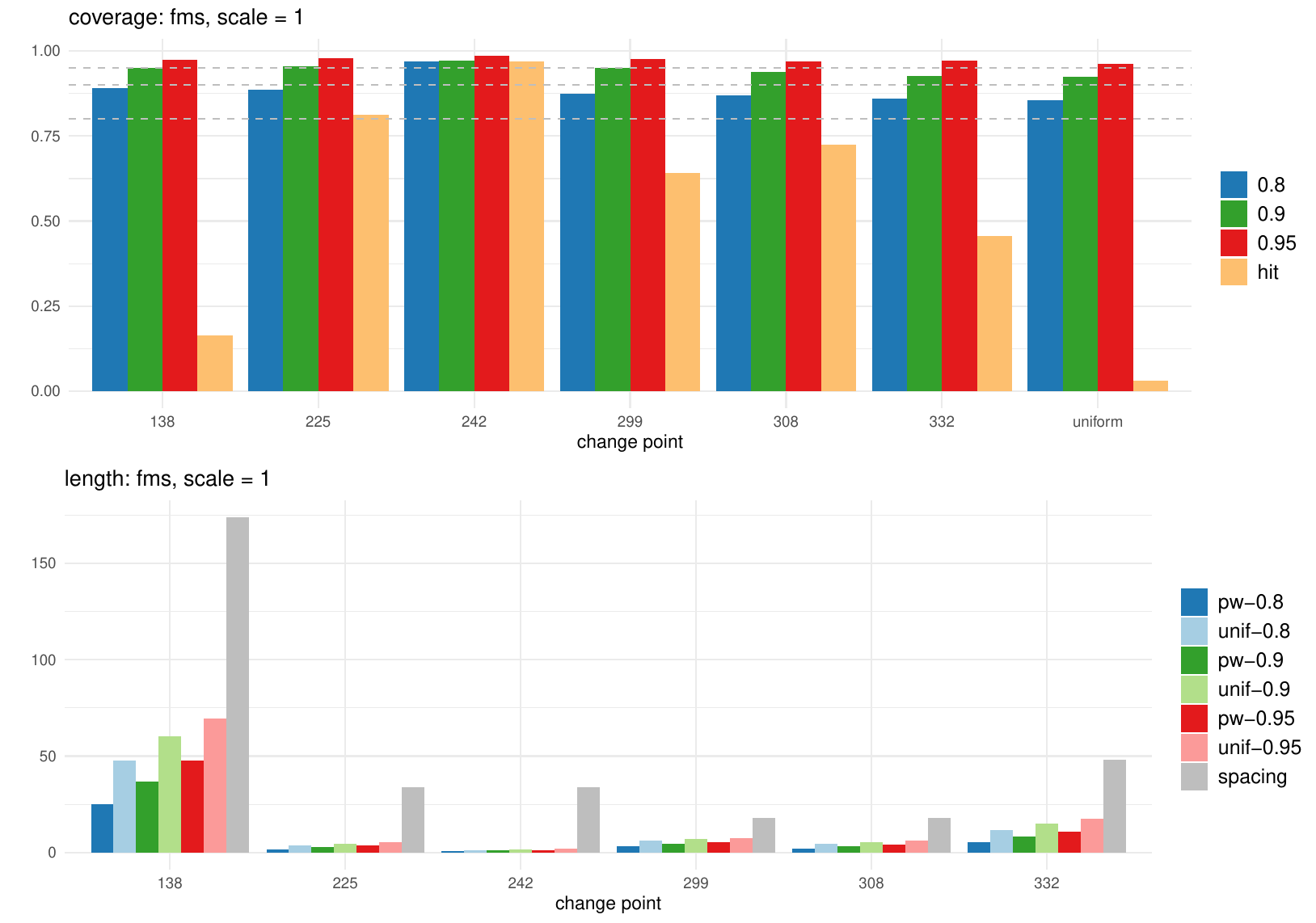}
\caption{Bootstrap CIs constructed with the oracle estimators in~\eqref{eq:cp:tilde}:
{\tt fms} with $\vartheta = 1$.}
\label{fig:fms:one}
\end{figure}


\begin{figure}[htbp]
\centering
\includegraphics[width=.8\textwidth]{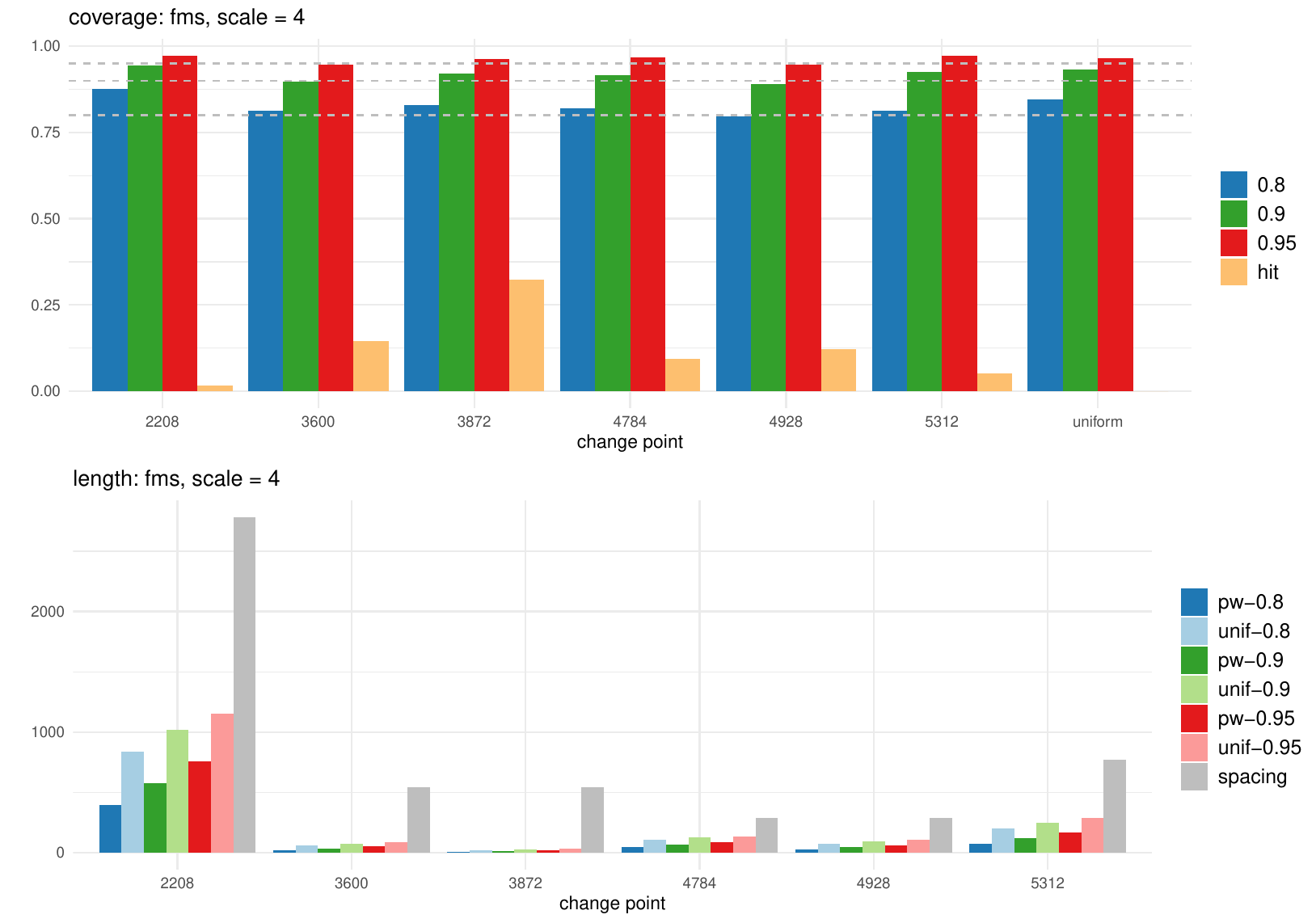}
\caption{Bootstrap CIs constructed with the oracle estimators in~\eqref{eq:cp:tilde}:
{\tt fms} with $\vartheta = 4$.}
\label{fig:fms:four}
\end{figure}


%
%
%
\begin{figure}[htbp]
\centering
\includegraphics[width=.8\textwidth]{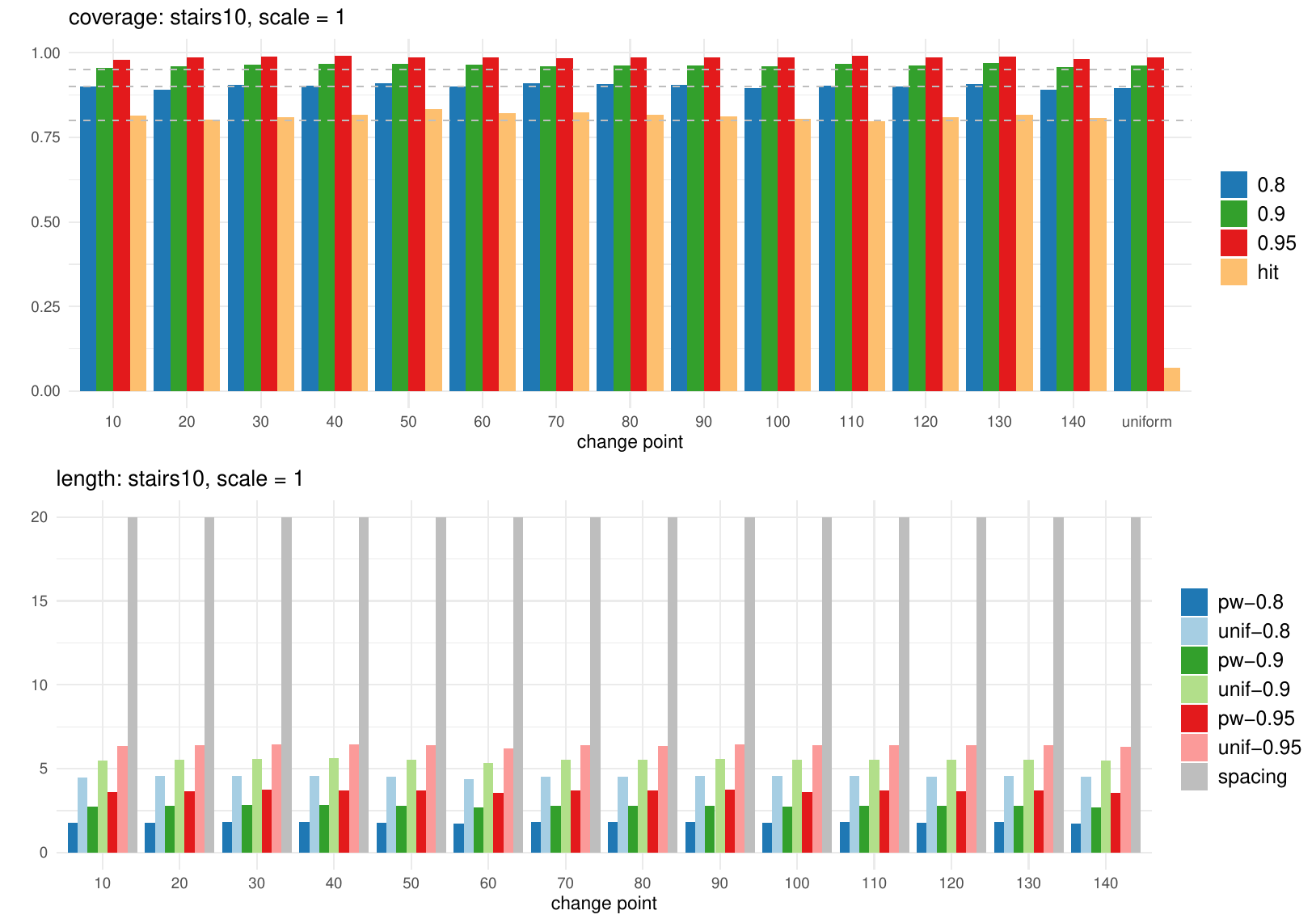}
\caption{Bootstrap CIs constructed with the oracle estimators in~\eqref{eq:cp:tilde}: {\tt stairs10} with $\vartheta = 1$.}
\label{fig:stairs10:one}
\end{figure}
%
%
\begin{figure}[htbp]
\centering
\includegraphics[width=.8\textwidth]{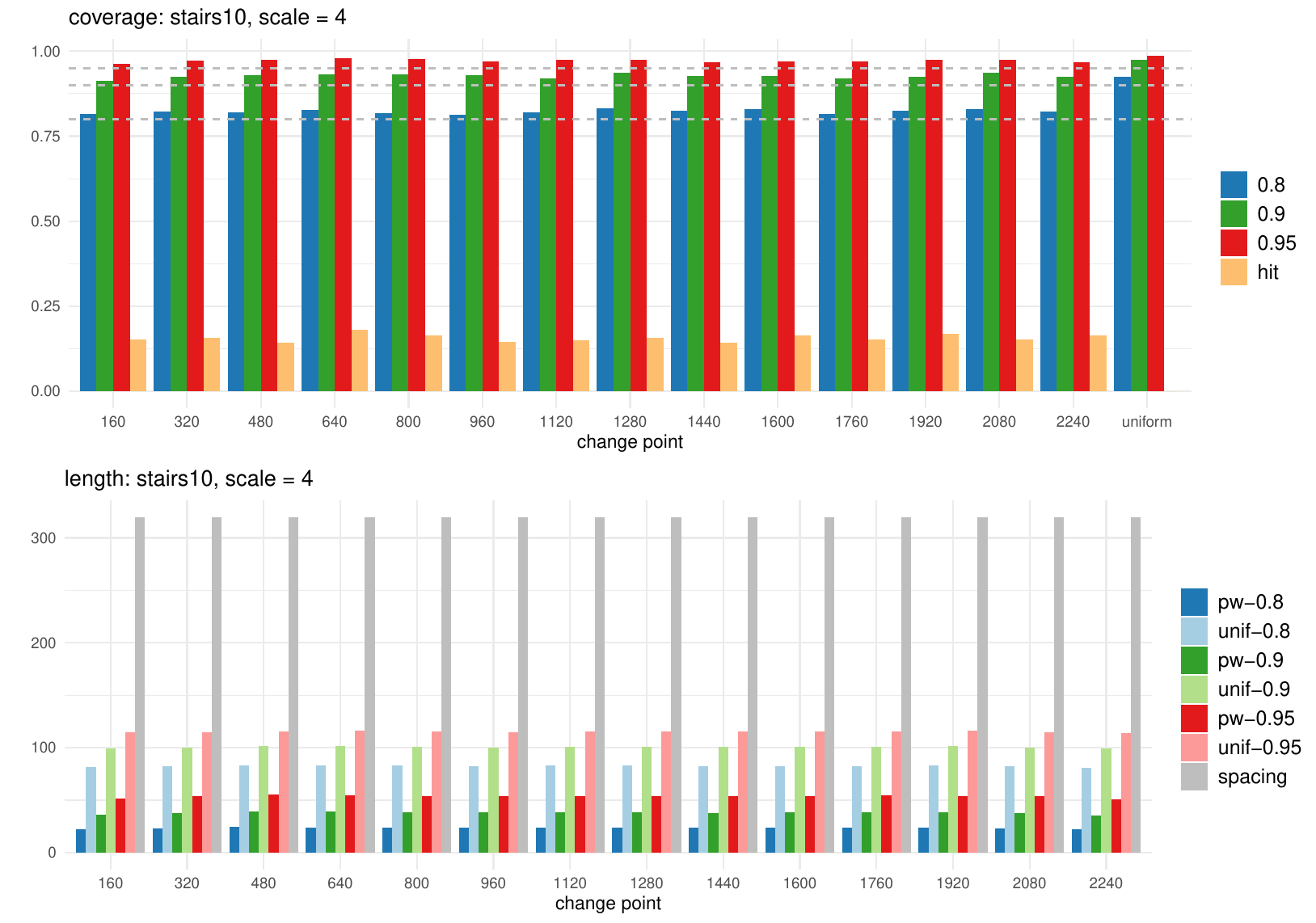}
\caption{Bootstrap CIs constructed with the oracle estimators in~\eqref{eq:cp:tilde}: {\tt stairs10} with $\vartheta = 4$.}
\label{fig:stairs10:four}
\end{figure}

\clearpage

\subsubsection{Bootstrap CIs constructed with model selection}

\begin{figure}[htbp]
\centering
\includegraphics[width=.75\textwidth]{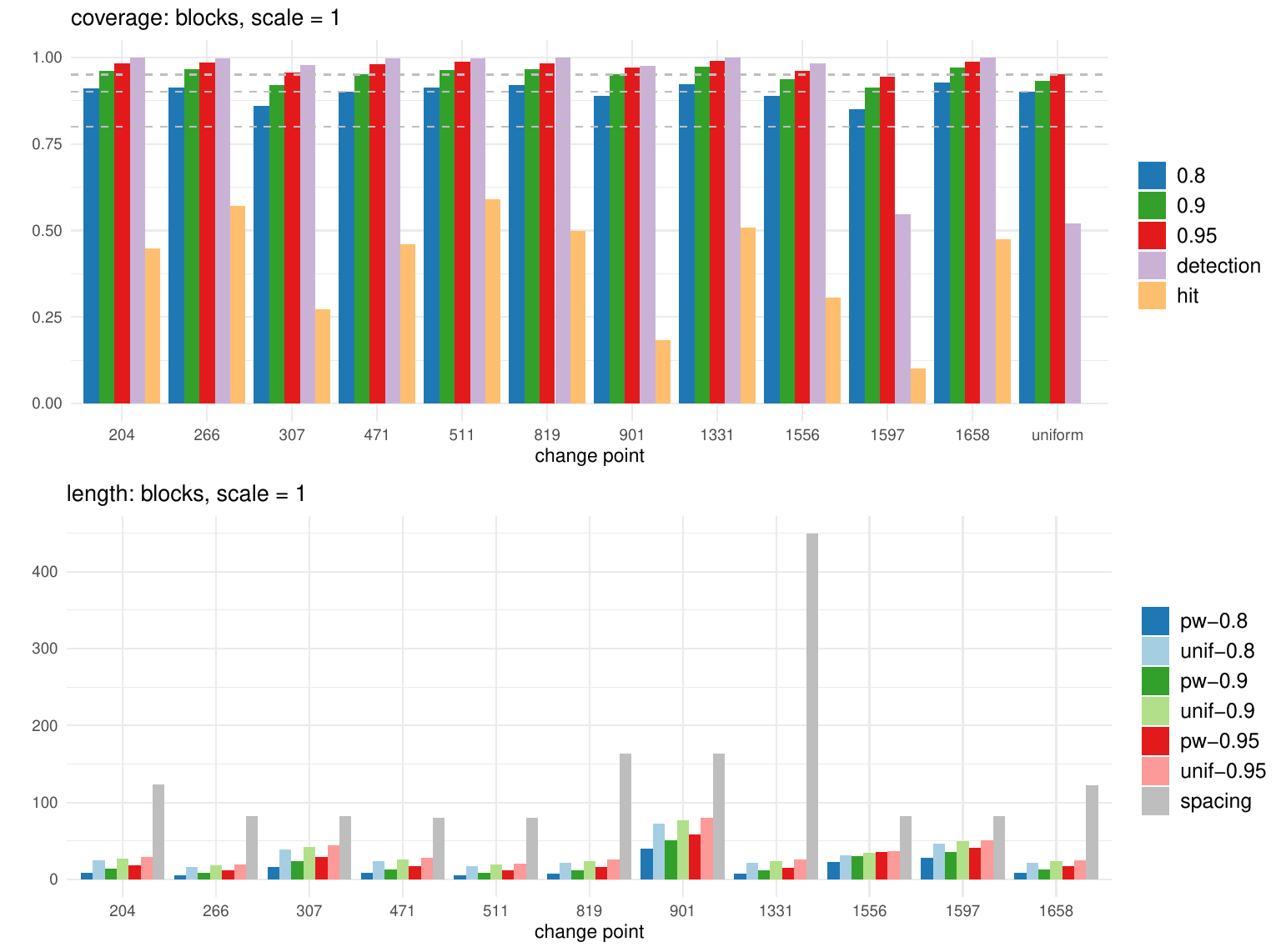}
\caption{Bootstrap CIs constructed with model selection: {\tt blocks} with $\vartheta = 1$.}
\label{fig:full:blocks:one}
\end{figure}


\begin{figure}[htbp]
\centering
\includegraphics[width=.75\textwidth]{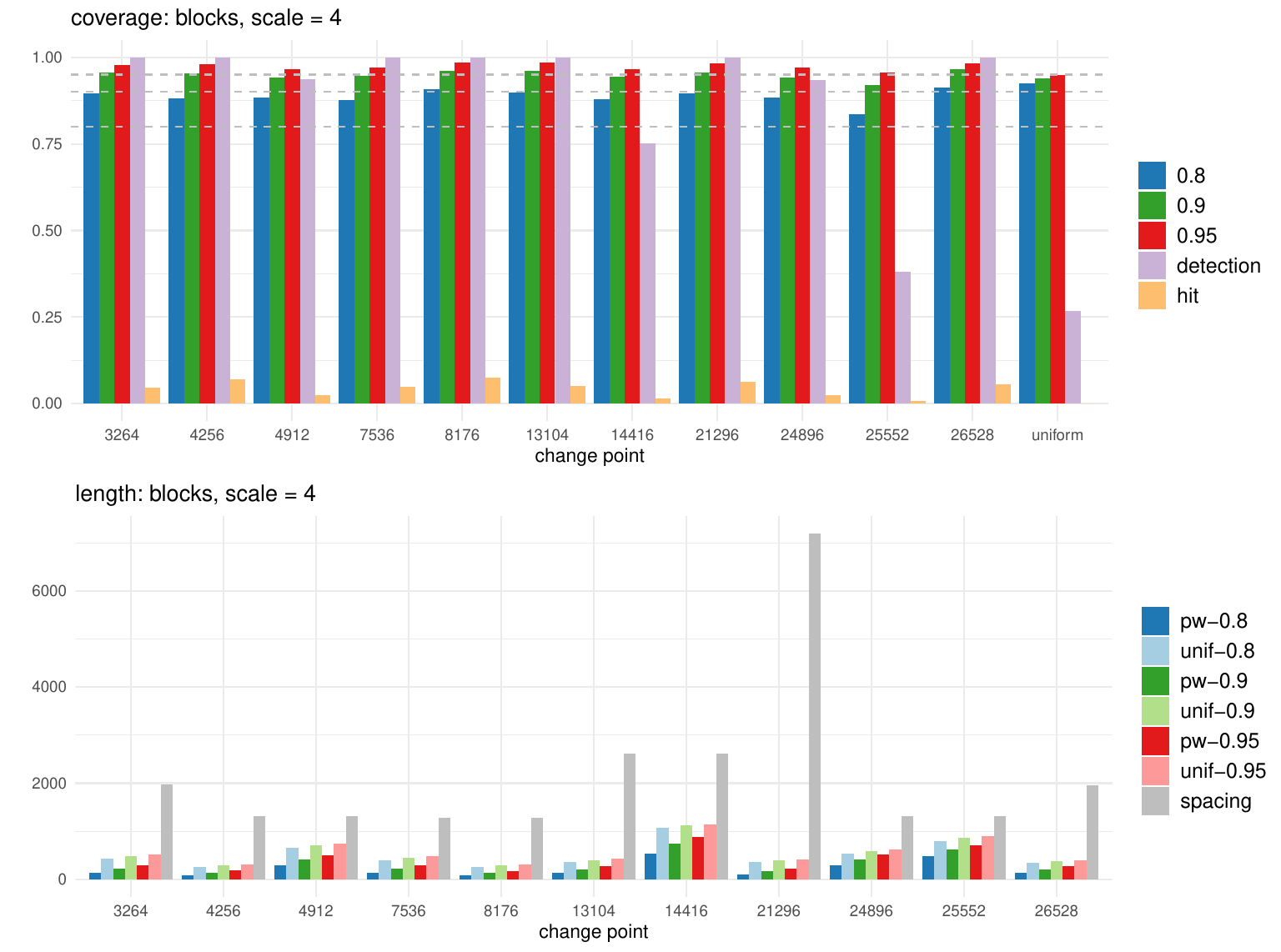}
\caption{Bootstrap CIs constructed with model selection: {\tt blocks} with $\vartheta = 4$.}
\label{fig:full:blocks:four}
\end{figure}

\begin{figure}[htbp]
\centering
\includegraphics[width=.8\textwidth]{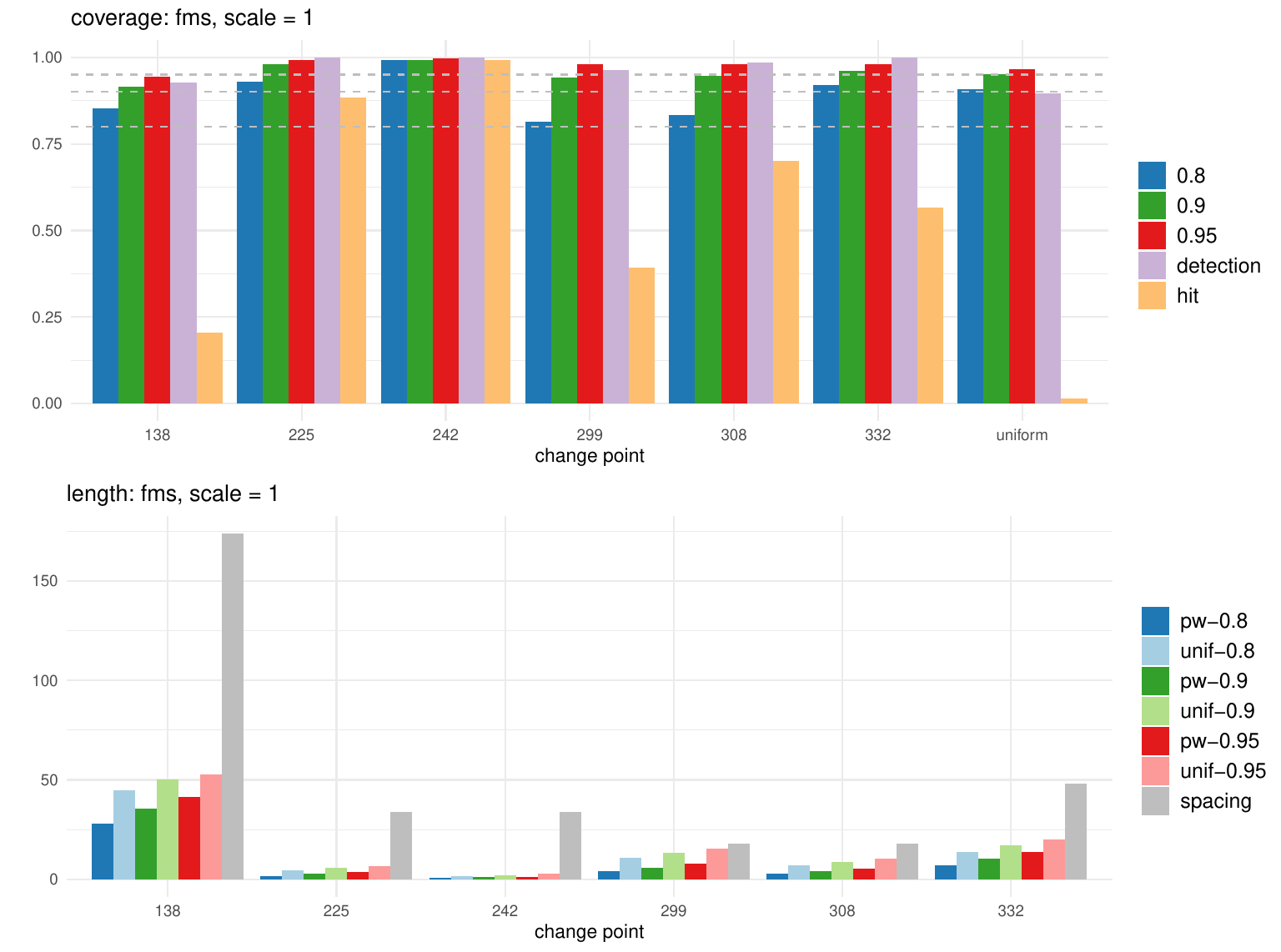}
\caption{Bootstrap CIs constructed with model selection: {\tt fms} with $\vartheta = 1$.}
\label{fig:full:fms:one}
\end{figure}


\begin{figure}[htbp]
\centering
\includegraphics[width=.8\textwidth]{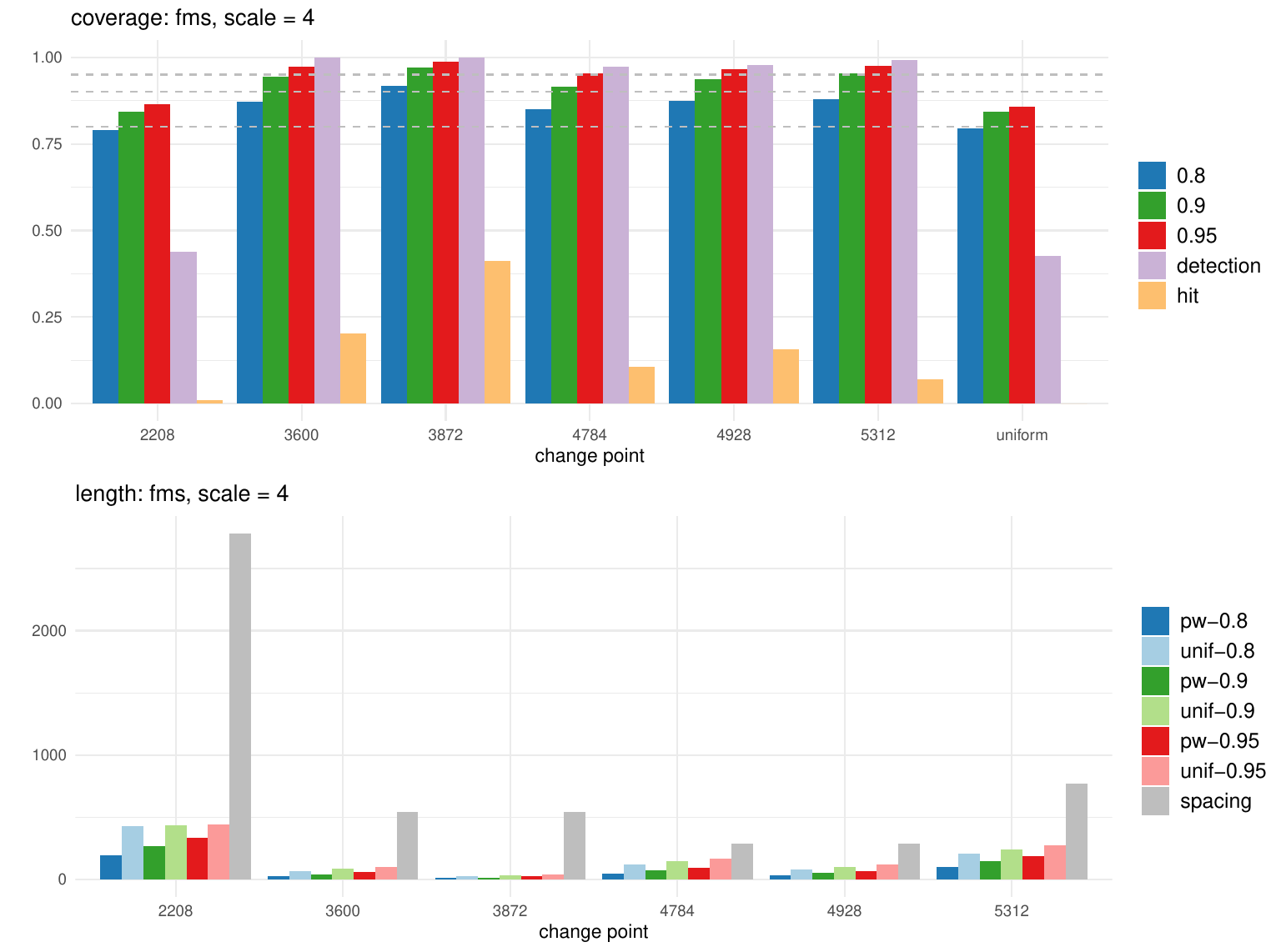}
\caption{Bootstrap CIs constructed with model selection: {\tt fms} with $\vartheta = 4$.}
\label{fig:full:fms:four}
\end{figure}


%
%
%
\begin{figure}[htbp]
\centering
\includegraphics[width=.8\textwidth]{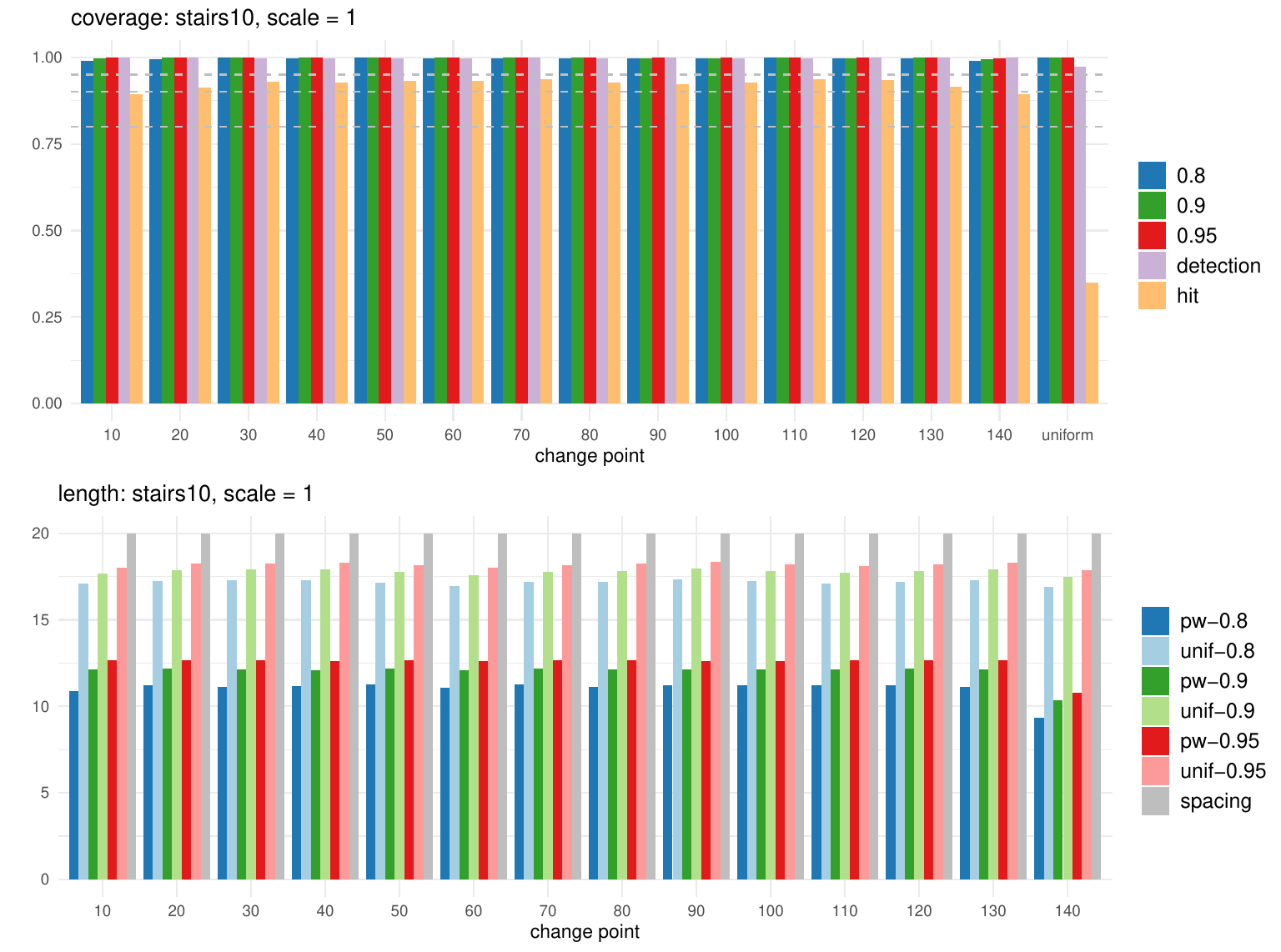}
\caption{Bootstrap CIs constructed with model selection: {\tt stairs10} with $\vartheta = 1$.}
\label{fig:full:stairs10:one}
\end{figure}
%
%
\begin{figure}[htbp]
\centering
\includegraphics[width=.8\textwidth]{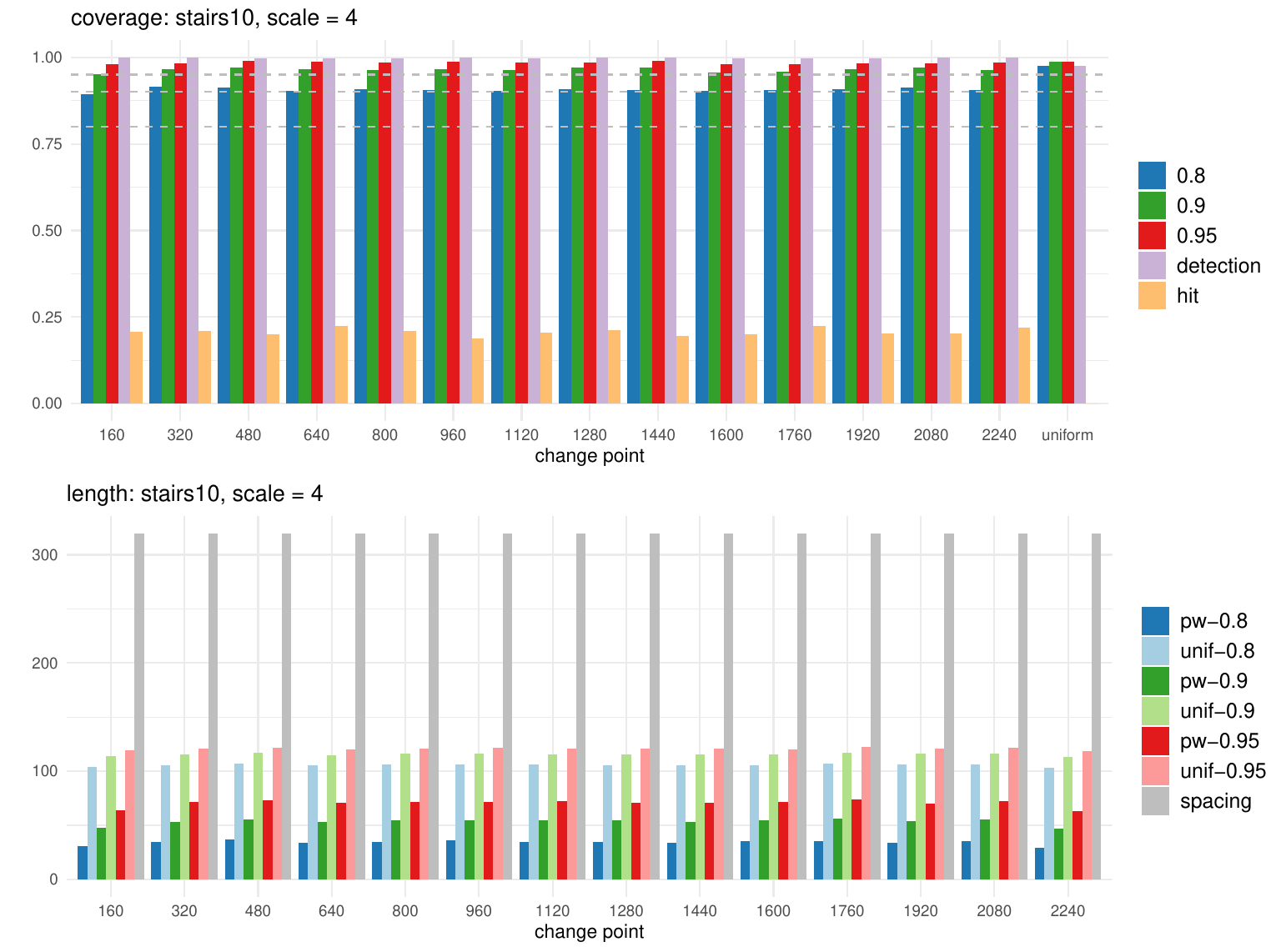}
\caption{Bootstrap CIs constructed with model selection: {\tt stairs10} with $\vartheta = 4$.}
\label{fig:full:stairs10:four}
\end{figure}

\clearpage

\subsubsection{Comparison of coverage}

Tables~\ref{table:cov:1}--\ref{table:cov:4} compare the coverage of the bootstrap CIs
constructed with the oracle estimators in~\eqref{eq:cp:tilde}
and those constructed with the model selection step,
averaged over $2000$ realisations.

\begin{table}[ht]
\caption{Average coverage of the bootstrap CIs constructed with the oracle estimators
and the estimators obtained from MoLP, when $\vartheta = 1$.
We also report the proportion of realisations where individual change points are detected (by MoLP)
and where all change points are correctly detected.}
\label{table:cov:1}
\centering
\resizebox{\columnwidth}{!}{
\begin{tabular}{ccc  ccc ccc ccc ccc cc  c}
\toprule
test signal &	$1 - \alpha$ &	estimator &	$\cp_1$ &	$\cp_2$ &	$\cp_3$ &	$\cp_4$ &	$\cp_5$ &	$\cp_6$ &	$\cp_7$ &	$\cp_8$ &	$\cp_9$ &	$\cp_{10}$ &	$\cp_{11}$ &	$\cp_{12}$ &	$\cp_{13}$ &	$\cp_{14}$ &	uniform	\\	
\cmidrule(lr){1-3} \cmidrule(lr){4-17} \cmidrule(lr){18-18}
{\tt blocks} &	0.8 &	oracle &	0.834 &	0.832 &	0.844 &	0.836 &	0.842 &	0.831 &	0.858 &	0.848 &	0.787 &	0.86 &	0.828 &	-- &	-- &	-- &	0.838	\\	
&	&	MoLP &	0.911 &	0.913 &	0.859 &	0.902 &	0.912 &	0.921 &	0.89 &	0.922 &	0.888 &	0.85 &	0.928 &	-- &	-- &	-- &	0.902	\\	
&	0.9 &	oracle &	0.921 &	0.915 &	0.932 &	0.926 &	0.912 &	0.912 &	0.954 &	0.922 &	0.892 &	0.941 &	0.914 &	-- &	-- &	-- &	0.926	\\	
&	&	MoLP &	0.96 &	0.965 &	0.92 &	0.95 &	0.964 &	0.965 &	0.95 &	0.974 &	0.937 &	0.913 &	0.97 &	-- &	-- &	-- &	0.932	\\	
&	0.95 &	oracle &	0.97 &	0.955 &	0.976 &	0.969 &	0.958 &	0.957 &	0.986 &	0.958 &	0.947 &	0.974 &	0.961 &	-- &	-- &	-- &	0.958	\\	
&	&	MoLP &	0.983 &	0.984 &	0.957 &	0.98 &	0.987 &	0.982 &	0.97 &	0.989 &	0.962 &	0.944 &	0.986 &	-- &	-- &	-- &	0.951	\\	
\cmidrule(lr){2-3} \cmidrule(lr){4-17} \cmidrule(lr){18-18}
&	— &	detection &	0.999 &	0.997 &	0.978 &	0.997 &	0.998 &	0.999 &	0.976 &	1 &	0.982 &	0.546 &	1 &	-- &	-- &	-- &	0.521	\\	
\cmidrule(lr){1-3} \cmidrule(lr){4-17} \cmidrule(lr){18-18}
{\tt fms} &	0.8 &	oracle &	0.891 &	0.887 &	0.968 &	0.875 &	0.87 &	0.859 &	-- &	-- &	-- &	-- &	-- &	-- &	-- &	-- &	0.856	\\	
&	&	MoLP &	0.852 &	0.93 &	0.992 &	0.815 &	0.833 &	0.919 &	-- &	-- &	-- &	-- &	-- &	-- &	-- &	-- &	0.908	\\	
&	0.9 &	oracle &	0.95 &	0.956 &	0.972 &	0.95 &	0.938 &	0.928 &	-- &	-- &	-- &	-- &	-- &	-- &	-- &	-- &	0.926	\\	
&	&	MoLP &	0.914 &	0.98 &	0.992 &	0.941 &	0.946 &	0.961 &	-- &	-- &	-- &	-- &	-- &	-- &	-- &	-- &	0.952	\\	
&	0.95 &	oracle &	0.974 &	0.978 &	0.986 &	0.978 &	0.969 &	0.972 &	-- &	-- &	-- &	-- &	-- &	-- &	-- &	-- &	0.962	\\	
&	&	MoLP &	0.945 &	0.993 &	0.996 &	0.979 &	0.981 &	0.98 &	-- &	-- &	-- &	-- &	-- &	-- &	-- &	-- &	0.966	\\	
\cmidrule(lr){2-3} \cmidrule(lr){4-17} \cmidrule(lr){18-18}
&	— &	detection &	0.928 &	1 &	1 &	0.964 &	0.986 &	1 &	-- &	-- &	-- &	-- &	-- &	-- &	-- &	-- &	0.895	\\	
\cmidrule(lr){1-3} \cmidrule(lr){4-17} \cmidrule(lr){18-18}
{\tt mix} &	0.8 &	oracle &	0.9 &	0.88 &	0.892 &	0.899 &	0.868 &	0.874 &	0.856 &	0.857 &	0.845 &	0.814 &	0.826 &	0.814 &	0.864 &	-- &	0.84	\\	
&	&	MoLP &	0.95 &	0.957 &	0.926 &	0.941 &	0.937 &	0.926 &	0.92 &	0.91 &	0.898 &	0.875 &	0.862 &	0.828 &	0.752 &	-- &	0.819	\\	
&	0.9 &	oracle &	0.956 &	0.948 &	0.95 &	0.946 &	0.926 &	0.938 &	0.922 &	0.926 &	0.928 &	0.908 &	0.934 &	0.922 &	0.942 &	-- &	0.927	\\	
&	&	MoLP &	0.964 &	0.981 &	0.972 &	0.971 &	0.974 &	0.969 &	0.965 &	0.964 &	0.953 &	0.931 &	0.93 &	0.894 &	0.857 &	-- &	0.888	\\	
&	0.95 &	oracle &	0.98 &	0.98 &	0.976 &	0.976 &	0.968 &	0.972 &	0.958 &	0.966 &	0.969 &	0.96 &	0.972 &	0.966 &	0.972 &	-- &	0.964	\\	
&	&	MoLP &	0.973 &	0.993 &	0.982 &	0.984 &	0.988 &	0.988 &	0.986 &	0.984 &	0.975 &	0.958 &	0.956 &	0.922 &	0.913 &	-- &	0.911	\\	
\cmidrule(lr){2-3} \cmidrule(lr){4-17} \cmidrule(lr){18-18}
&	— &	detection &	0.999 &	0.999 &	1 &	1 &	1 &	1 &	1 &	1 &	1 &	0.994 &	0.938 &	0.635 &	0.3 &	-- &	0.279	\\	
\cmidrule(lr){1-3} \cmidrule(lr){4-17} \cmidrule(lr){18-18}
{\tt teeth10} &	0.8 &	oracle &	0.878 &	0.874 &	0.868 &	0.872 &	0.876 &	0.876 &	0.872 &	0.868 &	0.88 &	0.856 &	0.863 &	0.864 &	0.862 &	-- &	0.784	\\	
&	&	MoLP &	0.894 &	0.977 &	0.983 &	0.982 &	0.985 &	0.989 &	0.981 &	0.976 &	0.98 &	0.983 &	0.982 &	0.981 &	0.905 &	-- &	0.889	\\	
&	0.9 &	oracle &	0.948 &	0.946 &	0.944 &	0.941 &	0.942 &	0.942 &	0.936 &	0.94 &	0.946 &	0.935 &	0.939 &	0.938 &	0.946 &	-- &	0.882	\\	
&	&	MoLP &	0.933 &	0.993 &	0.994 &	0.992 &	0.996 &	0.995 &	0.992 &	0.991 &	0.991 &	0.993 &	0.994 &	0.991 &	0.946 &	-- &	0.954	\\	
&	0.95 &	oracle &	0.976 &	0.976 &	0.974 &	0.976 &	0.972 &	0.974 &	0.974 &	0.971 &	0.971 &	0.976 &	0.974 &	0.974 &	0.972 &	-- &	0.939	\\	
&	&	MoLP &	0.958 &	0.996 &	0.997 &	0.997 &	0.998 &	0.996 &	0.996 &	0.996 &	0.997 &	0.997 &	0.997 &	0.995 &	0.975 &	-- &	0.979	\\	
\cmidrule(lr){2-3} \cmidrule(lr){4-17} \cmidrule(lr){18-18}
&	— &	detection &	0.962 &	0.946 &	0.948 &	0.952 &	0.952 &	0.951 &	0.951 &	0.954 &	0.953 &	0.951 &	0.948 &	0.952 &	0.964 &	-- &	0.7	\\	
\cmidrule(lr){1-3} \cmidrule(lr){4-17} \cmidrule(lr){18-18}
{\tt stairs10} &	0.8 &	oracle &	0.902 &	0.89 &	0.906 &	0.902 &	0.91 &	0.902 &	0.91 &	0.907 &	0.904 &	0.896 &	0.902 &	0.9 &	0.906 &	0.89 &	0.896	\\	
&	&	MoLP &	0.99 &	0.995 &	0.999 &	0.998 &	0.999 &	0.998 &	0.997 &	0.998 &	0.996 &	0.996 &	0.998 &	0.997 &	0.997 &	0.99 &	1	\\	
&	0.9 &	oracle &	0.956 &	0.96 &	0.966 &	0.967 &	0.967 &	0.964 &	0.96 &	0.964 &	0.964 &	0.96 &	0.968 &	0.962 &	0.97 &	0.958 &	0.962	\\	
&	&	MoLP &	0.997 &	0.998 &	0.999 &	0.999 &	0.999 &	0.999 &	0.999 &	0.999 &	0.998 &	0.998 &	0.999 &	0.998 &	0.999 &	0.995 &	1	\\	
&	0.95 &	oracle &	0.98 &	0.986 &	0.989 &	0.99 &	0.986 &	0.987 &	0.983 &	0.987 &	0.986 &	0.986 &	0.992 &	0.986 &	0.99 &	0.982 &	0.986	\\	
&	&	MoLP &	0.998 &	0.998 &	0.999 &	0.999 &	0.999 &	0.999 &	0.999 &	0.999 &	0.998 &	0.998 &	0.999 &	0.999 &	0.999 &	0.997 &	1	\\
\cmidrule(lr){2-3} \cmidrule(lr){4-17} \cmidrule(lr){18-18}
&	— &	detection &	0.998 &	0.998 &	0.998 &	0.998 &	0.996 &	0.997 &	0.998 &	0.998 &	0.999 &	0.998 &	0.998 &	0.998 &	0.998 &	0.999 &	0.974	\\	
\bottomrule
\end{tabular}}
\end{table}

\begin{table}[ht]
\caption{Average coverage of the bootstrap CIs constructed with the oracle estimators
and the estimators obtained from MoLP, when $\vartheta = 4$.
We also report the proportion of realisations where individual change points are detected (by MoLP)
and where all change points are correctly detected.}
\label{table:cov:4}
\centering
\resizebox{\columnwidth}{!}{
\begin{tabular}{ccc  ccc ccc ccc ccc cc  c}
\toprule
test signal &	$1 - \alpha$ &	estimator &	$\cp_1$ &	$\cp_2$ &	$\cp_3$ &	$\cp_4$ &	$\cp_5$ &	$\cp_6$ &	$\cp_7$ &	$\cp_8$ &	$\cp_9$ &	$\cp_{10}$ &	$\cp_{11}$ &	$\cp_{12}$ &	$\cp_{13}$ &	$\cp_{14}$ &	uniform	\\	
\cmidrule(lr){1-3} \cmidrule(lr){4-17} \cmidrule(lr){18-18}
{\tt blocks} &	0.8 &	oracle &	0.796 &	0.79 &	0.84 &	0.802 &	0.812 &	0.8 &	0.852 &	0.814 &	0.79 &	0.837 &	0.812 &	-- &	-- &	-- &	0.834	\\	
&	&	MoLP &	0.897 &	0.88 &	0.884 &	0.877 &	0.909 &	0.898 &	0.879 &	0.897 &	0.884 &	0.836 &	0.913 &	-- &	-- &	-- &	0.924	\\	
&	0.9 &	oracle &	0.895 &	0.902 &	0.938 &	0.903 &	0.905 &	0.894 &	0.957 &	0.908 &	0.899 &	0.93 &	0.9 &	-- &	-- &	-- &	0.911	\\	
&	&	MoLP &	0.956 &	0.953 &	0.941 &	0.947 &	0.962 &	0.962 &	0.945 &	0.956 &	0.943 &	0.92 &	0.965 &	-- &	-- &	-- &	0.94	\\	
&	0.95 &	oracle &	0.945 &	0.952 &	0.978 &	0.962 &	0.952 &	0.944 &	0.982 &	0.952 &	0.948 &	0.962 &	0.948 &	-- &	-- &	-- &	0.952	\\	
&	&	MoLP &	0.978 &	0.98 &	0.967 &	0.97 &	0.986 &	0.984 &	0.966 &	0.983 &	0.97 &	0.955 &	0.982 &	-- &	-- &	-- &	0.95	\\	
\cmidrule(lr){2-3} \cmidrule(lr){4-17} \cmidrule(lr){18-18}
&	— &	detection &	1 &	1 &	0.938 &	1 &	1 &	1 &	0.752 &	1 &	0.934 &	0.382 &	1 &	-- &	-- &	-- &	0.268	\\	
\cmidrule(lr){1-3} \cmidrule(lr){4-17} \cmidrule(lr){18-18}
{\tt fms} &	0.8 &	oracle &	0.877 &	0.812 &	0.83 &	0.82 &	0.796 &	0.813 &	-- &	-- &	-- &	-- &	-- &	-- &	-- &	-- &	0.846	\\	
&	&	MoLP &	0.79 &	0.872 &	0.918 &	0.85 &	0.874 &	0.88 &	-- &	-- &	-- &	-- &	-- &	-- &	-- &	-- &	0.796	\\	
&	0.9 &	oracle &	0.944 &	0.897 &	0.92 &	0.916 &	0.89 &	0.926 &	-- &	-- &	-- &	-- &	-- &	-- &	-- &	-- &	0.933	\\	
&	&	MoLP &	0.844 &	0.945 &	0.97 &	0.916 &	0.937 &	0.953 &	-- &	-- &	-- &	-- &	-- &	-- &	-- &	-- &	0.842	\\	
&	0.95 &	oracle &	0.971 &	0.947 &	0.963 &	0.968 &	0.946 &	0.972 &	-- &	-- &	-- &	-- &	-- &	-- &	-- &	-- &	0.966	\\	
&	&	MoLP &	0.865 &	0.974 &	0.988 &	0.954 &	0.966 &	0.976 &	-- &	-- &	-- &	-- &	-- &	-- &	-- &	-- &	0.858	\\	
\cmidrule(lr){2-3} \cmidrule(lr){4-17} \cmidrule(lr){18-18}
&	— &	detection &	0.438 &	1 &	1 &	0.974 &	0.977 &	0.992 &	-- &	-- &	-- &	-- &	-- &	-- &	-- &	-- &	0.426	\\	
\cmidrule(lr){1-3} \cmidrule(lr){4-17} \cmidrule(lr){18-18}
{\tt mix} &	0.8 &	oracle &	0.8 &	0.806 &	0.808 &	0.794 &	0.799 &	0.798 &	0.788 &	0.788 &	0.816 &	0.789 &	0.796 &	0.806 &	0.847 &	-- &	0.822	\\	
&	&	MoLP &	0.884 &	0.898 &	0.903 &	0.901 &	0.884 &	0.884 &	0.883 &	0.867 &	0.873 &	0.863 &	0.86 &	0.821 &	0.574 &	-- &	0.615	\\	
&	0.9 &	oracle &	0.906 &	0.904 &	0.905 &	0.9 &	0.898 &	0.908 &	0.898 &	0.895 &	0.911 &	0.9 &	0.922 &	0.918 &	0.935 &	-- &	0.917	\\	
&	&	MoLP &	0.948 &	0.959 &	0.96 &	0.961 &	0.945 &	0.952 &	0.946 &	0.942 &	0.938 &	0.933 &	0.919 &	0.883 &	0.672 &	-- &	0.692	\\	
&	0.95 &	oracle &	0.954 &	0.954 &	0.952 &	0.951 &	0.946 &	0.96 &	0.95 &	0.952 &	0.962 &	0.96 &	0.968 &	0.965 &	0.973 &	-- &	0.966	\\	
&	&	MoLP &	0.974 &	0.982 &	0.98 &	0.984 &	0.975 &	0.976 &	0.973 &	0.967 &	0.965 &	0.96 &	0.934 &	0.903 &	0.738 &	-- &	0.692	\\	
\cmidrule(lr){2-3} \cmidrule(lr){4-17} \cmidrule(lr){18-18}
&	— &	detection &	1 &	1 &	1 &	1 &	0.999 &	0.999 &	0.998 &	0.982 &	0.921 &	0.692 &	0.376 &	0.098 &	0.031 &	-- &	0.0065	\\	
\cmidrule(lr){1-3} \cmidrule(lr){4-17} \cmidrule(lr){18-18}
{\tt teeth10} &	0.8 &	oracle &	0.796 &	0.82 &	0.798 &	0.816 &	0.801 &	0.805 &	0.804 &	0.811 &	0.802 &	0.812 &	0.788 &	0.802 &	0.799 &	-- &	0.882	\\	
&	&	MoLP &	0.844 &	0.869 &	0.842 &	0.854 &	0.847 &	0.84 &	0.846 &	0.84 &	0.844 &	0.842 &	0.836 &	0.831 &	0.858 &	-- &	0.831	\\	
&	0.9 &	oracle &	0.904 &	0.928 &	0.916 &	0.927 &	0.916 &	0.916 &	0.926 &	0.923 &	0.919 &	0.92 &	0.918 &	0.922 &	0.908 &	-- &	0.964	\\	
&	&	MoLP &	0.923 &	0.933 &	0.928 &	0.929 &	0.926 &	0.912 &	0.929 &	0.924 &	0.925 &	0.919 &	0.914 &	0.92 &	0.921 &	-- &	0.91	\\	
&	0.95 &	oracle &	0.96 &	0.974 &	0.974 &	0.972 &	0.972 &	0.972 &	0.97 &	0.971 &	0.97 &	0.974 &	0.973 &	0.968 &	0.964 &	-- &	0.988	\\	
&	&	MoLP &	0.959 &	0.962 &	0.962 &	0.961 &	0.955 &	0.948 &	0.958 &	0.957 &	0.954 &	0.952 &	0.945 &	0.954 &	0.956 &	-- &	0.945	\\	
\cmidrule(lr){2-3} \cmidrule(lr){4-17} \cmidrule(lr){18-18}
&	— &	detection &	0.918 &	0.902 &	0.908 &	0.913 &	0.921 &	0.926 &	0.92 &	0.91 &	0.908 &	0.901 &	0.916 &	0.915 &	0.921 &	-- &	0.526	\\	
\cmidrule(lr){1-3} \cmidrule(lr){4-17} \cmidrule(lr){18-18}
{\tt stairs10} &	0.8 &	oracle &	0.814 &	0.822 &	0.82 &	0.828 &	0.818 &	0.814 &	0.82 &	0.833 &	0.826 &	0.83 &	0.816 &	0.826 &	0.83 &	0.823 &	0.926	\\	
&	&	MoLP &	0.893 &	0.915 &	0.913 &	0.904 &	0.909 &	0.906 &	0.903 &	0.907 &	0.906 &	0.904 &	0.905 &	0.908 &	0.913 &	0.907 &	0.975	\\	
&	0.9 &	oracle &	0.912 &	0.924 &	0.93 &	0.932 &	0.933 &	0.93 &	0.92 &	0.936 &	0.928 &	0.927 &	0.92 &	0.926 &	0.938 &	0.926 &	0.974	\\	
&	&	MoLP &	0.951 &	0.965 &	0.97 &	0.966 &	0.963 &	0.966 &	0.963 &	0.971 &	0.97 &	0.957 &	0.959 &	0.966 &	0.97 &	0.963 &	0.987	\\	
&	0.95 &	oracle &	0.962 &	0.972 &	0.976 &	0.979 &	0.976 &	0.97 &	0.974 &	0.974 &	0.967 &	0.97 &	0.97 &	0.975 &	0.974 &	0.967 &	0.987	\\	
&	&	MoLP &	0.981 &	0.982 &	0.991 &	0.986 &	0.984 &	0.986 &	0.985 &	0.984 &	0.99 &	0.98 &	0.979 &	0.982 &	0.983 &	0.984 &	0.989	\\	
\cmidrule(lr){2-3} \cmidrule(lr){4-17} \cmidrule(lr){18-18}
&	— &	detection &	0.998 &	0.998 &	0.998 &	0.996 &	0.998 &	0.998 &	0.998 &	1 &	1 &	0.996 &	0.998 &	0.997 &	0.998 &	0.998 &	0.976	\\	
\bottomrule
\end{tabular}}
\end{table}

\clearpage

\subsection{$t_5$-distributed errors}
\label{sec:sim:t}

\subsubsection{Bootstrap CIs constructed with the oracle estimators in~\eqref{eq:cp:tilde}}

\begin{figure}[htbp]
\centering
\includegraphics[width=.8\textwidth]{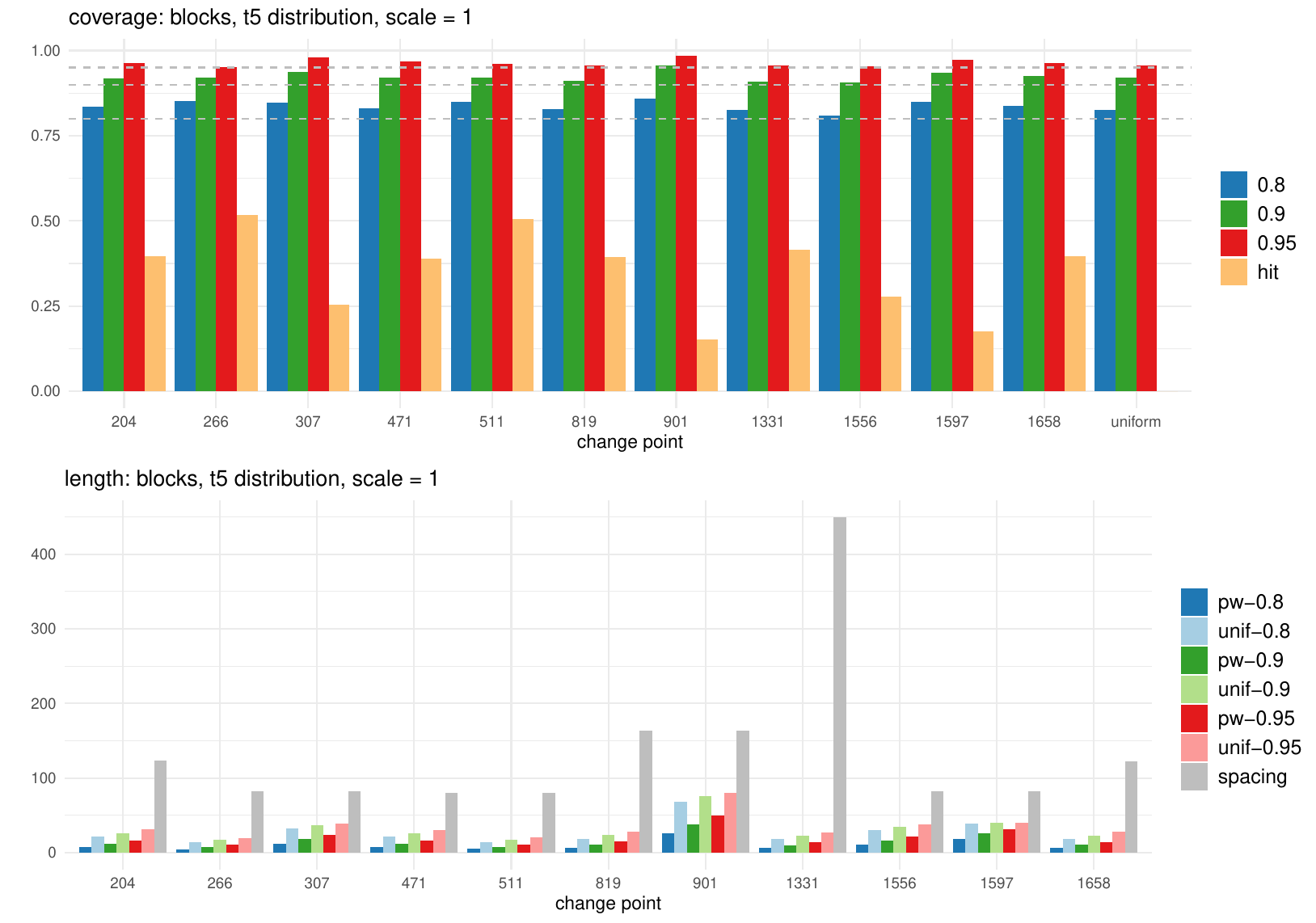}
\caption{{\tt blocks} with $\vartheta = 1$.}
\label{fig:t:blocks:one}
\end{figure}


\begin{figure}[htbp]
\centering
\includegraphics[width=.8\textwidth]{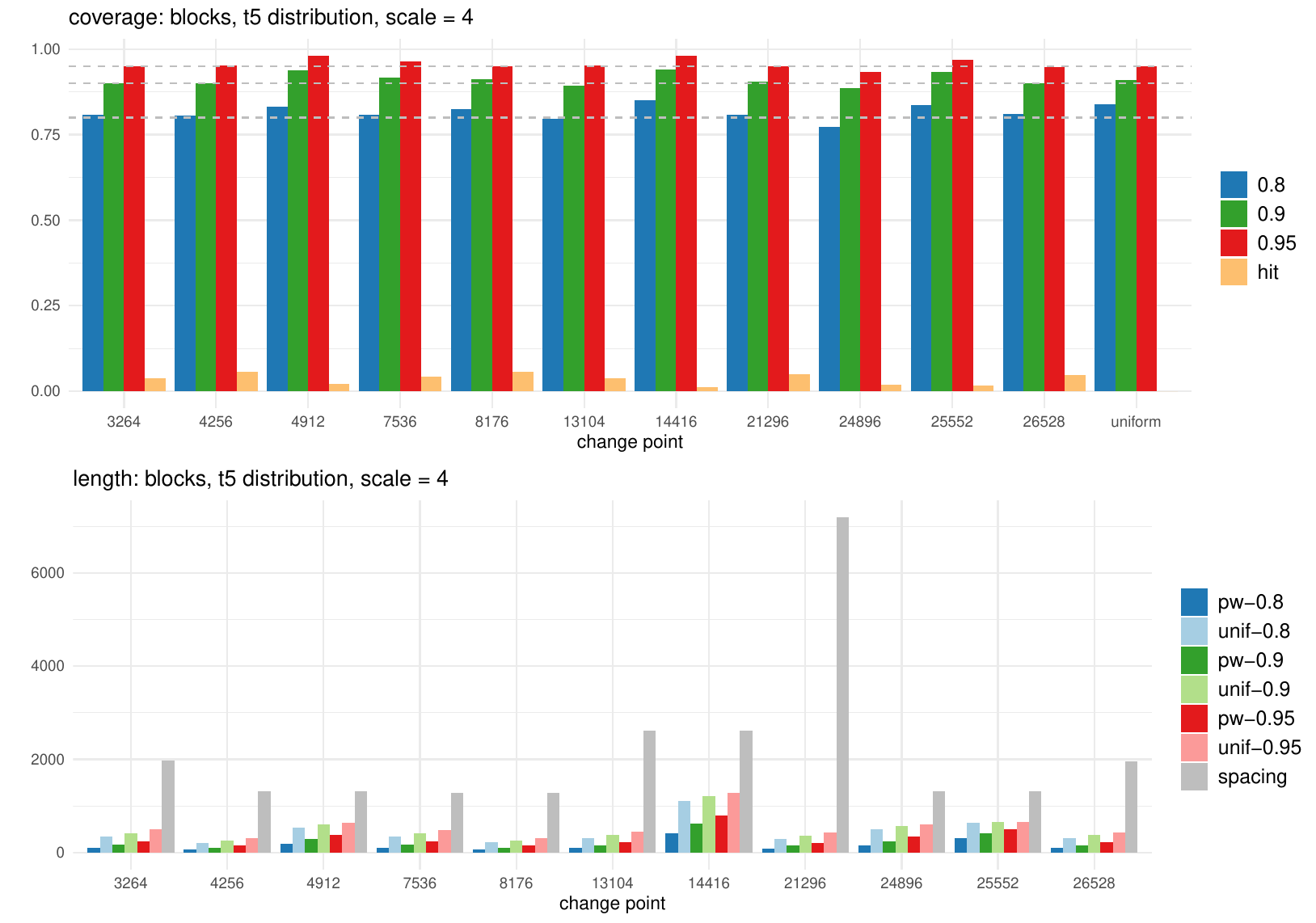}
\caption{{\tt blocks} with $\vartheta = 4$.}
\label{fig:t:blocks:four}
\end{figure}

\begin{figure}[htbp]
\centering
\includegraphics[width=.8\textwidth]{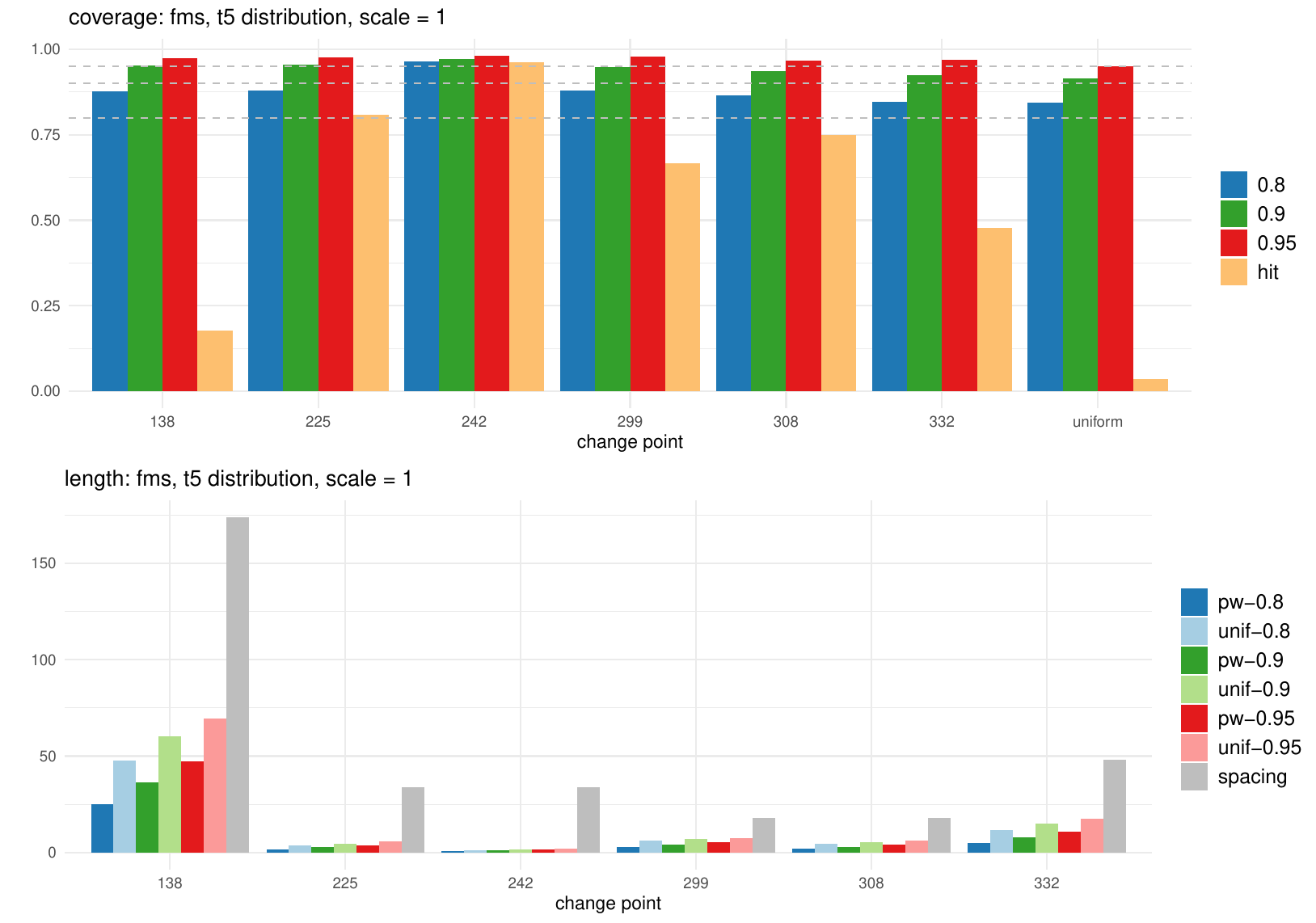}
\caption{Bootstrap CIs constructed with the oracle estimators in~\eqref{eq:cp:tilde}:
{\tt fms} with $\vartheta = 1$.}
\label{fig:t:fms:one}
\end{figure}


\begin{figure}[htbp]
\centering
\includegraphics[width=.8\textwidth]{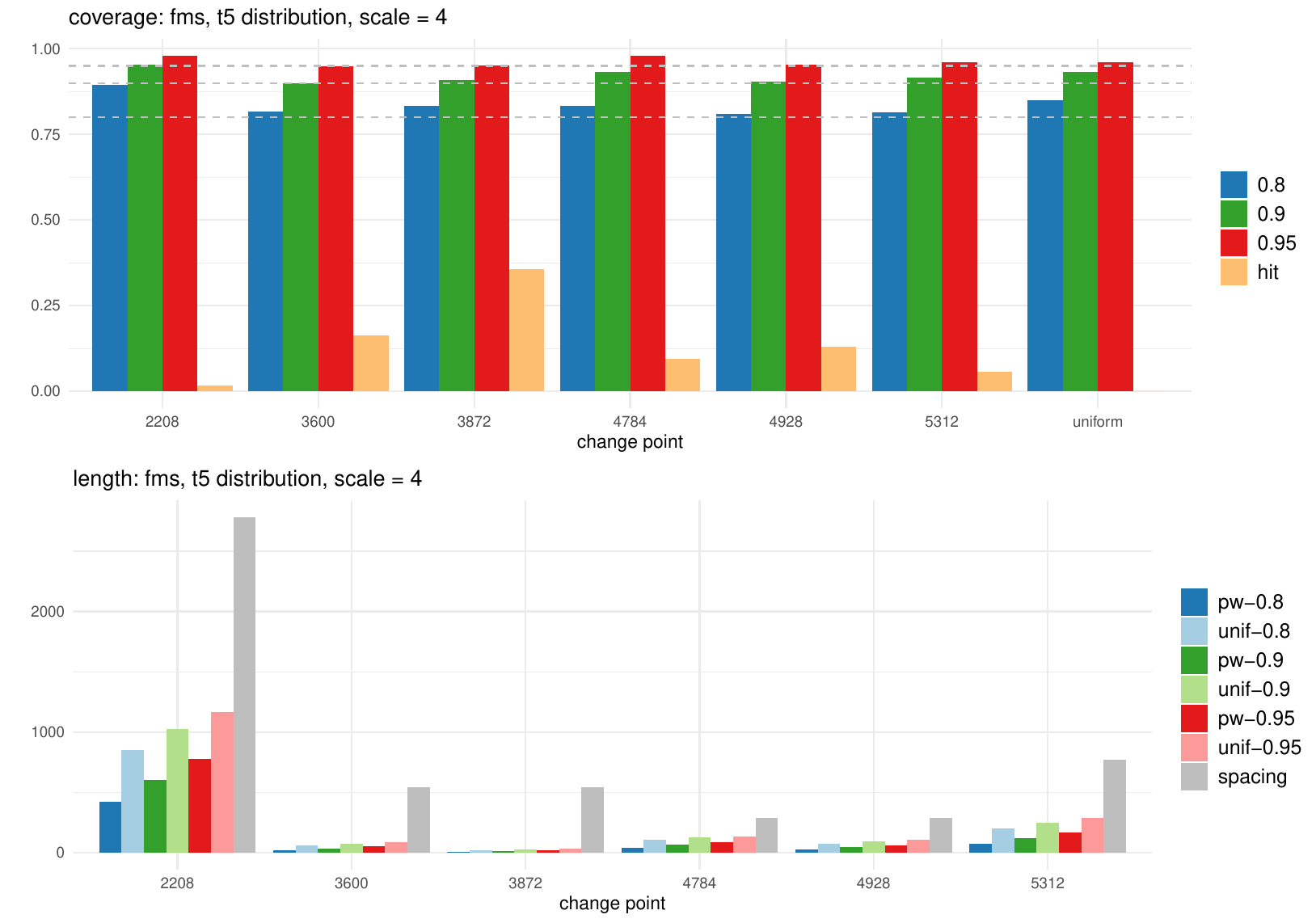}
\caption{Bootstrap CIs constructed with the oracle estimators in~\eqref{eq:cp:tilde}:
{\tt fms} with $\vartheta = 4$.}
\label{fig:t:fms:four}
\end{figure}

\begin{figure}[htbp]
\centering
\includegraphics[width=.8\textwidth]{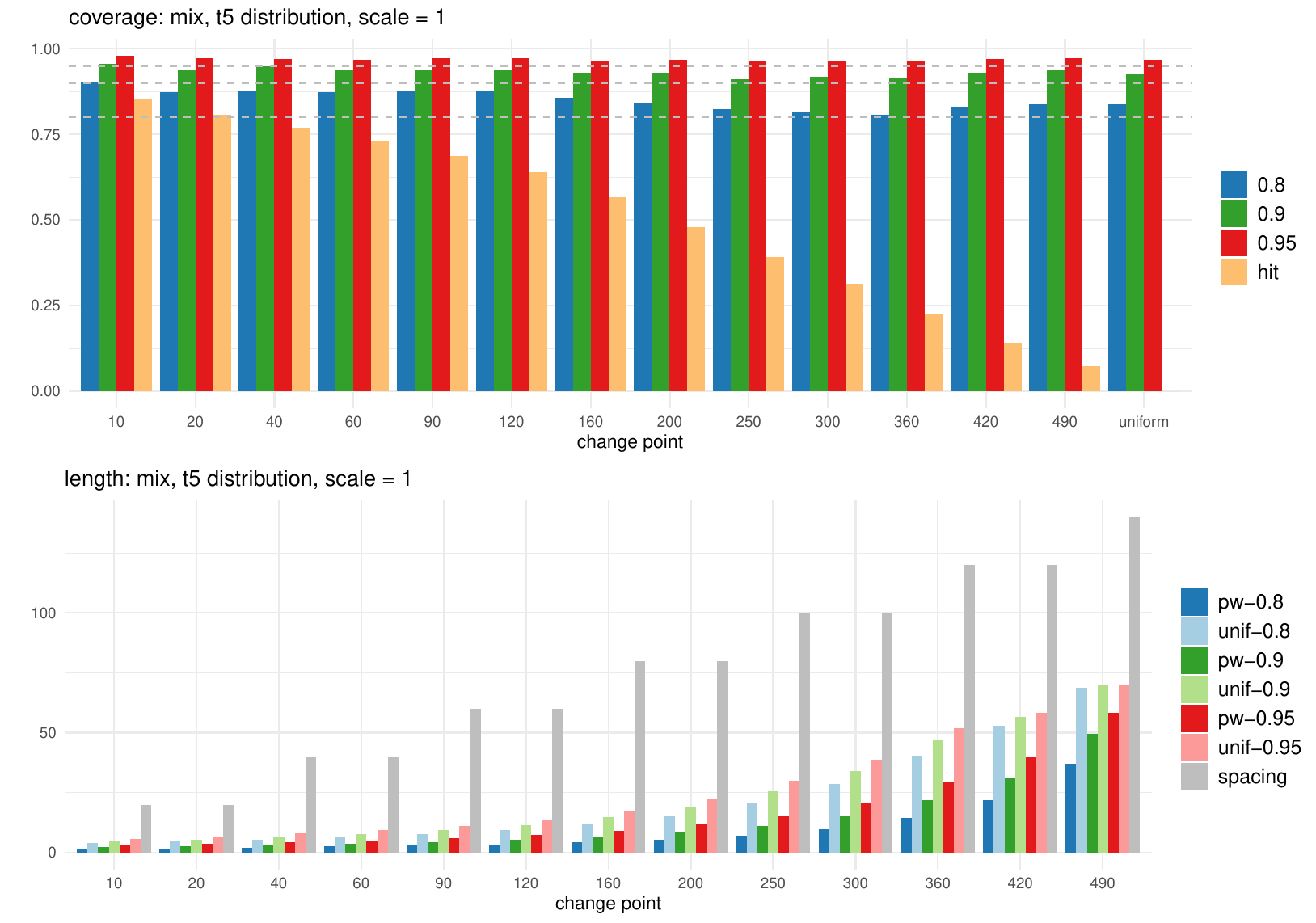}
\caption{Bootstrap CIs constructed with the oracle estimators in~\eqref{eq:cp:tilde}: {\tt mix} with $\vartheta =1$.}
\label{fig:t:mix:one}
\end{figure}

\begin{figure}[htbp]
\centering
\includegraphics[width=.8\textwidth]{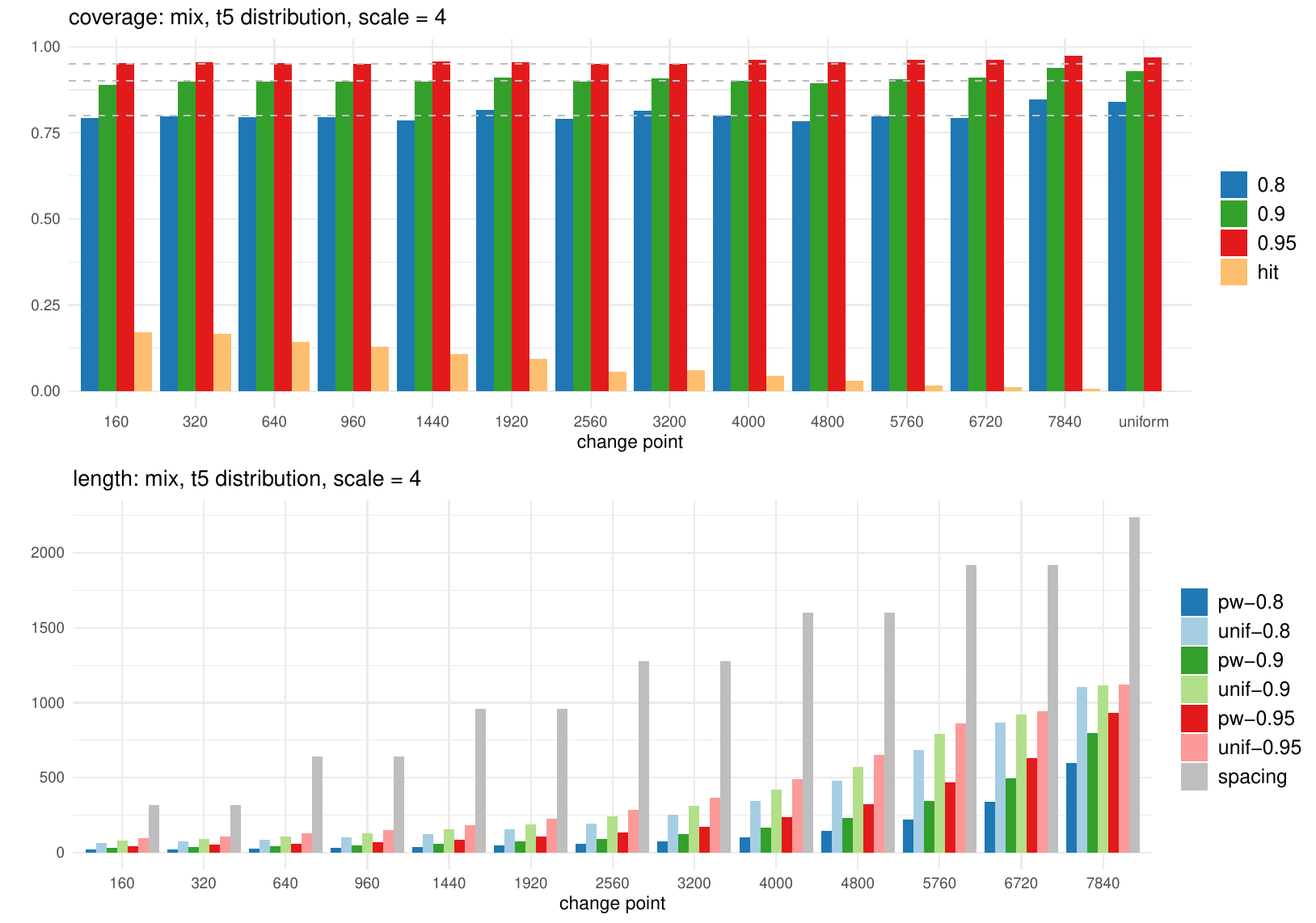}
\caption{Bootstrap CIs constructed with the oracle estimators in~\eqref{eq:cp:tilde}: {\tt mix} with $\vartheta = 4$.}
\label{fig:t:mix:four}
\end{figure}

\begin{figure}[htbp]
\centering
\includegraphics[width=.8\textwidth]{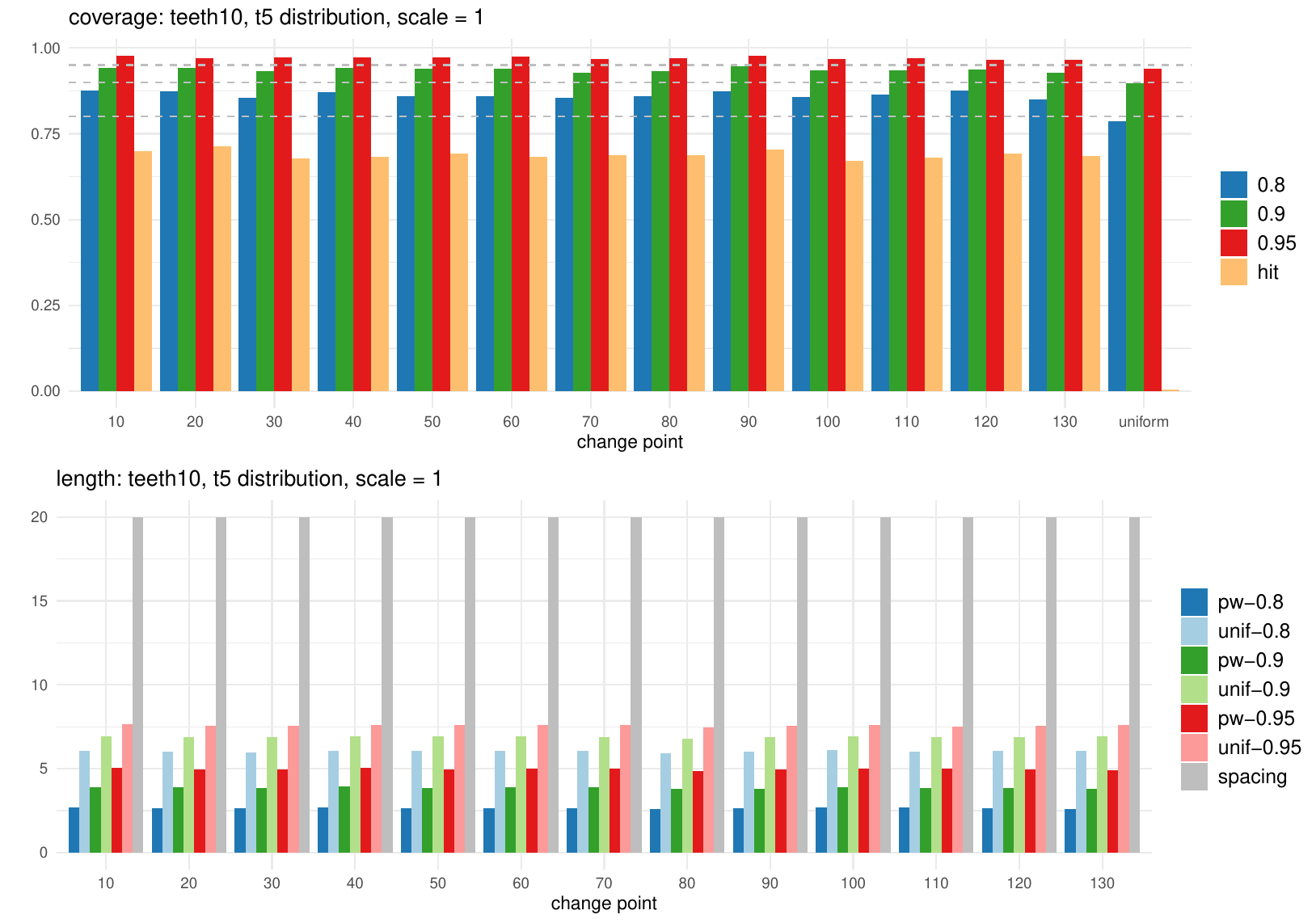}
\caption{Bootstrap CIs constructed with theoracle estimators in~\eqref{eq:cp:tilde}: {\tt teeth10} with $\vartheta = 1$.}
\label{fig:t:teeth10:one}
\end{figure}
%

\begin{figure}[htbp]
\centering
\includegraphics[width=.8\textwidth]{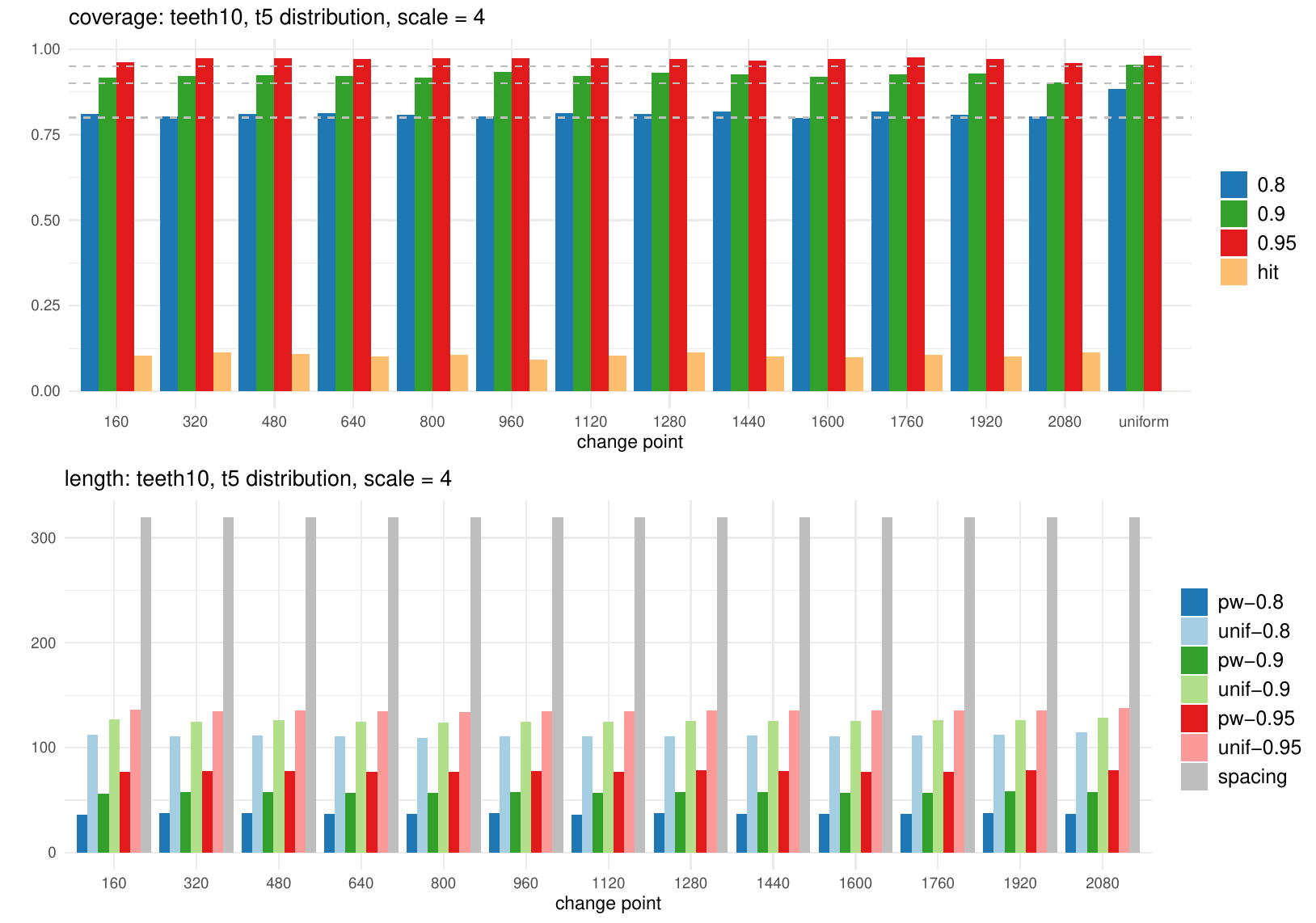}
\caption{Bootstrap CIs constructed with the oracle estimators in~\eqref{eq:cp:tilde}: {\tt teeth10} with $\vartheta = 4$.}
\label{fig:t:teeth10:four}
\end{figure}
\begin{figure}[htbp]
\centering
\includegraphics[width=.8\textwidth]{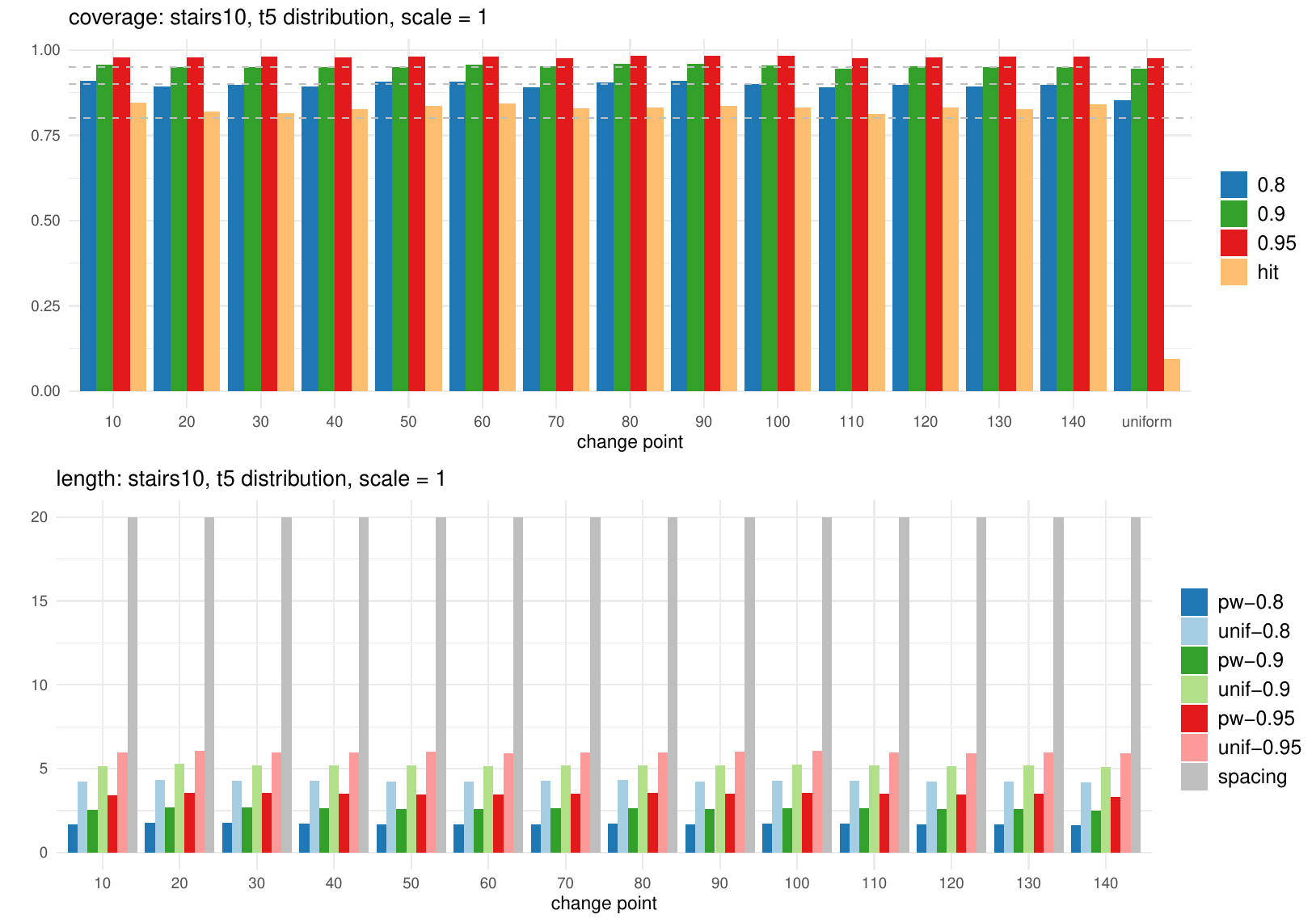}
\caption{Bootstrap CIs constructed with the oracle estimators in~\eqref{eq:cp:tilde}: {\tt stairs10} with $\vartheta = 1$.}
\label{fig:t:stairs10:one}
\end{figure}
%
%
\begin{figure}[htbp]
\centering
\includegraphics[width=.8\textwidth]{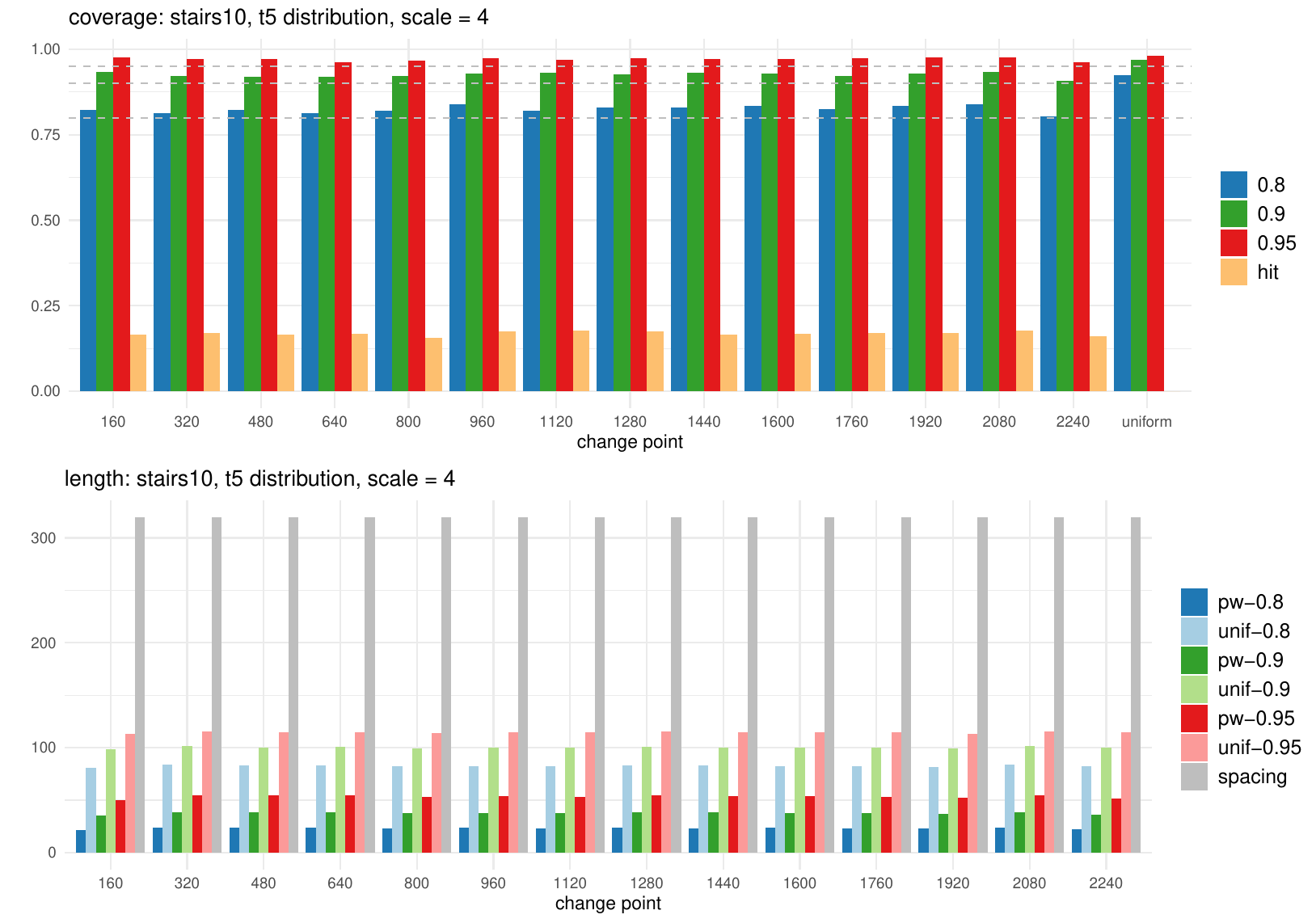}
\caption{Bootstrap CIs constructed with the oracle estimators in~\eqref{eq:cp:tilde}: {\tt stairs10} with $\vartheta = 4$.}
\label{fig:t:stairs10:four}
\end{figure}

\clearpage

\subsubsection{Bootstrap CIs constructed with model selection}

\begin{figure}[htbp]
\centering
\includegraphics[width=.8\textwidth]{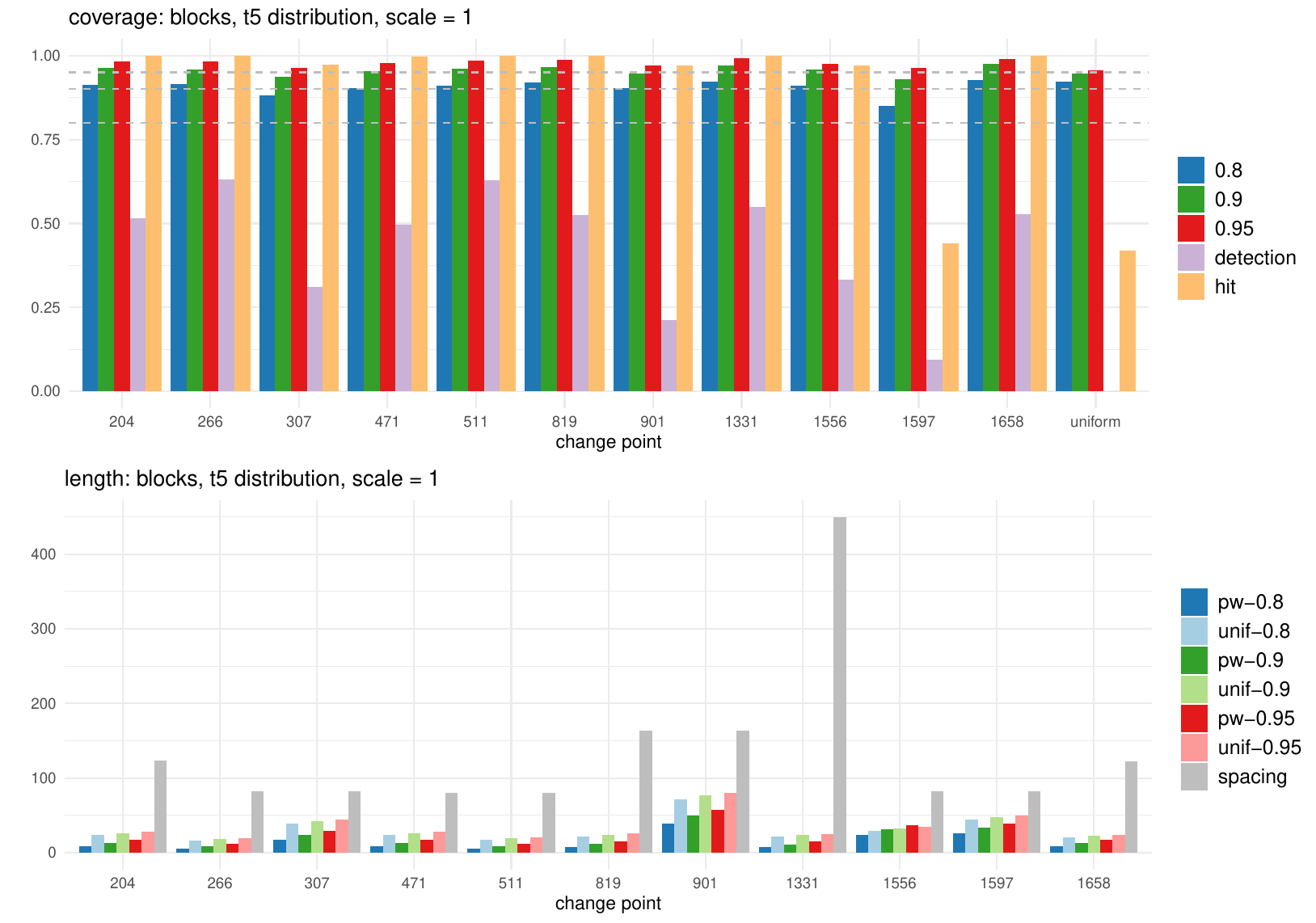}
\caption{Bootstrap CIs constructed with model selection: {\tt blocks} with $\vartheta = 1$.}
\label{fig:full:t:blocks:one}
\end{figure}


\begin{figure}[htbp]
\centering
\includegraphics[width=.8\textwidth]{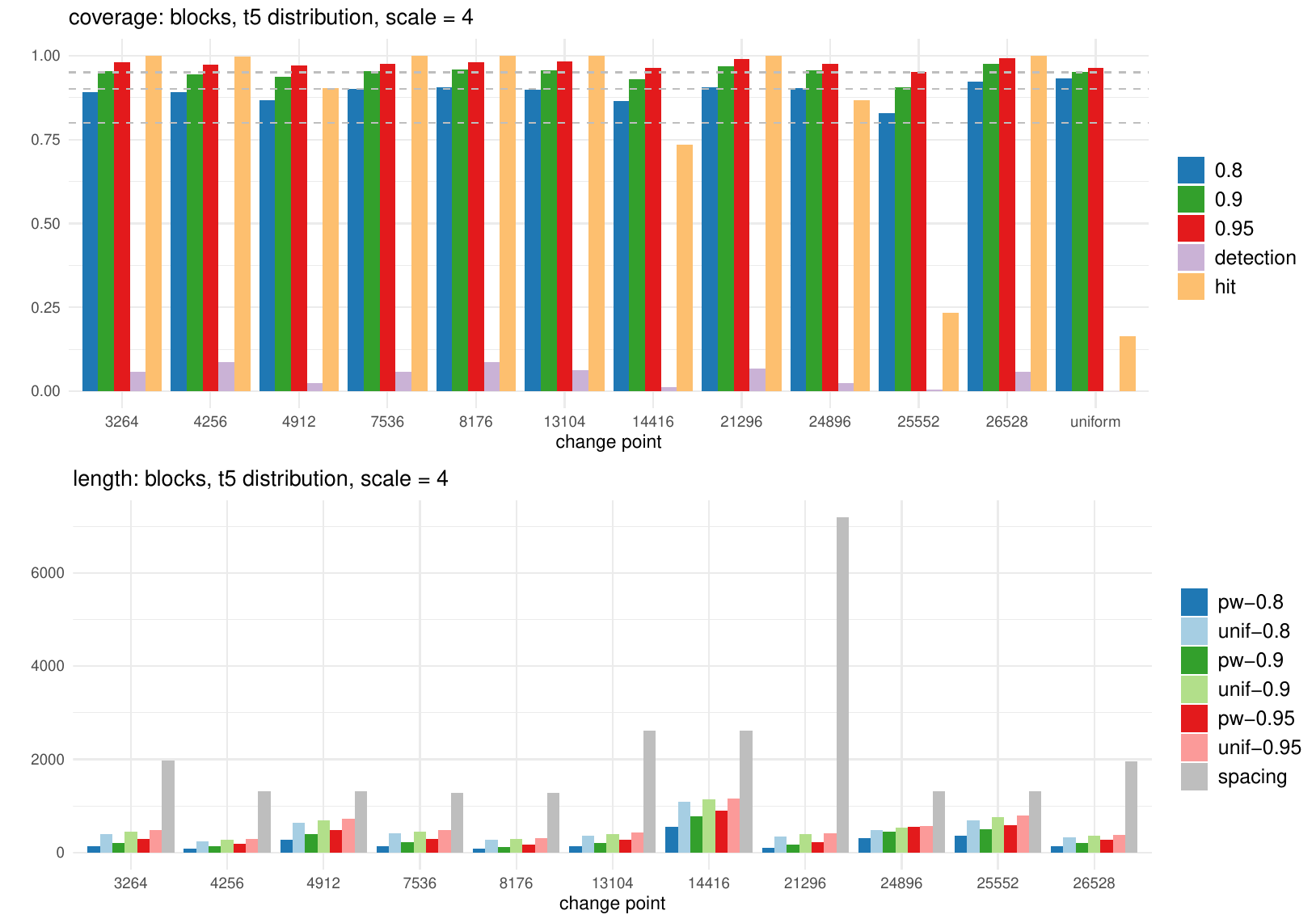}
\caption{Bootstrap CIs constructed with model selection: {\tt blocks} with $\vartheta = 4$.}
\label{fig:full:t:blocks:four}
\end{figure}

\begin{figure}[htbp]
\centering
\includegraphics[width=.8\textwidth]{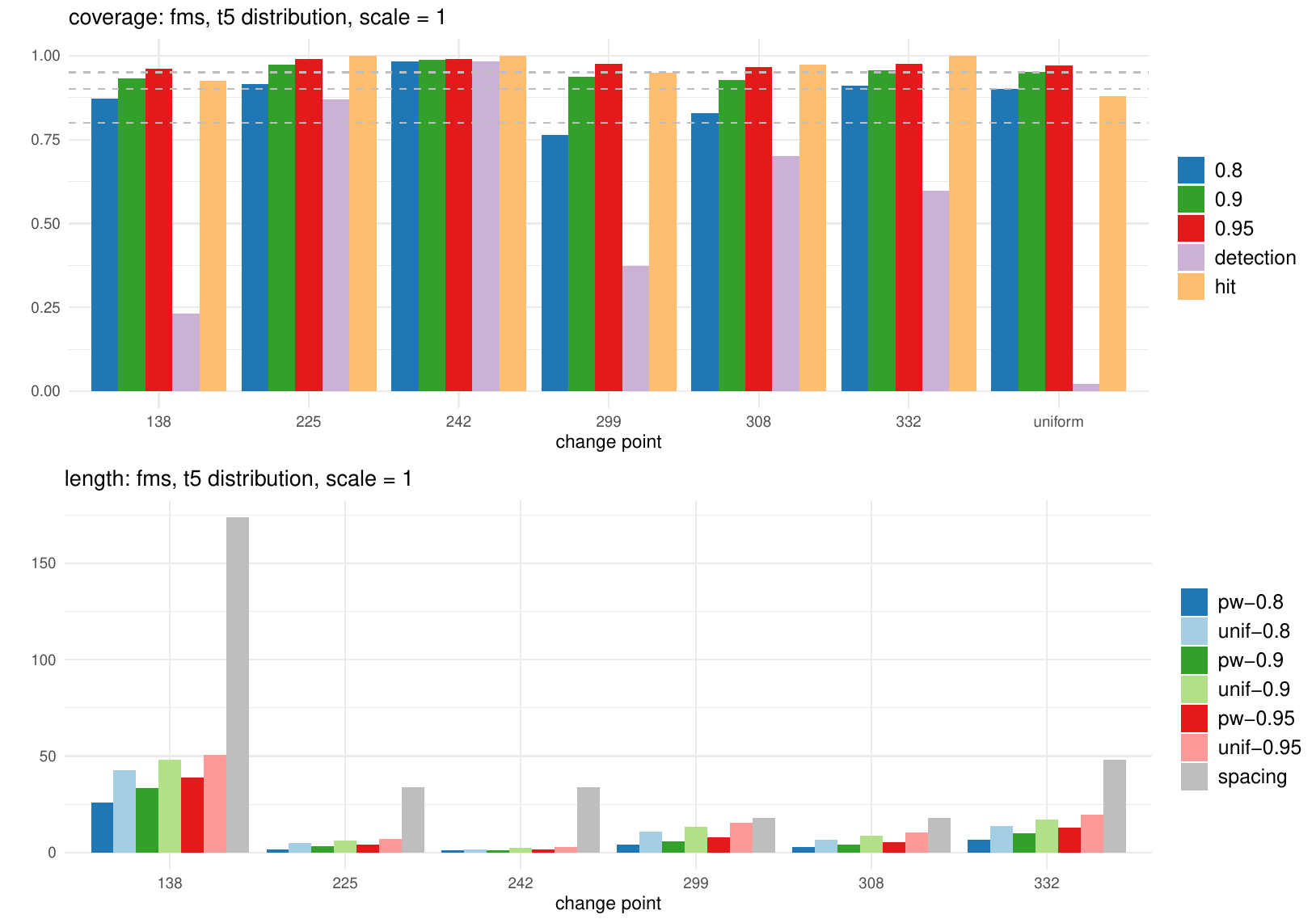}
\caption{Bootstrap CIs constructed with model selection: {\tt fms} with $\vartheta = 1$.}
\label{fig:full:t:fms:one}
\end{figure}


\begin{figure}[htbp]
\centering
\includegraphics[width=.8\textwidth]{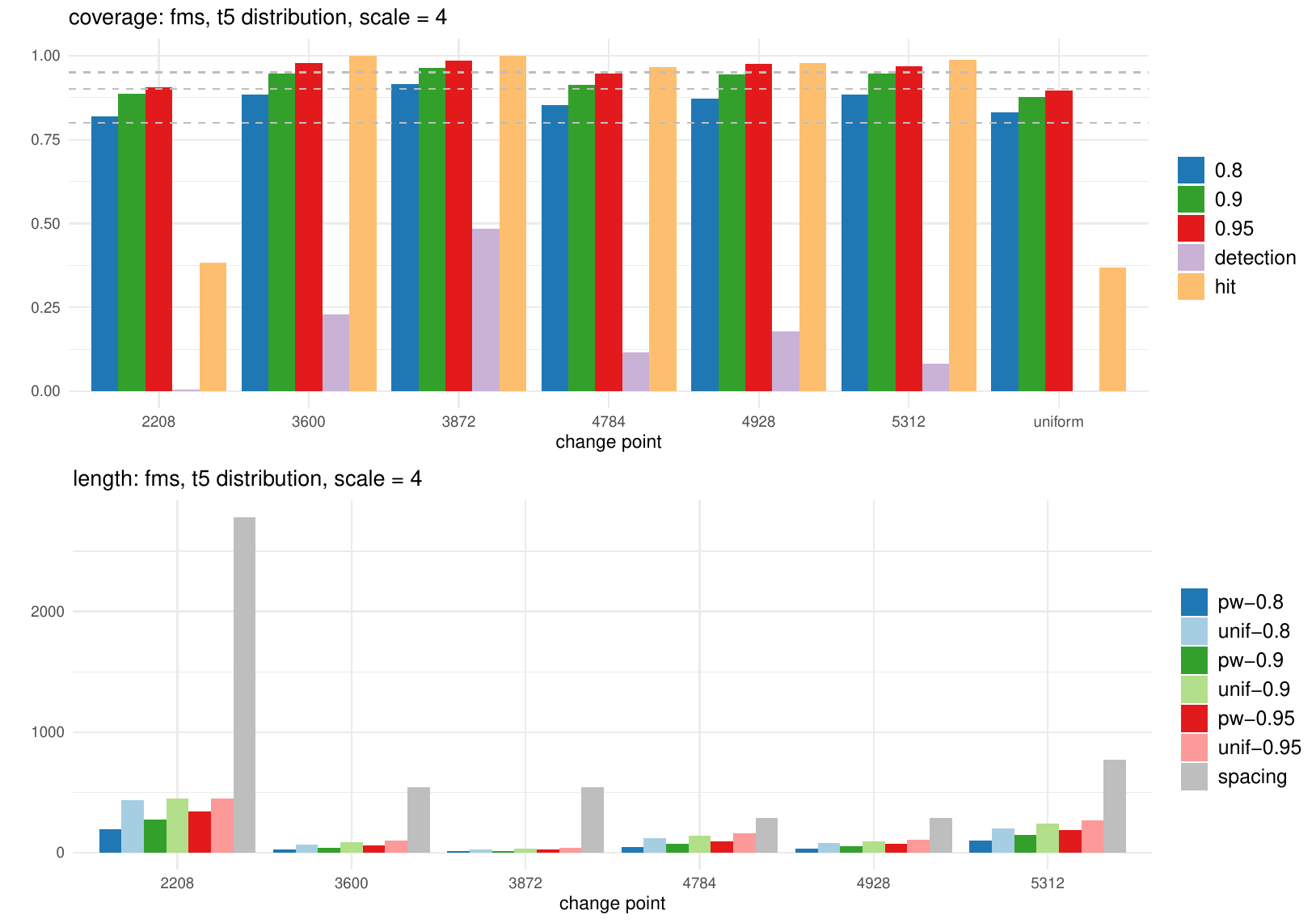}
\caption{Bootstrap CIs constructed with model selection: {\tt fms} with $\vartheta = 4$.}
\label{fig:full:t:fms:four}
\end{figure}

\begin{figure}[htbp]
\centering
\includegraphics[width=.8\textwidth]{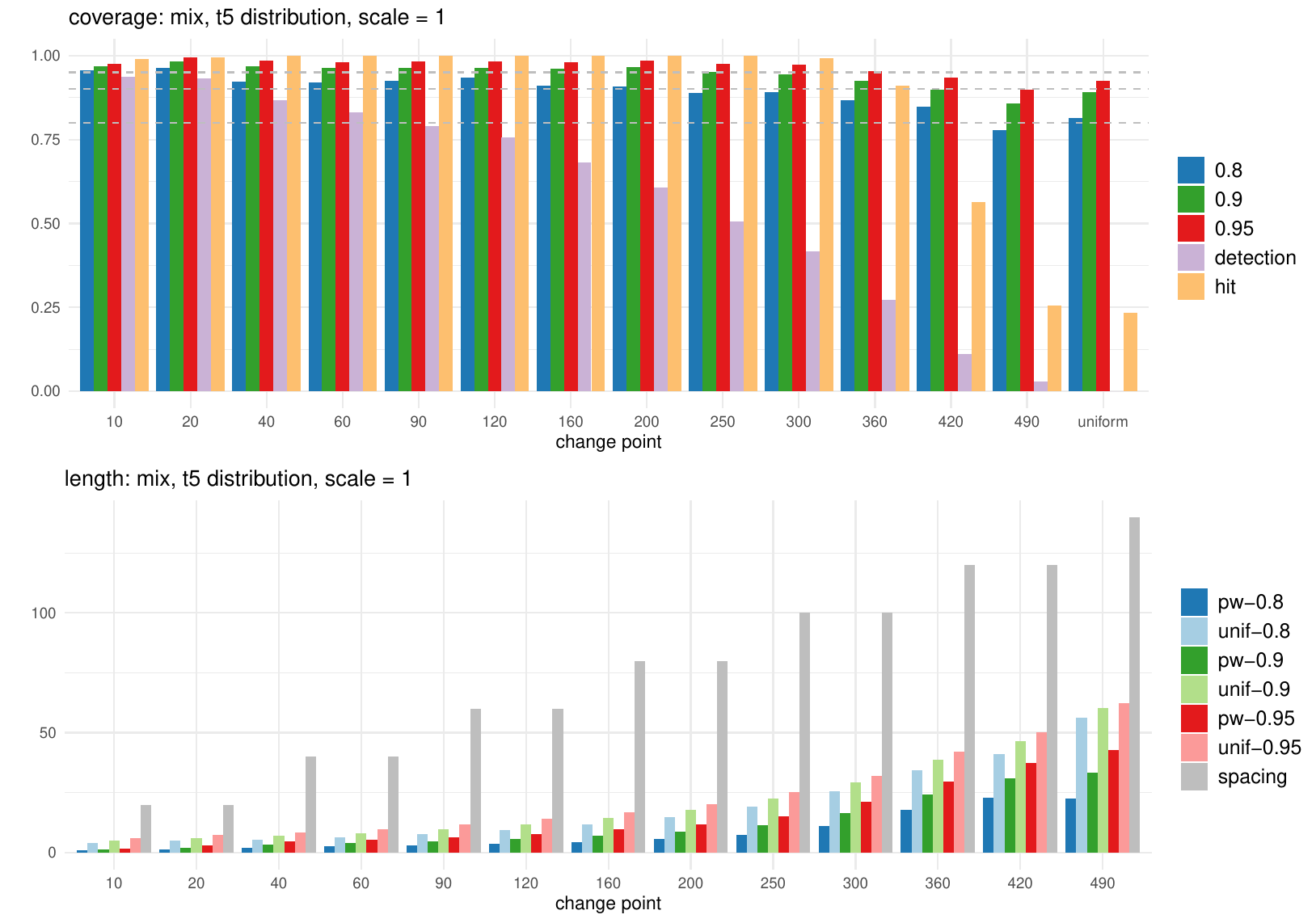}
\caption{Bootstrap CIs constructed with model selection: {\tt mix} with $\vartheta = 1$.}
\label{fig:full:t:mix:one}
\end{figure}

\begin{figure}[htbp]
\centering
\includegraphics[width=.8\textwidth]{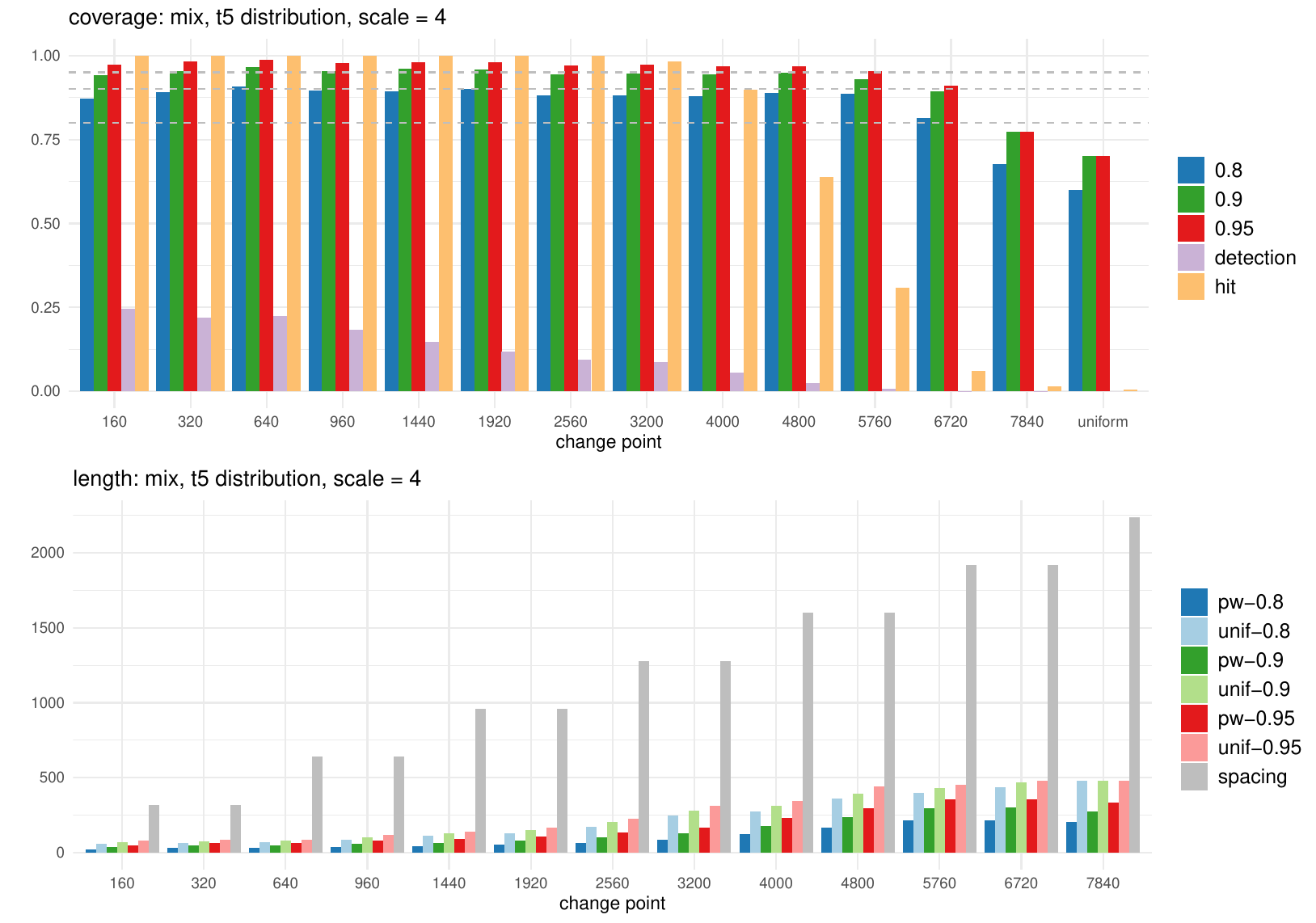}
\caption{Bootstrap CIs constructed with model selection: {\tt mix} with $\vartheta = 4$.}
\label{fig:full:t:mix:four}
\end{figure}

\begin{figure}[htbp]
\centering
\includegraphics[width=.8\textwidth]{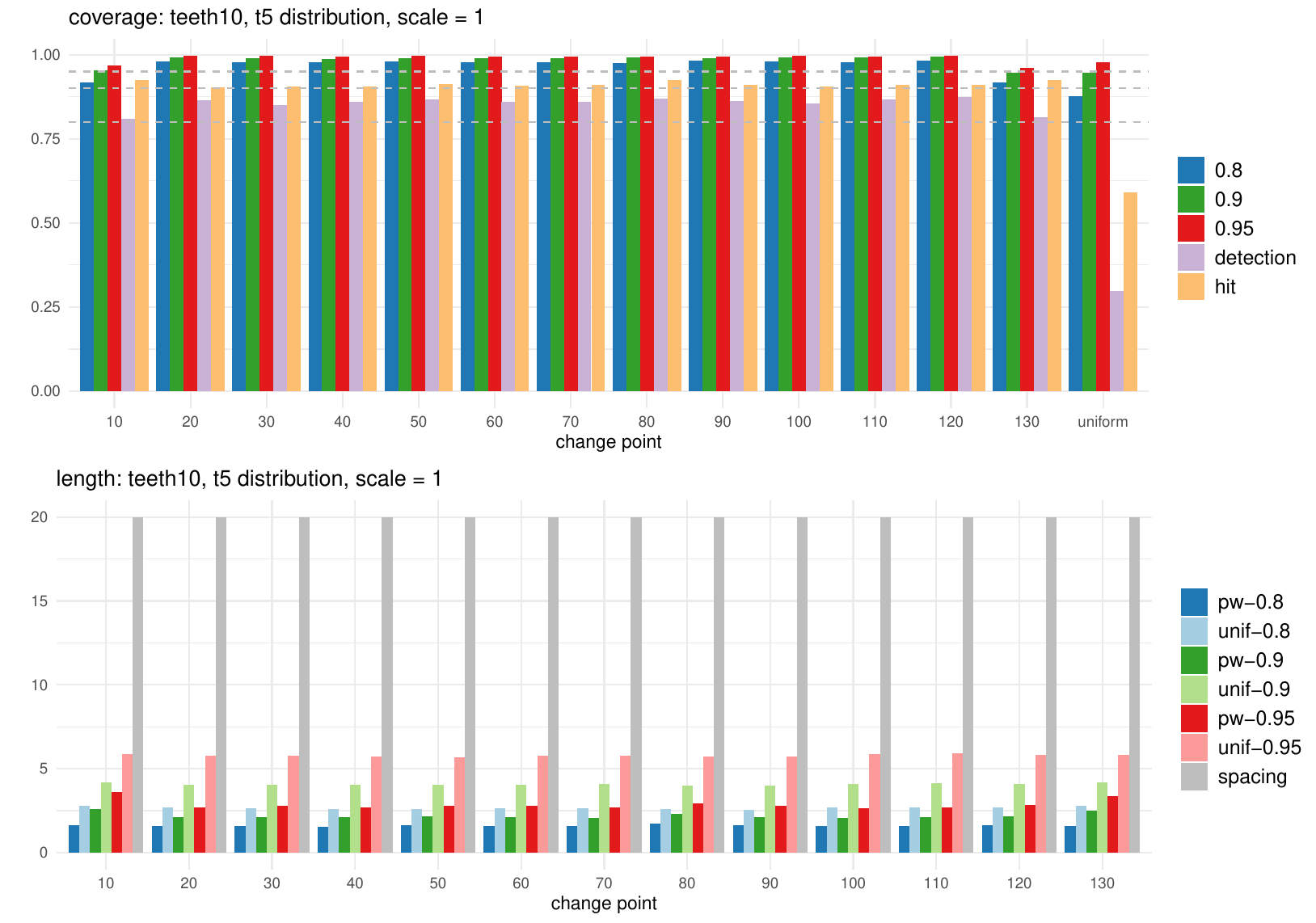}
\caption{Bootstrap CIs constructed with model selection: {\tt teeth10} with $\vartheta = 1$.}
\label{fig:full:t:teeth10:one}
\end{figure}
%
%
\begin{figure}[htbp]
\centering
\includegraphics[width=.8\textwidth]{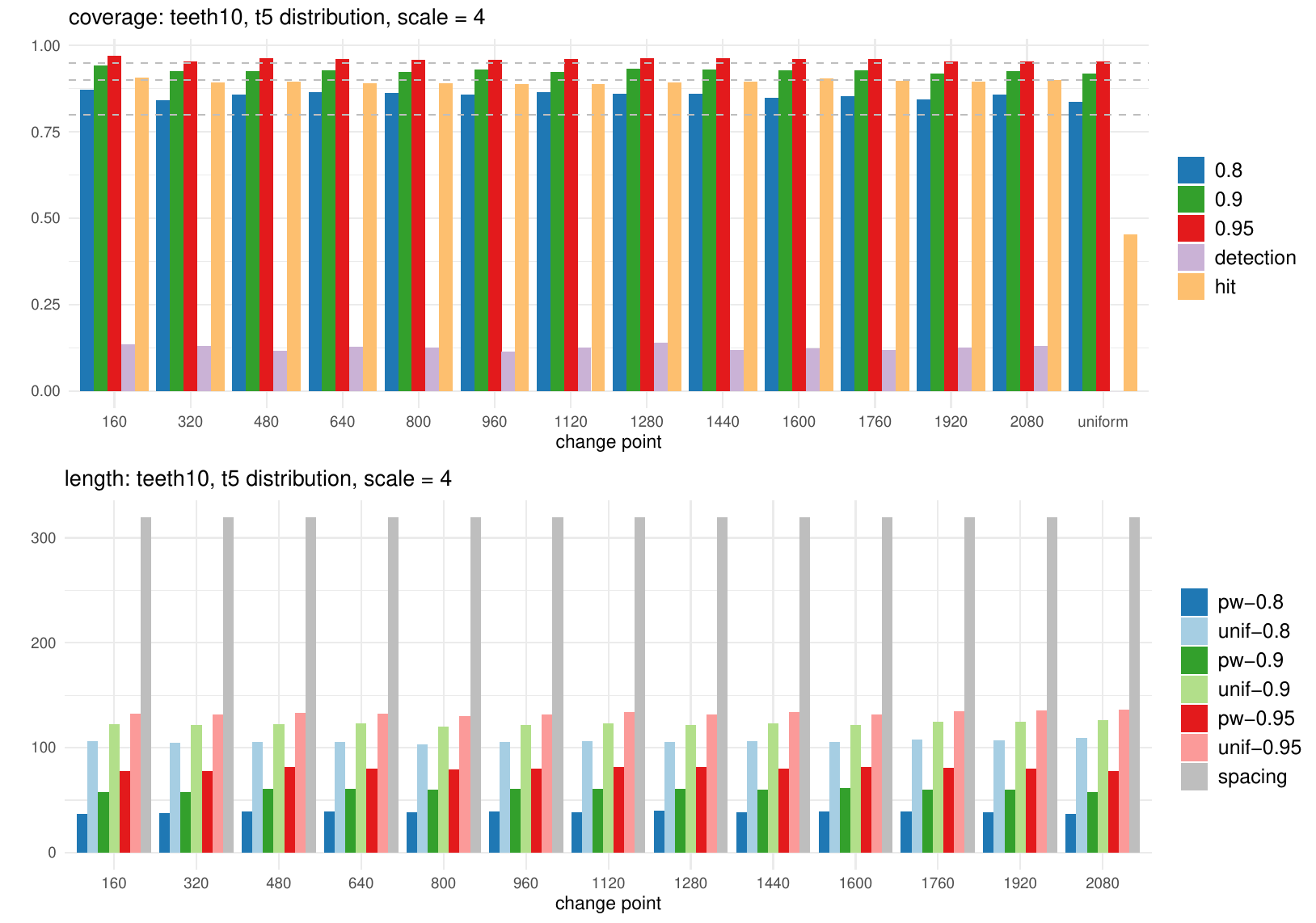}
\caption{Bootstrap CIs constructed with model selection: {\tt teeth10} with $\vartheta = 4$.}
\label{fig:full:t:teeth10:four}
\end{figure}

\begin{figure}[htbp]
\centering
\includegraphics[width=.8\textwidth]{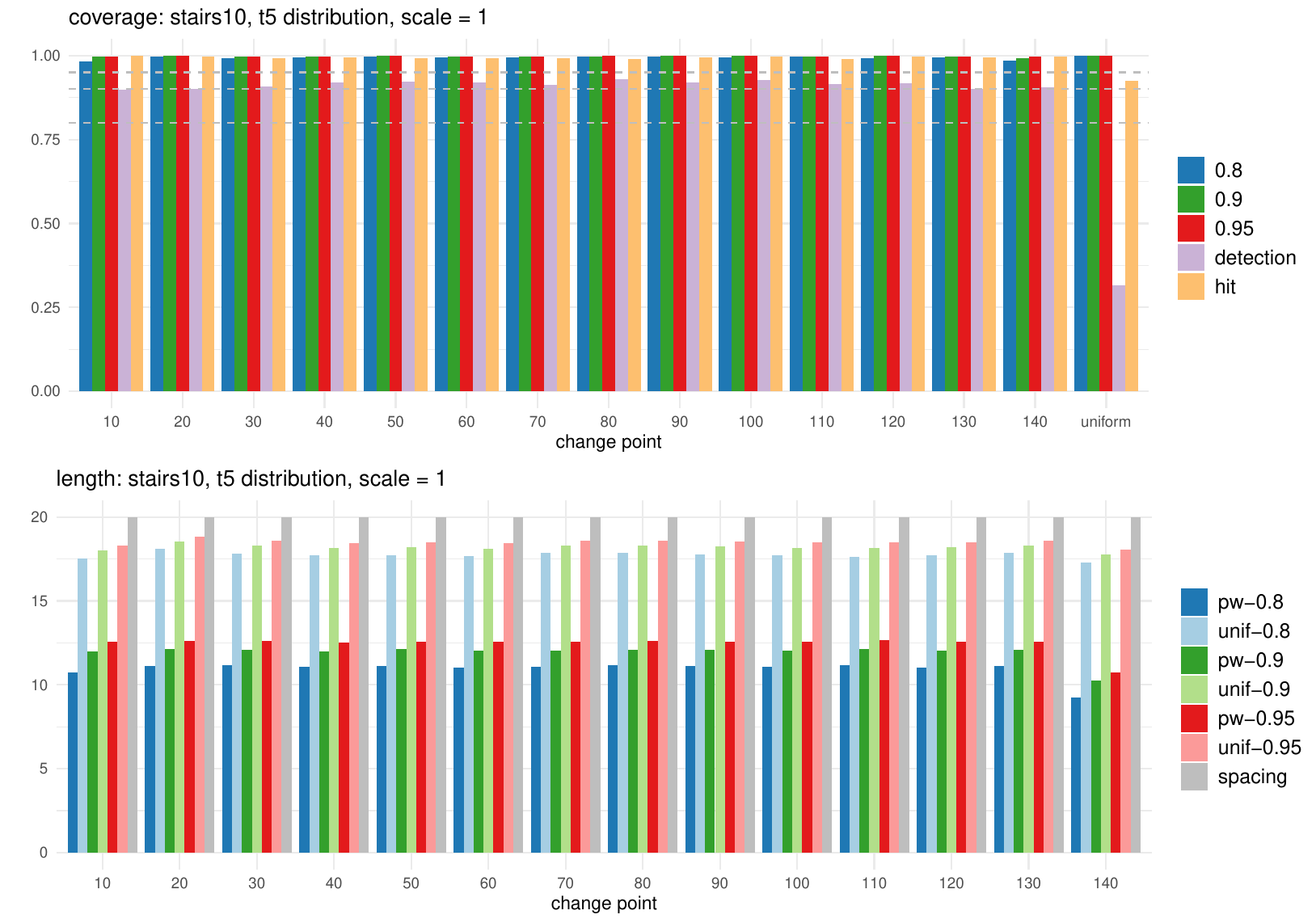}
\caption{Bootstrap CIs constructed with model selection: {\tt stairs10} with $\vartheta = 1$.}
\label{fig:full:t:stairs10:one}
\end{figure}
%
%
\begin{figure}[htbp]
\centering
\includegraphics[width=.8\textwidth]{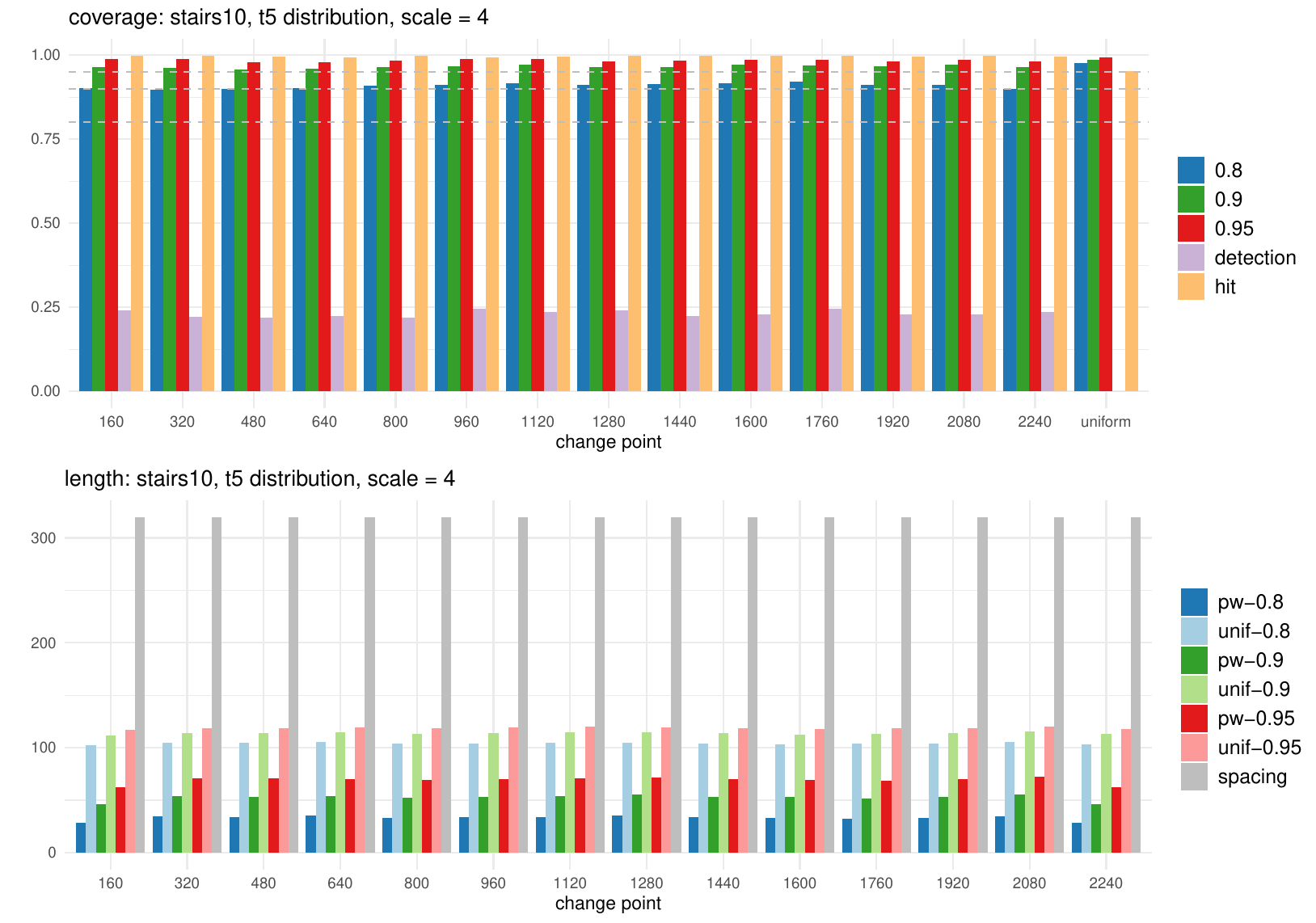}
\caption{Bootstrap CIs constructed with model selection: {\tt stairs10} with $\vartheta = 4$.}
\label{fig:full:t:stairs10:four}
\end{figure}

\clearpage

\subsubsection{Comparison of coverage}

\begin{table}[ht]
\caption{Average coverage of the bootstrap CIs constructed with the oracle estimators
and the estimators obtained from MoLP, when $\vartheta = 1$.
We also report the proportion of realisations where individual change points are detected (by MoLP)
and where all change points are correctly detected.}
\label{table:cov:t:1}
\centering
\resizebox{\columnwidth}{!}{
\begin{tabular}{ccc ccc ccc ccc ccc cc c}
\hline
\hline
test signal &	$1 - \alpha$ &	estimator &	$\cp_1$ &	$\cp_2$ &	$\cp_3$ &	$\cp_4$ &	$\cp_5$ &	$\cp_6$ &	$\cp_7$ &	$\cp_8$ &	$\cp_9$ &	$\cp_{10}$ &	$\cp_{11}$ &	$\cp_{12}$ &	$\cp_{13}$ &	$\cp_{14}$ &	uniform	\\	
\cmidrule(lr){1-3} \cmidrule(lr){4-17} \cmidrule(lr){18-18}
{\tt blocks} &	0.8 &	oracle &	0.835 &	0.851 &	0.847 &	0.831 &	0.848 &	0.828 &	0.858 &	0.826 &	0.81 &	0.849 &	0.838 &	-- &	-- &	-- &	0.826	\\	
&	&	MoLP &	0.913 &	0.914 &	0.88 &	0.904 &	0.91 &	0.919 &	0.903 &	0.922 &	0.909 &	0.851 &	0.928 &	-- &	-- &	-- &	0.923	\\	
&	0.9 &	oracle &	0.918 &	0.92 &	0.936 &	0.92 &	0.92 &	0.91 &	0.955 &	0.91 &	0.907 &	0.934 &	0.925 &	-- &	-- &	-- &	0.92	\\	
&	&	MoLP &	0.963 &	0.959 &	0.938 &	0.954 &	0.961 &	0.967 &	0.946 &	0.97 &	0.958 &	0.93 &	0.974 &	-- &	-- &	-- &	0.947	\\	
&	0.95 &	oracle &	0.964 &	0.952 &	0.98 &	0.968 &	0.96 &	0.955 &	0.985 &	0.956 &	0.954 &	0.973 &	0.963 &	-- &	-- &	-- &	0.956	\\	
&	&	MoLP &	0.983 &	0.982 &	0.963 &	0.979 &	0.984 &	0.988 &	0.971 &	0.992 &	0.975 &	0.963 &	0.99 &	-- &	-- &	-- &	0.957	\\	\cmidrule(lr){2-3} \cmidrule(lr){4-17} \cmidrule(lr){18-18}
&	— &	detection &	1 &	0.999 &	0.974 &	0.998 &	0.999 &	1 &	0.971 &	1 &	0.97 &	0.441 &	1 &	-- &	-- &	-- &	0.419	\\	
\cmidrule(lr){1-3} \cmidrule(lr){4-17} \cmidrule(lr){18-18}
{\tt fms} &	0.8 &	oracle &	0.876 &	0.878 &	0.964 &	0.88 &	0.866 &	0.846 &	-- &	-- &	-- &	-- &	-- &	-- &	-- &	-- &	0.843	\\	
&	&	MoLP &	0.873 &	0.915 &	0.984 &	0.763 &	0.828 &	0.91 &	-- &	-- &	-- &	-- &	-- &	-- &	-- &	-- &	0.9	\\	
&	0.9 &	oracle &	0.954 &	0.956 &	0.97 &	0.948 &	0.935 &	0.924 &	-- &	-- &	-- &	-- &	-- &	-- &	-- &	-- &	0.916	\\	
&	&	MoLP &	0.933 &	0.974 &	0.986 &	0.936 &	0.927 &	0.956 &	-- &	-- &	-- &	-- &	-- &	-- &	-- &	-- &	0.95	\\	
&	0.95 &	oracle &	0.974 &	0.976 &	0.982 &	0.979 &	0.966 &	0.968 &	-- &	-- &	-- &	-- &	-- &	-- &	-- &	-- &	0.95	\\	
&	&	MoLP &	0.961 &	0.99 &	0.99 &	0.975 &	0.967 &	0.976 &	-- &	-- &	-- &	-- &	-- &	-- &	-- &	-- &	0.97	\\	\cmidrule(lr){2-3} \cmidrule(lr){4-17} \cmidrule(lr){18-18}
&	— &	detection &	0.924 &	1 &	1 &	0.95 &	0.972 &	0.998 &	-- &	-- &	-- &	-- &	-- &	-- &	-- &	-- &	0.88	\\	
\cmidrule(lr){1-3} \cmidrule(lr){4-17} \cmidrule(lr){18-18}
{\tt mix} &	0.8 &	oracle &	0.904 &	0.873 &	0.878 &	0.874 &	0.876 &	0.876 &	0.858 &	0.84 &	0.823 &	0.814 &	0.806 &	0.828 &	0.838 &	-- &	0.838	\\	
&	&	MoLP &	0.955 &	0.964 &	0.922 &	0.919 &	0.924 &	0.935 &	0.91 &	0.907 &	0.89 &	0.891 &	0.868 &	0.848 &	0.777 &	-- &	0.814	\\	
&	0.9 &	oracle &	0.957 &	0.94 &	0.949 &	0.938 &	0.936 &	0.938 &	0.93 &	0.93 &	0.912 &	0.919 &	0.916 &	0.93 &	0.938 &	-- &	0.926	\\	
&	&	MoLP &	0.967 &	0.983 &	0.969 &	0.963 &	0.962 &	0.964 &	0.962 &	0.965 &	0.951 &	0.945 &	0.925 &	0.898 &	0.857 &	-- &	0.89	\\	
&	0.95 &	oracle &	0.98 &	0.973 &	0.97 &	0.968 &	0.972 &	0.973 &	0.966 &	0.968 &	0.963 &	0.962 &	0.962 &	0.97 &	0.972 &	-- &	0.968	\\	
&	&	MoLP &	0.975 &	0.994 &	0.984 &	0.98 &	0.983 &	0.982 &	0.98 &	0.986 &	0.975 &	0.973 &	0.955 &	0.934 &	0.898 &	-- &	0.924	\\	\cmidrule(lr){2-3} \cmidrule(lr){4-17} \cmidrule(lr){18-18}
&	— &	detection &	0.991 &	0.994 &	1 &	1 &	1 &	1 &	1 &	1 &	1 &	0.992 &	0.912 &	0.564 &	0.256 &	-- &	0.234	\\	
\cmidrule(lr){1-3} \cmidrule(lr){4-17} \cmidrule(lr){18-18}
{\tt teeth10} &	0.8 &	oracle &	0.876 &	0.873 &	0.854 &	0.872 &	0.86 &	0.86 &	0.854 &	0.86 &	0.874 &	0.858 &	0.864 &	0.876 &	0.849 &	-- &	0.788	\\	
&	&	MoLP &	0.917 &	0.981 &	0.977 &	0.978 &	0.98 &	0.978 &	0.978 &	0.976 &	0.981 &	0.981 &	0.978 &	0.984 &	0.917 &	-- &	0.877	\\	
&	0.9 &	oracle &	0.942 &	0.941 &	0.934 &	0.941 &	0.938 &	0.939 &	0.928 &	0.934 &	0.946 &	0.934 &	0.934 &	0.936 &	0.928 &	-- &	0.896	\\	
&	&	MoLP &	0.954 &	0.992 &	0.99 &	0.988 &	0.99 &	0.99 &	0.991 &	0.992 &	0.99 &	0.993 &	0.992 &	0.993 &	0.947 &	-- &	0.947	\\	
&	0.95 &	oracle &	0.977 &	0.97 &	0.972 &	0.972 &	0.974 &	0.974 &	0.968 &	0.97 &	0.978 &	0.967 &	0.97 &	0.966 &	0.966 &	-- &	0.938	\\	
&	&	MoLP &	0.969 &	0.997 &	0.997 &	0.996 &	0.996 &	0.994 &	0.996 &	0.995 &	0.996 &	0.997 &	0.996 &	0.996 &	0.96 &	-- &	0.979	\\	\cmidrule(lr){2-3} \cmidrule(lr){4-17} \cmidrule(lr){18-18}
&	— &	detection &	0.926 &	0.904 &	0.906 &	0.905 &	0.912 &	0.908 &	0.911 &	0.924 &	0.91 &	0.904 &	0.912 &	0.911 &	0.926 &	-- &	0.592	\\	
\cmidrule(lr){1-3} \cmidrule(lr){4-17} \cmidrule(lr){18-18}
{\tt stairs10} &	0.8 &	oracle &	0.909 &	0.892 &	0.899 &	0.892 &	0.906 &	0.907 &	0.89 &	0.904 &	0.91 &	0.9 &	0.89 &	0.897 &	0.894 &	0.898 &	0.854	\\	
&	&	MoLP &	0.983 &	0.996 &	0.993 &	0.995 &	0.998 &	0.995 &	0.995 &	0.997 &	0.996 &	0.995 &	0.996 &	0.993 &	0.994 &	0.986 &	1	\\	
&	0.9 &	oracle &	0.956 &	0.95 &	0.95 &	0.95 &	0.95 &	0.958 &	0.952 &	0.959 &	0.96 &	0.954 &	0.944 &	0.952 &	0.95 &	0.95 &	0.946	\\	
&	&	MoLP &	0.996 &	0.998 &	0.996 &	0.996 &	1 &	0.996 &	0.997 &	0.998 &	0.998 &	0.998 &	0.998 &	0.998 &	0.997 &	0.992 &	1	\\	
&	0.95 &	oracle &	0.979 &	0.978 &	0.98 &	0.978 &	0.982 &	0.982 &	0.977 &	0.984 &	0.984 &	0.982 &	0.976 &	0.98 &	0.982 &	0.98 &	0.976	\\	
&	&	MoLP &	0.997 &	0.998 &	0.997 &	0.997 &	1 &	0.998 &	0.997 &	0.998 &	0.998 &	0.999 &	0.998 &	0.999 &	0.998 &	0.996 &	1	\\	\cmidrule(lr){2-3} \cmidrule(lr){4-17} \cmidrule(lr){18-18}
&	— &	detection &	0.998 &	0.996 &	0.992 &	0.994 &	0.993 &	0.993 &	0.992 &	0.991 &	0.994 &	0.996 &	0.989 &	0.996 &	0.996 &	0.996 &	0.924	\\	
\bottomrule
\end{tabular}}
\end{table}

\begin{table}[ht]
\caption{Average coverage of the bootstrap CIs constructed with the oracle estimators
and the estimators obtained from MoLP, when $\vartheta = 4$.
We also report the proportion of realisations where individual change points are detected (by MoLP)
and where all change points are correctly detected.}
\label{table:cov:t:4}
\centering
\resizebox{\columnwidth}{!}{
\begin{tabular}{ccc  ccc ccc ccc ccc cc  c}
\toprule
test signal &	$1 - \alpha$ &	estimator &	$\cp_1$ &	$\cp_2$ &	$\cp_3$ &	$\cp_4$ &	$\cp_5$ &	$\cp_6$ &	$\cp_7$ &	$\cp_8$ &	$\cp_9$ &	$\cp_{10}$ &	$\cp_{11}$ &	$\cp_{12}$ &	$\cp_{13}$ &	$\cp_{14}$ &	uniform	\\	
\cmidrule(lr){1-3} \cmidrule(lr){4-17} \cmidrule(lr){18-18}
{\tt blocks} &	0.8 &	oracle &	0.808 &	0.806 &	0.831 &	0.809 &	0.824 &	0.797 &	0.85 &	0.808 &	0.774 &	0.836 &	0.811 &	-- &	-- &	-- &	0.838	\\	
&	&	MoLP &	0.892 &	0.89 &	0.867 &	0.9 &	0.905 &	0.898 &	0.866 &	0.906 &	0.904 &	0.83 &	0.923 &	-- &	-- &	-- &	0.933	\\	
&	0.9 &	oracle &	0.9 &	0.9 &	0.938 &	0.916 &	0.911 &	0.894 &	0.941 &	0.904 &	0.886 &	0.932 &	0.9 &	-- &	-- &	-- &	0.909	\\	
&	&	MoLP &	0.954 &	0.945 &	0.936 &	0.954 &	0.958 &	0.956 &	0.929 &	0.968 &	0.956 &	0.906 &	0.976 &	-- &	-- &	-- &	0.951	\\	
&	0.95 &	oracle &	0.95 &	0.952 &	0.981 &	0.964 &	0.951 &	0.952 &	0.98 &	0.949 &	0.934 &	0.969 &	0.946 &	-- &	-- &	-- &	0.949	\\	
&	&	MoLP &	0.98 &	0.973 &	0.97 &	0.975 &	0.98 &	0.984 &	0.963 &	0.989 &	0.975 &	0.951 &	0.992 &	-- &	-- &	-- &	0.963	\\	
\cmidrule(lr){2-3} \cmidrule(lr){4-17} \cmidrule(lr){18-18}
&	— &	detection &	1 &	0.997 &	0.903 &	0.999 &	1 &	1 &	0.734 &	1 &	0.867 &	0.235 &	1 &	-- &	-- &	-- &	0.164	\\	
\cmidrule(lr){1-3} \cmidrule(lr){4-17} \cmidrule(lr){18-18}
{\tt fms} &	0.8 &	oracle &	0.894 &	0.817 &	0.833 &	0.832 &	0.81 &	0.814 &	-- &	-- &	-- &	-- &	-- &	-- &	-- &	-- &	0.851	\\	
&	&	MoLP &	0.82 &	0.885 &	0.914 &	0.852 &	0.873 &	0.885 &	-- &	-- &	-- &	-- &	-- &	-- &	-- &	-- &	0.831	\\	
&	0.9 &	oracle &	0.953 &	0.899 &	0.91 &	0.932 &	0.903 &	0.917 &	-- &	-- &	-- &	-- &	-- &	-- &	-- &	-- &	0.932	\\	
&	&	MoLP &	0.885 &	0.947 &	0.964 &	0.913 &	0.944 &	0.947 &	-- &	-- &	-- &	-- &	-- &	-- &	-- &	-- &	0.876	\\	
&	0.95 &	oracle &	0.979 &	0.948 &	0.95 &	0.98 &	0.954 &	0.962 &	-- &	-- &	-- &	-- &	-- &	-- &	-- &	-- &	0.961	\\	
&	&	MoLP &	0.905 &	0.978 &	0.984 &	0.947 &	0.974 &	0.968 &	-- &	-- &	-- &	-- &	-- &	-- &	-- &	-- &	0.896	\\	
\cmidrule(lr){2-3} \cmidrule(lr){4-17} \cmidrule(lr){18-18}
&	— &	detection &	0.384 &	1 &	1 &	0.965 &	0.979 &	0.987 &	-- &	-- &	-- &	-- &	-- &	-- &	-- &	-- &	0.368	\\	
\cmidrule(lr){1-3} \cmidrule(lr){4-17} \cmidrule(lr){18-18}
{\tt mix} &	0.8 &	oracle &	0.794 &	0.797 &	0.796 &	0.795 &	0.786 &	0.817 &	0.791 &	0.814 &	0.8 &	0.782 &	0.796 &	0.794 &	0.846 &	-- &	0.841	\\	
&	&	MoLP &	0.873 &	0.891 &	0.907 &	0.896 &	0.894 &	0.901 &	0.88 &	0.882 &	0.878 &	0.888 &	0.887 &	0.813 &	0.677 &	-- &	0.6	\\	
&	0.9 &	oracle &	0.89 &	0.898 &	0.9 &	0.9 &	0.898 &	0.91 &	0.899 &	0.908 &	0.902 &	0.895 &	0.906 &	0.91 &	0.938 &	-- &	0.928	\\	
&	&	MoLP &	0.941 &	0.954 &	0.965 &	0.954 &	0.962 &	0.958 &	0.943 &	0.947 &	0.945 &	0.95 &	0.931 &	0.894 &	0.774 &	-- &	0.7	\\	
&	0.95 &	oracle &	0.954 &	0.955 &	0.952 &	0.95 &	0.958 &	0.956 &	0.95 &	0.95 &	0.961 &	0.954 &	0.962 &	0.962 &	0.974 &	-- &	0.97	\\	
&	&	MoLP &	0.972 &	0.983 &	0.986 &	0.978 &	0.982 &	0.98 &	0.971 &	0.973 &	0.968 &	0.968 &	0.955 &	0.911 &	0.774 &	-- &	0.7	\\	\cmidrule(lr){2-3} \cmidrule(lr){4-17} \cmidrule(lr){18-18}
&	— &	detection &	0.998 &	1 &	1 &	1 &	1 &	1 &	0.999 &	0.982 &	0.899 &	0.64 &	0.31 &	0.062 &	0.016 &	-- &	0.005	\\	
\cmidrule(lr){1-3} \cmidrule(lr){4-17} \cmidrule(lr){18-18}
{\tt teeth10} &	0.8 &	oracle &	0.812 &	0.804 &	0.812 &	0.813 &	0.808 &	0.803 &	0.814 &	0.81 &	0.818 &	0.798 &	0.818 &	0.808 &	0.804 &	-- &	0.884	\\	
&	&	MoLP &	0.873 &	0.842 &	0.858 &	0.865 &	0.863 &	0.857 &	0.865 &	0.861 &	0.861 &	0.849 &	0.853 &	0.843 &	0.858 &	-- &	0.836	\\	
&	0.9 &	oracle &	0.916 &	0.922 &	0.924 &	0.922 &	0.916 &	0.932 &	0.922 &	0.931 &	0.926 &	0.918 &	0.926 &	0.928 &	0.903 &	-- &	0.954	\\	
&	&	MoLP &	0.942 &	0.925 &	0.927 &	0.928 &	0.923 &	0.931 &	0.924 &	0.934 &	0.93 &	0.928 &	0.928 &	0.92 &	0.927 &	-- &	0.92	\\	
&	0.95 &	oracle &	0.962 &	0.973 &	0.974 &	0.972 &	0.974 &	0.974 &	0.974 &	0.972 &	0.968 &	0.971 &	0.975 &	0.972 &	0.96 &	-- &	0.981	\\	
&	&	MoLP &	0.971 &	0.955 &	0.963 &	0.962 &	0.958 &	0.958 &	0.962 &	0.963 &	0.963 &	0.962 &	0.961 &	0.955 &	0.953 &	-- &	0.953	\\	\cmidrule(lr){2-3} \cmidrule(lr){4-17} \cmidrule(lr){18-18}
&	— &	detection &	0.906 &	0.892 &	0.896 &	0.89 &	0.892 &	0.888 &	0.888 &	0.894 &	0.896 &	0.905 &	0.898 &	0.896 &	0.9 &	-- &	0.454	\\
\cmidrule(lr){1-3} \cmidrule(lr){4-17} \cmidrule(lr){18-18}
{\tt stairs10} &	0.8 &	oracle &	0.824 &	0.813 &	0.823 &	0.814 &	0.82 &	0.839 &	0.82 &	0.83 &	0.829 &	0.836 &	0.824 &	0.834 &	0.838 &	0.804 &	0.925	\\	
&	&	MoLP &	0.901 &	0.898 &	0.9 &	0.901 &	0.909 &	0.911 &	0.917 &	0.911 &	0.914 &	0.917 &	0.921 &	0.913 &	0.912 &	0.9 &	0.977	\\	
&	0.9 &	oracle &	0.934 &	0.921 &	0.92 &	0.919 &	0.923 &	0.928 &	0.931 &	0.926 &	0.932 &	0.928 &	0.923 &	0.93 &	0.934 &	0.908 &	0.97	\\	
&	&	MoLP &	0.965 &	0.961 &	0.956 &	0.959 &	0.964 &	0.966 &	0.972 &	0.964 &	0.965 &	0.971 &	0.969 &	0.967 &	0.971 &	0.964 &	0.987	\\	
&	0.95 &	oracle &	0.976 &	0.971 &	0.971 &	0.962 &	0.968 &	0.974 &	0.968 &	0.974 &	0.97 &	0.972 &	0.975 &	0.976 &	0.976 &	0.963 &	0.982	\\	
&	&	MoLP &	0.987 &	0.987 &	0.979 &	0.978 &	0.984 &	0.987 &	0.987 &	0.981 &	0.984 &	0.986 &	0.986 &	0.982 &	0.986 &	0.981 &	0.993	\\	\cmidrule(lr){2-3} \cmidrule(lr){4-17} \cmidrule(lr){18-18}
&	— &	detection &	0.998 &	0.997 &	0.996 &	0.993 &	0.998 &	0.994 &	0.995 &	0.998 &	0.997 &	0.998 &	0.998 &	0.995 &	0.998 &	0.996 &	0.952	\\	
\bottomrule
\end{tabular}}
\end{table}

\clearpage

\subsection{Comparison with SMUCE}
\label{sec:smuce}

As noted in Section~\ref{sec:sim:comp},
SMUCE \citep{frick2014} returns confidence bands around change points at a prescribed level $\alpha$
which are readily comparable with the proposed uniform bootstrap CIs.

Several authors noted that the smaller $\alpha$ is set,
the constraint imposed by SMUCE on the estimated residuals becomes more lenient
and thus it tends to under-estimate the number of change points
\citep{chen2014, fryzlewicz2020}.
This is demonstrated in Table~\ref{table:smuce}
comparing MoLP (\cite{cho2019two}, described in Section~\ref{sec:sim:lp})
and SMUCE with varying $\alpha \in \{0.1, 0.2, 0.45\}$ 
in terms of their detection accuracy.
Here, the latter performs poorly in correctly detecting all $q_n$ change points
compared to the former and in fact,
with the exception of the {\tt fms} test signal, attains this goal 
on far less than $10\%$ of realisations 
even when $\alpha$ is set as generously as $\alpha = 0.45$.
\begin{table}[htb]
\caption{Proportion of the realisations (out of $2000$)
where exactly $q_n$ change points estimators
correctly detecting the change points are returned by MoLP \citep{cho2019two} and 
SMUCE \citep{frick2014} applied with varying $\alpha \in \{0.1, 0.2, 0.45\}$ 
for the five test signals with $\vartheta = 1$.
We also provide in brackets the proportion of realisations
where {\it both} MoLP and SMUCE at prescribed $\alpha$
detect all $q_n$ change points.}
\label{table:smuce}
\centering
{\small
\begin{tabular}{l  c cc c}
\toprule
test signals & MoLP & SMUCE($0.45$) &  SMUCE($0.2$) & SMUCE($0.1$)  \\ 
\cmidrule(lr){1-1} \cmidrule(lr){2-2} \cmidrule(lr){3-5}
{\tt blocks} & 0.480 & 0.0065 & 0.0005 &  0 \\
& - & (0.0055) & (0) & (0) \\
{\tt fms} & 0.831 & 0.684 & 0.4215 & 0.248 \\
& - & (0.5855) & (0.3665) & (0.2235) \\
{\tt mix} & 0.2595 & 0.0105 & 0.001 &  0 \\
& - & (0.0085) & (0.001) & (0) \\
{\tt teeth10} & 0.698 & 0.0055 & 0.0005 &  0 \\
& - & (0.0055) & (0.0005) &  (0) \\
{\tt stairs10} & 0.9540 & 0.0785 & 0.011 & 0.002 \\
& - & (0.075) & (0.011) & (0.002) \\
\bottomrule
\end{tabular}}
\end{table}

As inferential statements made by SMUCE about the locations of the change points
is conditional on correctly estimating the number of change points,
the lack of detection accuracy of SMUCE 
makes fair comparison between our proposed bootstrap methodology and SMUCE difficult.
Figure~\ref{fig:smuce} compares the uniform bootstrap CIs constructed with MoLP estimators
and SMUCE CIs on the test signal {\tt fms}.
Here, both the coverage and the lengths of CIs are computed for each given confidence level
only using the realisations where both MoLP and SMUCE correctly detect
the all $q_n$ change points.
In doing so, both uniform bootstrap CIs and SMUCE CIs show conservative coverage,
while the latter are wider than the former at any given confidence level.

\begin{figure}[htbp]
\centering
\includegraphics[width=.9\textwidth]{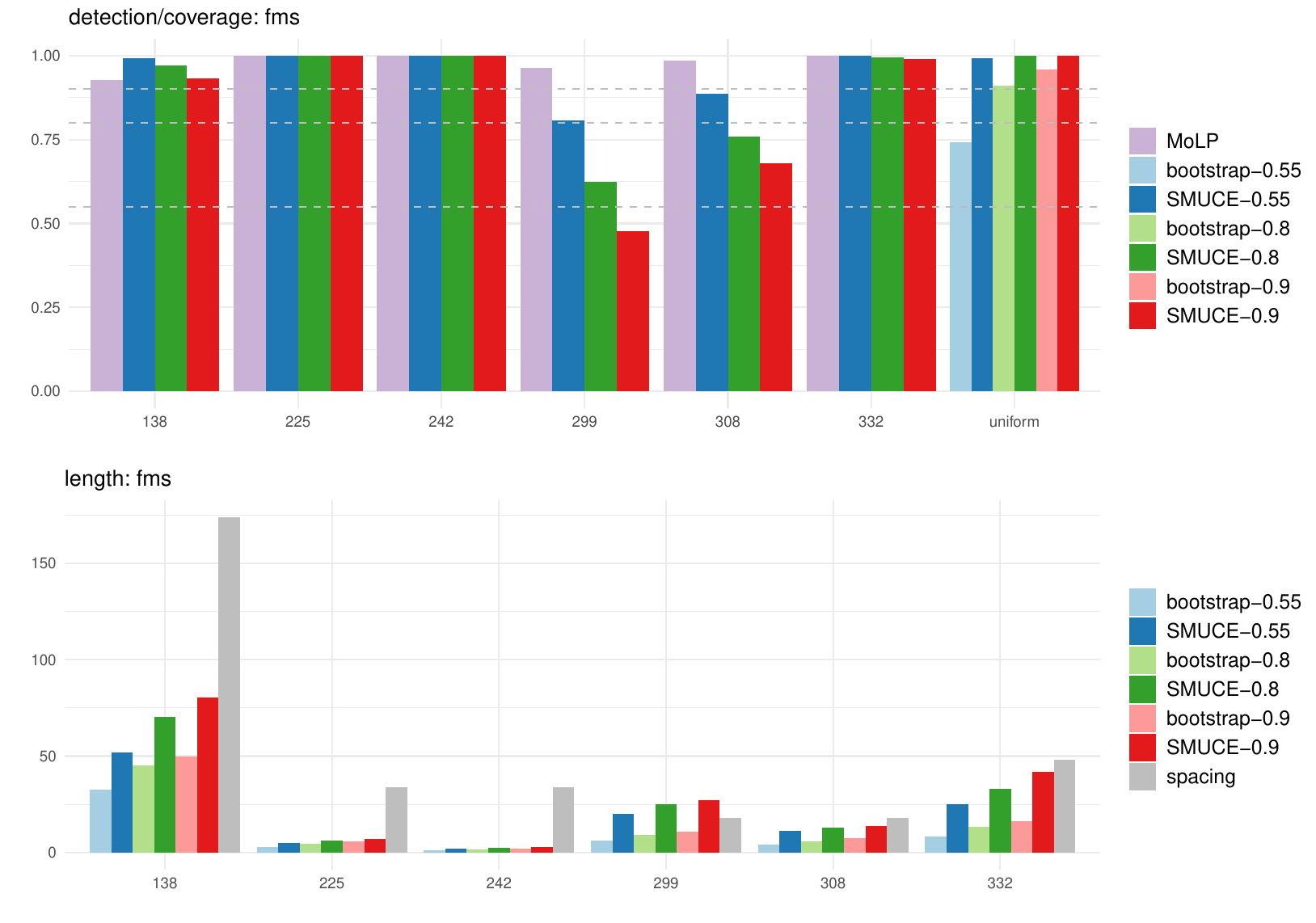}
\caption{Comparison between Bootstrap CIs constructed with model selection
and SMUCE:
{\tt fms} with $\vartheta = 1$.
Top: Proportion of correctly detecting each $\cp_j$ by MoLP and SMUCE
(their locations given as the $x$-axis labels) and the coverage of 
the uniform bootstrap CIs constructed with the estimators from MoLP,
and SMUCE CIs at varying confidence level $1 - \alpha \in \{0.55, 0.8, 0.9\}$.
Bottom: lengths of uniform bootstrap CIs and SMUCE CIs.
We also report $2\delta_j$, twice the minimum distance to adjacent change points,
for each $\cp_j, \, j = 1, \ldots, q_n$.}
\label{fig:smuce}
\end{figure}

\end{document}